\documentclass[12pt]{article}

\AtBeginDocument{%
   \setlength\abovedisplayskip{0.05in}
   \setlength\belowdisplayskip{0.08in}}

\usepackage[onehalfspacing]{setspace}

\makeatletter
\renewcommand\@footnotetext[1]{%
  \insert\footins{%
    \reset@font\footnotesize
    \interlinepenalty\interfootnotelinepenalty
    \splittopskip\footnotesep
    \splitmaxdepth \dp\strutbox \floatingpenalty \@MM
    \hsize\columnwidth \@parboxrestore
    {\setstretch{1.0}\protect\@makefntext{%
      \rule{\z@}{\footnotesep}\ignorespaces#1}}}}
\makeatother

\usepackage{geometry}
\usepackage{enumitem}
\usepackage{lineno}
\usepackage{booktabs}
\usepackage{amssymb}
\usepackage{amsmath}
\usepackage{amsthm}
\usepackage{caption}
\usepackage{epsfig}
\usepackage{graphicx}
\usepackage{graphics}
\usepackage{float}
\usepackage{subfigure}
\usepackage{multirow}
\usepackage{xcolor}
\usepackage{lineno}
\usepackage{fullpage}
\usepackage[normalem]{ulem} 
\usepackage{makeidx}
\usepackage{xspace}
\usepackage{wrapfig}
\usepackage{lscape}
\usepackage{mathrsfs}
\usepackage{adjustbox}
\usepackage{multirow}
\usepackage{varwidth}
\usepackage{rotating}
\usepackage{lscape}
\usepackage[x11names,dvipsnames]{xcolor}
\usepackage{tikz}
\usetikzlibrary{automata, positioning}
\usepackage[bottom]{footmisc}
\usepackage[colorlinks,citecolor=darkblue,urlcolor=darkblue,linkcolor=darkblue]{hyperref} 
\graphicspath{{C:/Users/flori/Desktop/folders/minimumWagePaper/paper/}}
\makeindex
\newtheorem{theorem}{Theorem}

\newtheorem{corollary}{Corollary}

\newtheorem{lemma}{Lemma}

\newtheorem{proposition}{Proposition}
\newtheorem{definition}{Definition}

\newtheorem{obs}{Observation}

\definecolor{purple}{rgb}{0.6, 0.4, 0.8}
\definecolor{darkred}{rgb}{1, 0.1, 0.3}
\definecolor{darkblue}{rgb}{0.0, 0.0, 0.55}
\definecolor{darkgreen}{rgb}{0,0.6,0.5}
\definecolor{forestgreen}{rgb}{0.0, 0.46, 0.37}
\definecolor{bittersweet}{rgb}{1.0, 0.44, 0.37}
\definecolor{navy}{rgb}{0.0, 0.0, 0.55}
\definecolor{brown}{rgb}{0.53, 0.18, 0.09}
\definecolor{Green}{rgb}{0.0, 0.47, 0.44}

\newcommand {\mm}[1] {\ifmmode{#1}\else{\mbox{$#1$}}\fi}

\newcommand\E{\mathbb{E}}

\newcommand\R{\mathbb{R}}

\usepackage{dsfont}
\newcommand{\m}[1]{\ensuremath{\boldsymbol {#1}}}

\newcommand{\DKL}{D_{\mathrm{KL}}}

\usepackage{sectsty}
\sectionfont{\centering\Large}
\subsectionfont{\large}
\usepackage{fancyhdr}
\usepackage{natbib}

\usepackage{pst-node, pst-plot, pst-poly}
\usepackage{rotating}

\usepackage{tikz}
\usepackage{pstricks}
\usepackage{tikz-cd}
\usepackage{siunitx}
\usepackage{pgfplots}
\usepgfplotslibrary{fillbetween}
\usepackage{algorithm}
\usepackage{algorithmic}
\usepackage{bbm}
\usepackage{csquotes}

\usepackage{scalerel}
\usetikzlibrary{shapes,arrows,trees,calc}
\pgfplotsset{compat=1.12}

\newcommand{\Y}{\mathcal{Y}}
\newcommand{\Acal}{\mathcal{A}}
\newcommand{\DeltaY}{\Delta(\Y)}
\newcommand{\co}{\operatorname{co}}

\newcommand{\Var}{\operatorname{Var}}

\newcommand{\argmax}{\operatorname*{arg\,max}}

\usetikzlibrary{trees}
\allowdisplaybreaks 
\begin{document}

\title{ \vspace{-3.6em}
\mbox{\small\bf\MakeUppercase{Contracting under Misspecification}}\footnote{I am very grateful to Drew Fudenberg and Stephen Morris for their support and mentorship.
}
\vspace{-0.5em}
}
\author{\small\MakeUppercase{Florian Mudekereza}\footnote{Department of Economics, MIT, \href{mailto:florianm@mit.edu}{\texttt{\footnotesize florianm@mit.edu}}.}}
\date{}

\maketitle
\thispagestyle{empty}
\setcounter{page}{0}
\vspace{-0.41in}

\begin{abstract}
This paper studies agency problems when both parties worry that the model linking action to output is misspecified. With observable actions, an optimal contract is linear in output, so performance pay arises solely to share exposure to misspecification, the slope reflects the parties' relative robustness concerns, and its allocation is Pareto efficient. With hidden actions, this sharing rule survives and incentives add a nonlinear correction. Misspecification concerns can  polarize effort by making intermediate actions impossible to implement. Moreover, ambiguity across competing models has asymmetric effects: uncertainty about desired actions raises the principal's payoff, whereas uncertainty about deviations can lower it.
\end{abstract}

\section{Introduction}

Consider a principal and an agent who agree on a probabilistic description of how the agent’s actions influence output realizations, but neither is willing to regard this model as the true data-generating process. Both parties recognize that the model is an approximation and want the contract to perform well under \textit{misspecification}. What should such a contract look like? When actions are observable, we find that an optimal contract is linear in output (Theorem \ref{thm:firstbest-linear}); it is typically unique when both parties exhibit a mild form of misspecification aversion (Corollary \ref{cor:firstbest-unique}). 
Performance pay therefore arises even without incentive problems: its only role is to share exposure to misspecification in the action-output relationship. The optimal output share is determined by the parties' relative concerns, so the contract allocates more exposure to the party that is less concerned about misspecification. Notably, this allocation is Pareto efficient (Corollary \ref{lem:FB-pareto}). 

We capture misspecification concerns using \citeauthor{mmr06}'s (\citeyear{mmr06}) variational preferences; \citeauthor{hansen01}'s (\citeyear{hansen01}) multiplier preference is a popular special case used in agency problems \citep[e.g.,][]{miao16,wu18,miao21}. Each party is risk neutral and evaluates a payoff by considering alternative output distributions while penalizing them with a statistical ambiguity index. The parties may differ in how strongly they are concerned about misspecification, but they use the same ambiguity index up to a positive scale. In many cases, this common statistical structure is precisely what makes a given linear sharing rule Pareto efficient across all payoffs (Proposition \ref{prop:different-divergences}). The mechanism behind Theorem \ref{thm:firstbest-linear} can be understood by asking how each party values a small reallocation of compensation across output realizations. A dollar paid to the agent in a realization is weighted by one of his effective models or ``marginal robust probabilities,'' while the same dollar costs the principal according to one of hers. If the parties have no common effective model, compensation can be shifted across outputs in a way that improves one party's valuation without reducing the other's, so first-best optimality requires their effective-model sets to intersect (Corollary \ref{cor:effective-models}). A linear contract achieves this alignment because it divides every additional unit of output in the same proportion; in fact, it forces the parties' effective-model sets to coincide exactly, so every marginal robust probability that supports one party's valuation also supports the other's. 

\par Our results have implications for the interpretation of performance pay. In classical agency problems where both parties have expected utility (EU) preferences, performance pay is associated with hidden actions and the informativeness of output about effort \citep[see,][]{holmstrom79}. When actions are observable, there is no incentive constraint requiring compensation to depend on output. With bilateral risk neutrality and unrestricted transfers, only the agent's expected payment matters for  surplus: a constant transfer is optimal, but it is not uniquely so, because any output-contingent contract delivering the same expected payment gives the same expected surplus \citep[][Section 4]{laffontmartimort02}. Misspecification concerns remove this indeterminacy by making how output is divided across the parties payoff relevant and select a contract that assigns the agent a constant output share (Corollary \ref{cor:firstbest-unique}).
This therefore suggests that observing performance pay in practice need not be evidence of incentive problems. The slope of performance pay may instead contain information about how parties allocate exposure to  misspecification. Specifically, Theorem \ref{thm:firstbest-linear} predicts greater performance sensitivity when the principal is relatively more concerned about misspecification and lower sensitivity when the agent is.

\par A real-world analogue of our framework is catastrophe reinsurance for rare natural disasters, such as hurricanes. Primary insurers purchase reinsurance to transfer part of the losses from disasters to reinsurers, so the contract divides a common uncertain claim. Both parties rely on ``catastrophe models''---stochastic simulation models that generate hypothetical disasters and map them into portfolio losses---to price catastrophe exposure. While these models are useful, the parties understand that they are approximations: \citet[p. 205]{gray2021hazardous} reports that after 2004--2005, the models ``underperformed [...] and did poorly at anticipating specific kinds of property damage.'' Thus, these reinsurance contracts can be viewed as allocating exposure to misspecified catastrophe models. In this industry, a linear contract is known as ``quota-share'' reinsurance, where the reinsurer bears a constant fraction of  losses. \citet{rein08} find that a quota share becomes optimal when tail risk is sufficiently important relative to the cost of reinsurance, while \citet[][]{froot2008pricing} argue that it is a useful approximation in empirical analyses. 

\par Optimal misspecification sharing changes the value of the agency relationship. By splitting a common misspecified output claim, the contract weakly reduces the cost of misspecification relative to leaving the entire claim with either party (Proposition \ref{prop:robustness-pooling}). When output is genuinely stochastic, this gain can be strict and can change which observable action the principal prefers. Misspecification sharing therefore affects not only compensation, but also which technologies are more attractive to operate.

\par Theorem \ref{thm:firstbest-linear} provides a static representation of an \textit{uncertainty-transfer} mechanism identified in dynamic robust contracting. Under unilateral robustness, where a robust-multiplier principal contracts with an EU agent, our linear contract transfers all output exposure to the agent. This comparative static isolates the same mechanism that operates in \citet{miao16}, \citet{wu18}, and \citet{miao21}: they find that a robust principal benefits from shifting cash-flow exposure toward an EU agent, so that the agent can insure her by making her payoff flows constant; this is achieved by having the agent bear more uncertainty from cash flows than incentive provision alone requires. 

\par The first-best approach provides a benchmark for studying moral hazard. When the agent's actions become hidden, Theorem \ref{thm:MH-general} shows that the common-effective-model condition from first best is distorted only by the marginal incentive wedges associated with active deviations. Thus, moral hazard changes how the parties' marginal valuations are aligned rather than eliminating the first-best sharing motive. Under multiplier preferences,  there is a unique optimal contract in closed form, which retains the first-best linear component and adds a nonlinear likelihood-ratio correction (Corollary \ref{thm:MH-contract}). In settings with binary output, the optimal bonus must be large enough both to provide incentives and to satisfy the first-best sharing requirement. More generally, when the desired action dominates deviations in the monotone-likelihood-ratio sense, incentive provision can steepen the contract but cannot undo its first-best performance sensitivity (Corollary \ref{prop:MH-slope-floor}).

Misspecification concerns also change which actions can be implemented. With a continuum of actions, variational preferences impose a curvature restriction on any interior implemented action; under multiplier preferences, this restriction becomes sharper because the curvature index is constant. The agent's multiplier value of performance pay is convex in the success probability, unlike the EU benchmark where compensation value is affine in that probability. The resulting curvature can make an interior action locally unattractive even when it satisfies the agent's first-order condition. When the convexity induced by robust compensation dominates effort-cost curvature throughout the interior, every interior action becomes nonimplementable. Thus, an arbitrarily small change in bonus can move the agent discontinuously from lowest to highest action (Theorem \ref{thm:continuum-polarization}).

We extend our framework to environments where an action is associated with multiple competing models  to study the effects of model uncertainty and aggregation on incentive provision. This extension captures settings where an organization has access to a set of plausible models or ``model library,'' which may represent expert predictions or competing scientific theories. The multitude of models introduces a new friction: \textit{model ambiguity}. Thus, both parties must trade off model misspecification and model ambiguity. 
To formalize this friction, we follow \citeauthor{kellner17}'s (\citeyear{kellner15,kellner17}) approach, which introduces a second-order belief that assigns weights to candidate models in a Bayesian way. Within each candidate model, we allow a  reference-indexed variational criterion and average the corresponding robust valuations. The resulting criterion is the linear-aggregator special case of \citeauthor{hansenmiss25}'s (\citeyear{hansenmiss25}) smooth Bayesian criterion. The distribution of models matters even when its mean model is fixed. Greater dispersion in the library of the desired action raises both the agent's valuation of that action and the principal's valuation of residual output, so the set of implementing contracts expands and the principal's optimized payoff weakly rises (Theorem \ref{thm:library}). Dispersion in the library of a deviation instead makes that deviation more attractive to the agent, shrinks the implementation set, and weakly lowers the principal's payoff. We then allow dispersion to be endogenous by allowing the principal to optimize it according to \citeauthor{linear20}’s (\citeyear{linear20}) confidence-region problem. We show that their worst-case principal trades greater statistical confidence against lower contractual performance, whereas here these two forces move in the same direction, so our principal chooses the widest admissible region (Proposition \ref{prop:endogenous-dispersion}).

We find that the first-best linear contract relies on the particular structure of variational preferences for two reasons. First,  variational preferences represent a particularly \textit{simple} attitude toward model uncertainty: alternative models are evaluated using statistical penalties that are independent of payoffs. Second, this preference class also satisfies translation invariance: adding the same sure amount to every realization changes a party's valuation one-for-one without changing how uncertainty is assessed. These properties help separate compensation into a level that satisfies participation and a constant output share that allocates misspecification exposure. In contrast, other non-EU preferences, such as \citeauthor{smooth05}'s (\citeyear{smooth05}) smooth ambiguity preferences, which need not be translation invariant, allow ambiguity attitudes to vary with payoffs, so the level of compensation can itself change how uncertainty is evaluated; Appendix \ref{app:general-criteria}  shows that this dependence can give rise to nonlinear contracts. This distinction has the following behavioral interpretation. \citet{hansenmiss25} interpret variational distortions around reference models as reflecting aversion to misspecification, whereas smooth ambiguity preferences reflect a \textit{neutral} attitude toward misspecification. Thus, the nonlinear contracts that can arise outside the variational class help identify the behavioral restrictions behind linearity: linear contracts typically emerge as an efficient rule to hedge against misspecification in the action-output relationship, rather than against general ambiguity aversion.

\par The simplicity of variational preferences is closely related to contracting. \citet{contract25} shows, axiomatically, that the canonical moral-hazard problem admits a parsimonious representation in which the agent chooses an output distribution according to the cost of generating it, and that the inverse of the preference over contracts is variational. Thus, the same cost over candidate models that characterizes variational model uncertainty arises naturally from moral hazard. This connects two forms of simplicity emphasized here: variational preferences provide a simple representation of model uncertainty, while linear contracts provide a simple rule for sharing exposure to that uncertainty.

\subsection{Related work}

A literature explains linear contracts through robustness to incomplete knowledge of the contracting environment. In \citet{linear15}, the principal does not know all of the actions available to the agent and chooses a contract to perform well against the worst technology consistent with the actions she knows. Subsequent work has extended this approach in various ways; for example, \citet{dai22} analyze robustness of team incentives, and \citet{linear22} study when the possible responses to a contract are sufficiently rich and responsive for linear contracts to maximize the principal's worst-case guarantee. \citet{linear20} instead place uncertainty over the parameters of the effort-contingent output distribution and require the contract to perform well, and to preserve incentives, across a specified uncertainty set. As emphasized by \citeauthor{linear19}'s (\citeyear{linear19}) survey, these approaches formalize robustness through a class of possible environments and protect the principal against its least favorable member. Our notion of uncertainty is different. Both parties share a variational ambiguity index for output conditional on an action, and alternative models are penalized according to the same statistical criterion rather than treated as equally plausible members of a worst-case set. The results can therefore be viewed as complementary explanations for linear contracting. In the robust-contract literature, linearity protects the principal against incomplete knowledge of what the agent can do or of distributional parameters. By contrast, linearity arises here even without incentive problems: its role is to share exposure to misspecification in the action-output relationship. Section \ref{subsec:carroll} develops these distinctions in more detail.

A conceptually closer literature introduces model misspecification into principal-agent problems using multiplier preferences. \citet{miao16} pioneered the study of continuous-time financial-contracting problems where the principal distrusts the project's cash-flow model while the agent trusts it. \citet{miao21} embed the same unilateral multiplier robustness in a richer dynamic environment with investment, financing, and liquidation decisions. Similarly, \citet{wu18} introduce the same robustness in \citeauthor{hm87}'s (\citeyear{hm87}) framework to study robust long-term incentives and focus on relative performance evaluation. All these papers report the ``uncertainty-transfer'' mechanism introduced earlier. However, their optimal first-best contracts are not linear because this mechanism interacts with other frictions such as continuation-value dynamics. After stripping away these additional forces, our first-best analysis reveals that uncertainty transfer is part of a more general variational sharing problem and determines how output exposure should be allocated between the two parties.  A bit further afield,  \citet{seller16} study a different form of model misspecification in a screening problem, where the principal's model is an approximation of the agent's preference.

\par Our first-best results connect to the literature on efficient sharing under uncertainty. \citet{wilson68} defines a syndicate as a group that makes a collective decision under uncertainty and jointly shares the resulting payoff. He shows that efficient risk sharing assigns each member a marginal share in proportion to that member's risk tolerance. We obtain the same sharing formula for misspecification sharing, despite studying risk-neutral variational parties rather than risk-averse EU agents. Since risk sharing is absent here, the coincidence shows that Wilson's sharing rule is not specific to curvature in utility: the same marginal allocation principle governs how an aggregate uncertain payoff is divided when the relevant friction is model misspecification rather than risk aversion. \citet{rss08} characterize Pareto-efficient allocations under ambiguity by the existence of a common subjective supporting belief; for risk-neutral variational preferences these beliefs are the effective models, so Corollary \ref{cor:effective-models} implements their Pareto-efficiency condition in agency problems. \citet{strzaleckiwerner11} sharpen this welfare perspective for a broad class of ambiguity preferences when aggregate uncertainty is present by introducing conditional beliefs: agreement on conditional beliefs yields measurability of efficient allocations with respect to aggregate endowment. \citet{rigottishannon12} study a complementary sharing problem in complete markets with variational preferences by analyzing risk and ambiguity sharing, and show that equilibria are generically determinate.

Our analysis relates to the broader literature on ambiguity in agency problems. \citet{ghir94} studies  contracts when both parties have Choquet preferences \citep{schmeidler89}, which allow nonadditive beliefs to alter the classical incentive problem. \citet{amb11} instead give the agent an incomplete preference, with an EU principal, and study how ambiguity affects the information that should enter contracts. \citet{kellner17} uses smooth-ambiguity preferences and shows that an ambiguity-averse agent can change implementation costs and qualitative properties of the optimal incentive scheme, whereas \citet{kellner15} applies this framework to tournaments. \citet{amb24} analyze how an EU principal can deliberately introduce ambiguity in contract design to exploit the agent's ambiguity aversion. More broadly, these papers primarily study how ambiguity-sensitive preferences interact with moral hazard. Our focus is different: we show that misspecification concerns  create a contractual sharing problem even before incentive frictions are introduced, and the distinction between variational and other non-EU preferences determines whether the first-best problem can select a linear contract.

We also contribute to the ongoing effort to understand why linear contracts arise often in both economic theory and practice. Classical explanations include \citeauthor{hm87}'s (\citeyear{hm87}) dynamic aggregation, as well as the mechanisms in \citet{hur78}, \citet{diamond98}, and \citet{chassang13}. \citet{linear19} emphasizes a broader lesson from the robust-contract literature in mechanism design: simple contracts can become attractive precisely when the designer is unwilling to tailor incentives to a finely specified environment. The repeated appearance of linear rules in recent robust-contracting frameworks reinforces this connection between simplicity and robustness. Our results add misspecification sharing as a different reason for the same contractual form. 

\par\noindent --- \textit{Outline}. Section \ref{sec:model} introduces the framework. Section \ref{sec:firstbest} studies the first-best problem, and Section \ref{sec:moralhazard} introduces moral hazard. Section \ref{sec:ext} analyzes effort polarization and model ambiguity, and Section \ref{sec:conclusion} concludes the paper.  Appendix \ref{app:proofs} contains all the proofs.

\section{Setup}\label{sec:model}

There is a finite set of actions $\Acal$ and a finite set of outputs $\Y\subset\R$. The agent chooses an action $a\in\Acal$. Output is the principal's gross payoff, so she receives $y\in \Y$ before paying the agent. A contract is an unrestricted transfer schedule $w:\Y\to\R$. The agent's cost from choosing action $a$ is $c(a)\in\R$, and his outside option is $U_0\in\R$.

Both parties are risk neutral but are concerned that the probabilistic description of the action-output relationship is misspecified. We represent this misspecification concern using \citeauthor{mmr06}'s (\citeyear{mmr06}) variational preferences. Formally, for each action $a$, let $\Gamma^a:\DeltaY\to[0,\infty]$ be a grounded, convex, and lower-semicontinuous ambiguity index. The agent's concern is indexed by $\lambda_A>0$, and the principal's concern is indexed by $\lambda_P>0$. For any constant $\lambda>0$ and a payoff vector $x:\Y\to\R$, the variational value is
$$
\mathcal V_\lambda^a(x):=\min_{p\in\DeltaY}\Big\{\E_p[x(y)]+\frac{1}{\lambda}\Gamma^a(p)\Big\}.
$$
This assumes both parties have proportional ambiguity indices $\Gamma^a/\lambda_A$ and $\Gamma^a/\lambda_P$,\footnote{Proposition \ref{prop:different-divergences} characterizes what happens to the first-best contract when proportionality is relaxed.} which means that they agree on the relative statistical cost of every alternative model but may differ in how strongly they are concerned about misspecification. The EU benchmark arises in the zero-concern limit $\lambda\to0$ when $\Gamma^a$ has a unique minimizer.

For any $\lambda>0$, let $\Pi_\lambda^a(x)$ denote the nonempty set of models attaining the minimum in the criterion $\mathcal V_\lambda^a(x)$. We call these minimizers the \textit{effective models}. For a contract $w$, write $\Pi_A^a(w):=\Pi_{\lambda_A}^a(w)$ and $\Pi_P^a(w):=\Pi_{\lambda_P}^a(y-w)$. Whenever party $i$'s effective-model set is a singleton, denote its element by $\pi_i^a(\cdot;w)$, for $i\in\{A,P\}$. These effective models are the marginal probability weights of the variational value $\mathcal V_\lambda^a$ (Appendix \ref{app:prelim}), so they will play a central role in characterizing the first-best  and moral-hazard problems.

 The most popular special case is \citeauthor{hansen01}'s (\citeyear{hansen01}) multiplier preference. In this specification, we associate each action $a$ with a full-support distribution $q^a\in\DeltaY$ that both parties use as their reference model. The ambiguity index is $\Gamma^a(p)=\DKL(p\|q^a)$, where $\DKL(\cdot\|\cdot)$ denotes the relative entropy or Kullback-Leibler divergence. We can write the multiplier criterion as $M(q,x;\lambda)=-\lambda^{-1}\log\E_q[e^{-\lambda x(y)}]$ \citep[][Proposition 1.4.2]{dupuis97}. This structure delivers a unique effective model for each party:
$$
\pi_A^a(y;w)=\frac{q^a(y)e^{-\lambda_A w(y)}}{\E_{q^a}[e^{-\lambda_A w(z)}]},\qquad
\pi_P^a(y;w)=\frac{q^a(y)e^{-\lambda_P[y-w(y)]}}{\E_{q^a}[e^{-\lambda_P[z-w(z)]}]}.
$$
Thus, multiplier effective models have closed-form expressions and reveal that robustness to misspecification operates as pessimistic distortions: when a party evaluates a payoff, lower-payoff outputs receive extra weight, but deviations from $q$ are penalized by relative entropy.  
The agent's payoff from action $a$ under contract $w$ is $V_A^a(w):=\mathcal V_{\lambda_A}^a(w)-c(a)$, and the principal's payoff from implementing action $a$ is $V_P^a(w):=\mathcal V_{\lambda_P}^a(y-w(y))$.

\section{First Best}\label{sec:firstbest}

The first-best problem fixes an observable action $a$, so contracts must deliver the agent's outside option without incentive constraints. The variational first-best problem becomes
$$
(\mathsf{FB}_a)\qquad \max_{w\in\R^{\Y}} V_P^a(w)\quad\text{subject to}\quad V_A^a(w)\ge U_0.
$$
In the classical EU benchmark with observable actions, bilateral risk neutrality, and unrestricted transfers, the principal's payoff depends on the contract only through its expected payment. Any contract satisfying participation with equality therefore gives the same expected surplus; a constant transfer is an optimal contract, but it is not uniquely selected \citep[see,][Sections 4.2.3--4.3]{laffontmartimort02}. Theorem \ref{thm:firstbest-linear} and Corollary \ref{cor:firstbest-unique} will show that introducing misspecification concerns can remove this indeterminacy because the  allocation of output itself affects both parties' robust valuations.

\begin{lemma}\label{lem:FB-participation}
Fix an action $a$. Every solution of $(\mathsf{FB}_a)$ satisfies $V_A^a(w)=U_0$.
\end{lemma}
Lemma \ref{lem:FB-participation} separates the level of compensation from its exposure to output. Translation invariance implies that a common shift in wages transfers utility one-for-one between parties, so the main question is how the stochastic output claim should be divided. Giving the agent too little output exposure leaves the principal disproportionately exposed to misspecification in residual output; giving him too much shifts that exposure to the wage claim. Theorem \ref{thm:firstbest-linear} identifies a division of output that balances these two valuations.

\begin{theorem}\label{thm:firstbest-linear}
Fix an action $a$. The first-best problem $(\mathsf{FB}_a)$ has an optimal solution
$$
 w^{FB}_a(y)=\alpha_a+\frac{\lambda_P}{\lambda_A+\lambda_P}\hspace{0.03in}y
$$
for every  $y\in\mathcal{Y}$, where the intercept  $\alpha_a$ is uniquely determined by binding participation.
\end{theorem}

Theorem \ref{thm:firstbest-linear} gives performance sensitivity a first-best interpretation. Each additional unit of output is divided in fixed proportions; the agent receives the share $\frac{\lambda_P}{\lambda_A+\lambda_P}$ and the principal retains $\frac{\lambda_A}{\lambda_A+\lambda_P}$. These shares depend only on the parties' relative misspecification concerns. The ambiguity index $\Gamma^a$, the action cost $c(a)$, and the outside option $U_0$ affect the level of compensation, but not this marginal division of output. A more misspecification-averse principal therefore transfers a larger fraction of output exposure to the agent, while a more misspecification-averse agent bears less of it.

\subsection{Welfare implication and risk sharing}
Linearity in Theorem \ref{thm:firstbest-linear} follows from matching the parties' effective models. To make their robust evaluations agree on which outputs are costly at the margin, wages must allocate the common output claim so that both sides can use the same probability weights. The next result formalizes this condition without requiring a unique effective model.

\begin{corollary}\label{cor:effective-models}
 Fix an action $a$. At every first-best optimum $w$, the parties have a common effective model:
$\Pi_A^a(w)\cap\Pi_P^a(w)\neq\varnothing$.
For the linear contract $w^{FB}_a$,  $\Pi_A^a(w^{FB}_a)=\Pi_P^a(w^{FB}_a)$.
\end{corollary}

Corollary \ref{cor:effective-models} has a welfare implication. \citet{rss08} define subjective beliefs at an act as normalized supporting probabilities to its upper contour set. Their Proposition 3 shows that, for variational preferences, these supporting probabilities are effective models weighted by marginal utility. Since both parties here are risk neutral, marginal utility is constant and the supporting beliefs coincide with the minimizing models in the variational representation. Their Proposition 7 then characterizes an interior Pareto-efficient allocation by the existence of a subjective belief common to all agents. Our first-best problem is a two-agent sharing problem for output $y$, so Corollary \ref{cor:effective-models} is the corresponding contracting analogue: an efficient division of output must admit a model that prices marginal transfers for both parties. If their effective-model sets were disjoint, the parties would assign incompatible marginal prices to transfers, so a small reallocation of compensation could make both strictly better off. Notably, the linear first-best contract delivers stronger alignment by requiring the two effective-model sets to coincide.

\begin{corollary}\label{lem:FB-pareto}
Fix an action $a$. Every solution of $(\mathsf{FB}_a)$ is Pareto efficient.
\end{corollary}

Our linear sharing formula also has a direct counterpart in \citet{wilson68}. Define the parties' \textit{misspecification tolerances} as $\rho_A:=1/\lambda_A$ and $\rho_P:=1/\lambda_P$, so that the agent's share $\frac{\lambda_P}{\lambda_A+\lambda_P}$ in Theorem \ref{thm:firstbest-linear} equals $\frac{\rho_A}{\rho_A+\rho_P}$. Wilson's Theorem 4 shows that a syndicate's risk tolerance is the sum of its members' risk tolerances, and his Theorem 5 shows that each member's marginal share of aggregate payoff equals that member's risk tolerance divided by the syndicate's total risk tolerance. In this sense, the sum $\rho_A+\rho_P$ and its reciprocal $\frac{1}{\rho_A+\rho_P}=:\lambda_{AP}$ can be interpreted, respectively, as the organization's misspecification tolerance and concern. Thus, the formula that efficiently shares risk in Wilson's syndicate also shares exposure to model misspecification in our two-party contracting problem.

The coincidence is economically informative because the source of marginal valuation is different. Wilson's members are strictly risk-averse EU maximizers, so marginal utility varies with the payoff they receive. Here, both parties are risk neutral, so conventional risk sharing is absent; marginal valuations instead come from the effective models selected by their variational criteria. Efficiency nevertheless has the same implication in both environments: marginal exposure is allocated in proportion to tolerance so that the parties' marginal valuations are aligned. Misspecification tolerance therefore plays for model uncertainty the same allocative role that risk tolerance plays for payoff uncertainty.
\subsection{Uniqueness and an example}\label{subsec:FB-uniqueness}

Theorem \ref{thm:firstbest-linear} is an existence result because arbitrary variational preferences can have kinks or flat regions that leave the sharing slope underidentified. A mild regularity condition restores uniqueness. For notation, let $\operatorname{ri}B$ denote the relative interior of a set $B.$

\begin{corollary}\label{cor:firstbest-unique}
Fix an action $a$. If there exists a probability distribution $p^*\in\operatorname{ri}\DeltaY$ that minimizes the pooled problem $p\mapsto\E_p[y]+\frac{1}{\lambda_{AP}}\Gamma^a(p)$ at which $\Gamma^a$ is differentiable, then the linear contract $w^{FB}_a$ in Theorem \ref{thm:firstbest-linear} is the unique solution of $(\mathsf{FB}_a)$.
\end{corollary}

The objective minimized by the effective model $p^*$ defines a pooled problem that evaluates output against the organization's misspecification tolerance $\frac{1}{\lambda_{AP}}=\frac{1}{\lambda_A}+\frac{1}{\lambda_P}$. Restricting this aggregate objective identifies the common effective model required for first-best optimality. The condition in Corollary  \ref{cor:firstbest-unique} therefore has a simple economic meaning: the model that prices output must have full support, and the statistical penalty must assign a unique marginal cost to small probability reallocations. 

\par Appendix \ref{app:divergences} shows that Corollary \ref{cor:firstbest-unique} holds for a large class of popular divergences. A notable example is the subclass of $\varrho$-divergences \citep[e.g.,][]{zhang06,kitamura1}. Proposition \ref{prop:divergence-classes} reveals that these divergences capture a \textit{mild} form of misspecification aversion---a behavioral axiom introduced in \citet[][Definition 7]{hansenmiss25}.\footnote{Proposition \ref{prop:divergence-classes} also shows that, under two regularity conditions, Corollary \ref{cor:firstbest-unique} holds for the subclass of $\phi$-divergences \citep[][eq. (6)]{hansenmiss25}. These divergences also capture mild misspecification.} This axiom says that although the parties are concerned that the reference model is misspecified, they still find it useful in the sense that it is given a privileged role when ranking payoffs; see Appendix \ref{app:divergences} for  details. Familiar examples of these divergences are the squared Hellinger distance and the relative entropy. Thus, multiplier preferences satisfy  Corollary \ref{cor:firstbest-unique}, so the linear contract $w^{FB}_a$ is uniquely optimal in this case. 
\par Figure \ref{fig:misspecification-sharing} retains the multiplier case to visualize Theorem
\ref{thm:firstbest-linear} and Corollary \ref{cor:effective-models}.  Fix an action $a$, let $\Y=\{1,2,3\}$,
$\lambda_A=\lambda_P=1$, $q^a=(1/3,1/3,1/3)$, $c(a)=0$, and
$U_0=M(q^a,y/2;1)
=-\log\frac{e^{-1/2}+e^{-1}+e^{-3/2}}{3}.
$
Binding participation implies $\alpha_a=0$. The horizontal axis measures the agent's misspecification exposure, $\lambda_Aw(y)=w(y)$, while the vertical axis measures the principal's, $\lambda_P[y-w(y)]=y-w(y)$. For every output $y\in\{1,2,3\}$, each gray downward-sloping line consists of all feasible divisions of $y$, defined by $w(y)+[y-w(y)]=y$; moving southeast  transfers exposure from  principal to agent.

\begin{figure}[hbt!]
\centering

\begin{tikzpicture}[scale=1.35, x=2.8cm,y=1.75cm]

    \draw[->,thick] (0,0) -- (3.85,0);
    \draw[->,thick] (0,0) -- (0,3.35);

    \node[align=center] at (1.9,-0.25)
        {Agent's exposure to misspecification};
    \node[rotate=90,align=center] at (-0.35,1.68)
        {Principal's exposure to misspecification};

    \draw[gray!55,thin] (0,1) -- (1,0);
    \draw[gray!55,thin] (0,2) -- (2,0);
    \draw[gray!55,thin] (0,3) -- (3,0);

    \node[gray!80,anchor=south west] at (1.02,0.02)
        {\small $y=1$};
    \node[gray!80,anchor=south west] at (2.02,0.02)
        {\small $y=2$};
    \node[gray!80,anchor=south west] at (3.02,0.02)
        {\small $y=3$};

    \draw[black,very thick] (0,0) -- (3.12,3.12)
        node[pos=0.82, sloped, above, black] {\small Effective-model alignment};

    \draw[orange!60!yellow,very thick,dashed]
        (0.5,0.5) -- (1,1) -- (1.5,1.5);
    \filldraw[orange!60!yellow] (0.5,0.5) circle (1.7pt);
    \filldraw[orange!60!yellow] (1,1) circle (1.7pt);
    \filldraw[orange!60!yellow] (1.5,1.5) circle (1.7pt);

    \draw[darkblue,very thick,densely dotted]
        (0.75,0.25) -- (1.5,0.5) -- (2.25,0.75);
    \filldraw[darkblue] (0.75,0.25) circle (1.5pt);
    \filldraw[darkblue] (1.5,0.5) circle (1.5pt);
    \filldraw[darkblue] (2.25,0.75) circle (1.5pt);

    \draw[brown,very thick,loosely dotted]
        (0.25,0.75) -- (0.5,1.5) -- (0.75,2.25);
    \filldraw[brown] (0.25,0.75) circle (1.5pt);
    \filldraw[brown] (0.5,1.5) circle (1.5pt);
    \filldraw[brown] (0.75,2.25) circle (1.5pt);

    \draw[darkred,very thick,dotted]
        plot[domain=1:3,samples=100,variable=\t]
        ({0.5*\t+0.25*(\t-2)*(\t-2)},
         {0.5*\t-0.25*(\t-2)*(\t-2)});
    \filldraw[darkred] (0.75,0.25) circle (1.5pt);
    \filldraw[darkred] (1,1) circle (1.5pt);
    \filldraw[darkred] (1.75,1.25) circle (1.5pt);

    \begin{scope}[shift={(2.50, 1.05)}, scale=0.81, transform shape]
        \filldraw[fill=white,fill opacity=0.94,draw=gray!45,
                  text opacity=1]
            (0,0) rectangle (1.63,1.15);

        \draw[orange!60!yellow,very thick,dashed]
            (0.05,0.99) -- (0.20,0.99);
        \node[anchor=west] at (0.22,0.99)
            {\small $w^{FB}_a(y)=\frac12y$};

        \draw[darkblue,very thick,densely dotted]
            (0.05,0.71) -- (0.20,0.71);
        \node[anchor=west] at (0.22,0.71)
            {\small $w(y)=\frac34y$};

        \draw[brown,very thick,loosely dotted]
            (0.05,0.43) -- (0.20,0.43);
        \node[anchor=west] at (0.22,0.43)
            {\small $w(y)=\frac14y$};

        \draw[darkred,very thick,dotted]
            (0.05,0.15) -- (0.20,0.15);
        \node[anchor=west] at (0.22,0.15)
            {\small $w(y)=\frac12y+\frac14(y-2)^2$};
    \end{scope}

\end{tikzpicture}
\vspace{-0.08in}
\caption{Misspecification sharing}
\label{fig:misspecification-sharing}
\end{figure}
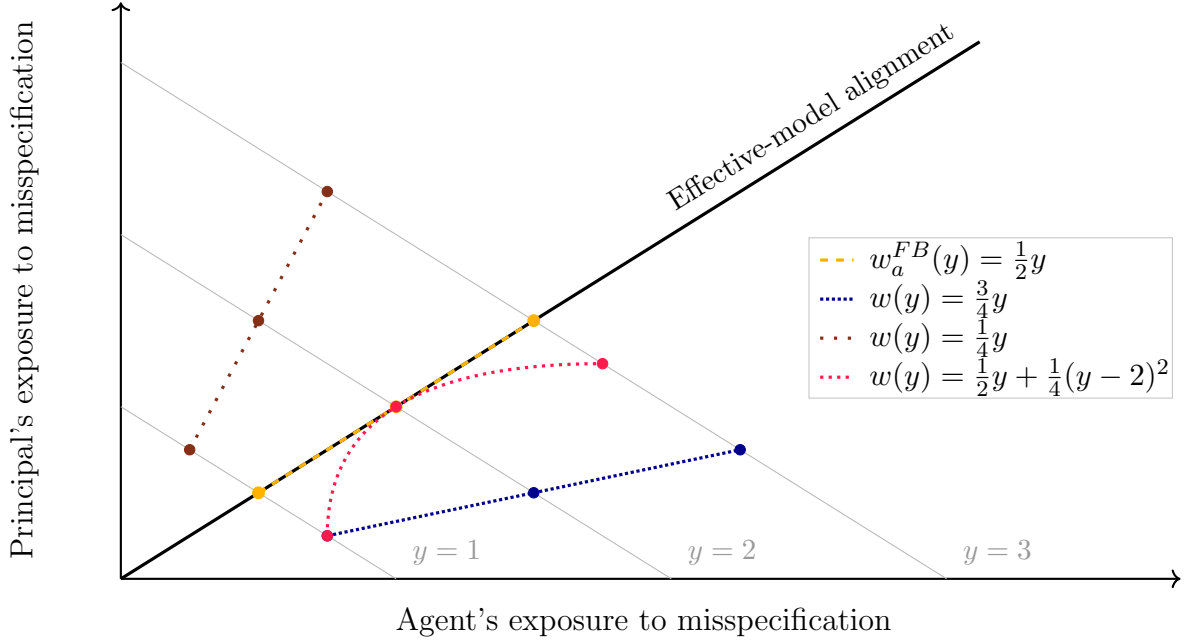

 \par The $45$-degree line is the \textit{effective-model alignment} line: it is derived from Corollary \ref{cor:effective-models}, which requires $\pi_P^a(y;w)=\pi_A^a(y;w)$ for all $y$, or equivalently $w(y)=y-w(y)$ for all $y$ since $\lambda_A=\lambda_P=1$ and $\alpha_a=0$. Each curve connects the allocations selected by one contract across the three outputs. The nonlinear contract $w(y)=y/2+(y-2)^2/4$ is aligned at $y=2$ but changes the division of output across realizations and therefore fails to align the parties' effective models at all three outputs. The linear contracts $w(y)=3y/4$ and $w(y)=y/4$ allocate a constant output share, but give the agent too much and too little misspecification exposure, respectively. In contrast, the optimal contract $w^{FB}_a(y)=y/2$ lies directly on the alignment line at every output. Thus, linearity keeps the sharing rule constant across outputs, while the unique slope $\frac{\lambda_P}{\lambda_A+\lambda_P}=\frac{1}{2}$ is optimal because it selects the constant share that aligns the parties' marginal valuations at every output.

\subsection{Robustness pooling and action choice}\label{subsec:robustness-pooling}

The linear sharing rule has an implication for the value of the agency relationship. The contract does more than divide a fixed robust surplus: by allocating output exposure across two parties, it determines how costly misspecification is for the organization. The next result shows how optimal misspecification sharing affects the choice of actions.

\begin{proposition}\label{prop:robustness-pooling}
For every observable action $a$, the principal's first-best value after optimizing over contracts is
$$
\max_{w: V_A^a(w)\ge U_0}V_P^a(w)
=
\mathcal V_{\lambda_{AP}}^a(y)-c(a)-U_0.
$$
Thus, if the principal also chooses the observable action, every first-best action belongs to
$$
\argmax_{a\in\Acal}
\Big\{
\mathcal V_{\lambda_{AP}}^a(y)-c(a)
\Big\}.
$$
The optimized first-best value weakly exceeds the principal's value under either a constant participating wage or a participating contract that gives the agent the entire output claim. If output is nonconstant and  Corollary \ref{cor:firstbest-unique} holds, then both inequalities are strict.
\end{proposition}

Proposition \ref{prop:robustness-pooling} gives the linear contract a simple organizational interpretation. In the first-best problem, the parties behave as if output were evaluated against the organization's misspecification concern $\lambda_{AP}=\frac{\lambda_A\lambda_P}{\lambda_A+\lambda_P}$. Since $\lambda_{AP}<\min\{\lambda_A,\lambda_P\}$, the organization is effectively less sensitive to misspecification than either party would be if one side alone bore the entire output claim. This reduction does not come from risk diversification, as both observe the same output. Instead, dividing a single uncertain payoff optimally dilutes the pooled (or aggregate) misspecification concern. The sharing rule can therefore create a strict first-best gain relative to allocating the entire output claim to either party.

\subsection{Discussion of key assumptions}\label{subsec:common-variational}
We now explore what happens to Theorem \ref{thm:firstbest-linear} when some of our assumptions are relaxed.

\subsubsection{Misspecification structure}
In our baseline analysis, the principal and agent use the same ambiguity index  while differing only in the intensity of their misspecification concerns. This means that they agree on the relative statistical cost of departures from the benchmark description of the action-output relationship but disagree about how cautious to be. This common statistical structure is what allows a single output share to align their valuations. \par Two departures can break that alignment: the parties may start from different reference models, creating disagreement about output probabilities, or they may use nonproportional ambiguity indices, assigning different relative statistical costs to the same probability distortion. Proposition \ref{prop:different-reference-models} illustrates the first departure in the multiplier case by allowing the parties to use different  reference models $q^a_A$ and $q^a_P$. Proposition \ref{prop:different-divergences} characterizes the second departure for a broad subclass of variational preferences.

\begin{proposition}\label{prop:different-reference-models}
Fix an action $a$ and suppose the parties have multiplier preferences with reference models $q^a_A,q^a_P\in\operatorname{ri}\DeltaY$. Then, the first-best contract is unique:
$$w_a(y)=\alpha_a+\frac{\lambda_P}{\lambda_A+\lambda_P}\hspace{0.03in}y+\frac{1}{\lambda_A+\lambda_P}\log \frac{q^a_A(y)}{q^a_P(y)},$$
where the intercept $\alpha_a$ is uniquely determined by binding participation. Thus, the contract is linear in output if and only if $\log\frac{q^a_A(y)}{q^a_P(y)}$ is linear in output on $\Y$.
\end{proposition}

The first two terms in this optimal contract are the sharing rule from Theorem \ref{thm:firstbest-linear}. The last term is different. It transfers wealth toward outputs that the agent regards as relatively more likely than the principal does and away from outputs that the principal regards as relatively more likely. It is therefore a disagreement term. Unless the likelihood ratio between the two reference models happens to have the appropriate exponential form in output, this term breaks linearity. Assuming a common reference model rules out this belief-disagreement channel and isolates the effect of bilateral misspecification concerns.

The proportionality assumption on the ambiguity indices is also substantive. The following result shows that, within a large subclass of variational preferences, proportionality is precisely the restriction that selects linear contracts for all payoff vectors.

\begin{proposition}\label{prop:different-divergences}
Let the parties' ambiguity indices  $\Gamma_A,\Gamma_P:\DeltaY\to[0,\infty)$ be grounded, continuous, and strictly convex, and continuously differentiable on $\operatorname{ri}\DeltaY$.  Fix a constant $s\in(0,1)$. Then, for any $\alpha\in\R$, the linear sharing rule $w=\alpha\m{1}+sx$ is Pareto efficient for every payoff vector $x\in\R^\Y$ if and only if $\frac{1}{s}\Gamma_A(p)=\frac{1}{1-s}\Gamma_P(p)$ for every $p\in\DeltaY$.
\end{proposition}

Smoothness and strict convexity in Proposition \ref{prop:different-divergences} have a transparent role: each payoff must induce a unique effective model, and the statistical penalty must have well-defined marginal costs of moving probability mass. This captures a broad class of smooth variational preferences, such as those whose ambiguity indices are $\varrho$-divergences. If the parties disagree about the relative statistical cost of models, a single slope may not align their marginal valuations, and therefore nonlinear contracts can become useful.

\subsubsection{Unilateral robustness and risk neutrality}\label{subsec:miao-sharing}

Suppose the ambiguity index $\Gamma^a$ has a unique minimizer, so that when there is no misspecification concern, the variational preference becomes an EU preference. Then, as the agent's concern vanishes, $\lambda_A\downarrow0$, the sharing rule in Theorem \ref{thm:firstbest-linear} satisfies $\frac{\lambda_P}{\lambda_A+\lambda_P}\rightarrow1$. An EU agent therefore becomes the residual claimant at the margin: a misspecification-concerned principal finds it optimal to transfer all marginal output exposure to the agent.

In \citeauthor{miao16}'s (\citeyear{miao16}) first-best dynamic benchmark, a robust principal would ideally have the risk-neutral agent insure her by absorbing all project cash-flow uncertainty; however, limited liability prevents complete transfer. Their moral-hazard problem then adds incentive and liquidation considerations, which make the amount of uncertainty borne by the agent state dependent. Theorem \ref{thm:firstbest-linear} isolates the underlying force: without those dynamic frictions, unilateral robustness leads to complete transfer of model uncertainty to the agent, while bilateral robustness determines the interior share $\frac{\lambda_P}{\lambda_A+\lambda_P}$.

\par As in the robust-contract literature \citep[][]{linear19}, risk neutrality is important for the interpretation of the linear sharing rule. It removes conventional risk sharing from the first-best problem, so that all output-dependent curvature in the parties' valuations comes from their concern about misspecification. If either party instead has nonlinear utility over money, then marginal utility itself varies with payoff realized in each state, so a linear contract need not equalize the two parties' marginal valuations across outputs.

\subsection{Relation to the robust-contract literature}\label{subsec:carroll}

The key distinction from \citet{linear15}, \citet{linear20}, and \citet{linear22} is the role that linearity plays. In those papers, linearity is a device for protecting incentives against incomplete knowledge of the contracting environment: uncertainty concerns the actions that a contract may induce, or the parameters of the output distribution over which incentive compatibility must remain valid. As summarized in \citet{linear19}, these are \textit{worst-case} approaches based on \citeauthor{gilboa89}'s (\citeyear{gilboa89}) maxmin logic. Theorem \ref{thm:firstbest-linear} instead has no incentive constraint and its variational logic differs from maxmin. The action is observable and fixed, and the two parties evaluate misspecification through proportional variational penalties over output models. Linearity is therefore an efficient sharing rule rather than a robust incentive rule: it divides output so that the parties attach the same marginal weight to every realization under their effective models.

The mathematical arguments in the proofs reflect this difference. \citet{linear15} obtains a supporting affine rule from the payoff guarantee generated by the set of technologies the principal considers possible, while \citet{linear22} abstracts the properties of contract responses that sustain this guarantee and \citet{linear20} exploits the geometry of a robust-optimization problem. As a result, the slope in these contracts typically depends on the technological primitives governing expected output and action cost \citep[e.g.,][eq. (12)]{linear15}. Moreover, most of these papers rely on limited liability, because without it, their problems would generally admit the familiar ``sell the firm to the agent'' solution. By contrast, our argument uses neither an unknown technology, limited liability, nor incentive constraints that must hold uniformly over an uncertainty set. First-best optimality instead requires a common effective model (Corollary \ref{cor:effective-models}), and proportional ambiguity indices make a constant output share implement that condition. This changes what determines the slope in our framework: the share $\lambda_P/(\lambda_A+\lambda_P)$ depends only on the parties' relative misspecification concerns.

\section{Moral Hazard}\label{sec:moralhazard}

The first-best analysis removed hidden actions in order to isolate misspecification sharing. The natural next step is to introduce hidden actions into our framework and study how the first-best sharing rule changes when the contract must also induce the desired action. Fix a desired action $a^*\in\Acal$. The principal chooses a contract to maximize her variational payoff under the desired action subject to participation and incentive compatibility:
$$
(\mathsf{MH}_{a^*})\qquad \max_{w\in\R^{\Y}} V_P^{a^*}(w)
\quad\text{subject to}\quad
V_A^{a^*}(w)\ge U_0,
\quad
V_A^{a^*}(w)\ge V_A^a(w)\ \forall a\ne a^*.
$$
The main question is whether hidden actions overturn the first-best sharing rule. We start by identifying when a first-best contract survives the addition of hidden actions.
\begin{lemma}\label{lem:MH-binding}
Every solution of $(\mathsf{MH}_{a^*})$ satisfies $V_A^{a^*}(w)=U_0$. If a first-best solution $w^{FB}_{a^*}$ satisfies $V_A^{a^*}(w^{FB}_{a^*})\ge V_A^a(w^{FB}_{a^*})$ for every $a\ne a^*$, then $w^{FB}_{a^*}$ solves $(\mathsf{MH}_{a^*})$.
\end{lemma}

\subsection{The moral-hazard contract}\label{subsec:MH-contract}

For any action $a\ne a^*$, call the incentive constraint against $a$ \textit{active} at $w$ if $V_A^{a^*}(w)=V_A^a(w)$. We say that a solution is \textit{smooth} if the effective-model sets entering the principal's objective, participation, and every active incentive constraint are singletons. A smooth solution is \textit{regular} if these differentiable constraints satisfy the Mangasarian-Fromovitz constraint qualification.\footnote{This is a standard sufficient condition for the existence of KKT multipliers. Appendix \ref{app:regularity} gives a characterization in terms of effective models and primitive sufficient conditions.} The next result is the moral-hazard analogue of Corollary \ref{cor:effective-models}.

\begin{theorem}\label{thm:MH-general}
Let $w^*$ be a smooth, regular solution of $(\mathsf{MH}_{a^*})$. Then, there exist nonnegative Lagrange multipliers $\{\beta^a\}_{a\ne a^*}$, with $\beta^a=0$ for every inactive incentive constraint, such that
$\pi_P^{a^*}(\cdot;w^*)=\pi_A^{a^*}(\cdot;w^*)+\sum_{a\ne a^*}\beta^a\big[\pi_A^{a^*}(\cdot;w^*)-\pi_A^a(\cdot;w^*)\big]$.
\end{theorem}

Theorem \ref{thm:MH-general} captures 
the main implications  of moral hazard. Recall that at the first best the principal and agent must agree on a probability model that prices marginal transfers. Hidden actions create an additional wedge. The vector $\pi_A^{a^*}-\pi_A^a$ is the output-by-output marginal effect of compensation on the agent's incentive advantage of the desired action over deviation $a$, while $\beta^a$ is the shadow value of that incentive constraint. The principal's marginal valuation therefore differs from the agent's desired-action valuation only through a nonnegative combination of the wedges required to deter active deviations. If no incentive constraint is active, all the $\beta^a$'s vanish and the condition in Theorem \ref{thm:MH-general} reduces to the common-effective-model condition in Corollary \ref{cor:effective-models}.

For general variational preferences, Theorem \ref{thm:MH-general} is naturally stated in terms of marginal robust probabilities because the mapping from effective models back to compensation depends on the ambiguity index. The next result shows that multiplier preferences make that inversion explicit and therefore yield a unique optimal contract in closed form.

\begin{corollary}\label{thm:MH-contract}
Suppose both parties have multiplier preferences and the problem $(\mathsf{MH}_{a^*})$ is feasible.\footnote{Corollary \ref{thm:MH-contract} does not require the smoothness or regularity  conditions imposed in Theorem \ref{thm:MH-general}. Its proof works with a convex reformulation of $(\mathsf{MH}_{a^*})$ for multiplier criteria, so feasibility is sufficient.} Then, it has a unique solution $w^*$. There are nonnegative Lagrange multipliers $\{\beta^a\}_{a\ne a^*}$ and $\alpha\in\R$, with $\beta^a=0$ for every inactive incentive constraint, such that
\begin{align*}
 w^*(y)=\alpha+\frac{\lambda_P}{\lambda_A+\lambda_P}y
 +\frac{1}{\lambda_A+\lambda_P}\log\Big(1+\sum_{a\ne a^*}\beta^a\Big[1-e^{\lambda_A(c(a)-c(a^*))}\frac{q^a(y)}{q^{a^*}(y)}\Big]\Big)
\end{align*}
for every output $y\in\Y$. The expression inside the logarithm is strictly positive for every output. The intercept $\alpha$ is uniquely determined by binding participation.
\end{corollary}

Corollary \ref{thm:MH-contract} shows how hidden actions change the optimal contract in the multiplier case. The first two terms are the first-best sharing rule from Theorem \ref{thm:firstbest-linear}. The extra logarithmic term is the incentive correction. It depends on likelihood ratios $q^a(y)/q^{a^*}(y)$ for deviations, adjusted by effort costs and by the agent's misspecification concern. Outputs that are relatively more likely under deviations must be made less attractive, while outputs that are relatively more likely under the desired action can be rewarded. Appendix \ref{app:rho-incentives} shows that this qualitative insight and the fact that moral hazard adds an incentive correction are not special to multiplier preferences: they hold also for all $\varrho$-divergences.  

Notice that the coefficient on the nonlinear term is $1/(\lambda_A+\lambda_P)$. Thus, when the principal's misspecification concern $\lambda_P$ is sufficiently large and the correction term remains well-behaved, the moral-hazard contract is close to the first-best linear contract. The reason is that a highly robust principal is very reluctant to retain residual-output exposure, so the extra nonlinear correction needed for incentives receives a small coefficient.

\subsection{A floor on performance pay}\label{subsec:MH-slope-floor}

For the remainder of this section, we retain the multiplier specification so that the incentive correction is in closed form. Suppose the desired action $a^*$ dominates every deviation in the monotone likelihood-ratio order: $q^a(y)/q^{a^*}(y)$ is weakly decreasing in $y$ for every $a\neq a^*$. Then, the nonlinear term in Corollary \ref{thm:MH-contract} cannot undo the first-best component.

\begin{corollary}\label{prop:MH-slope-floor}
Let $w^*$ be the solution of $(\mathsf{MH}_{a^*})$ in Corollary \ref{thm:MH-contract}, and suppose $q^a(y)/q^{a^*}(y)$ is weakly decreasing in $y$ for every $a\ne a^*$. Then, for any $y'>y$ in $\Y$,
$$
w^*(y')-w^*(y)\ge \frac{\lambda_P}{\lambda_A+\lambda_P}(y'-y).
$$
The optimal contract is strictly increasing. Moreover, the inequality is strict between $y$ and $y'$ whenever some deviation with $\beta^a>0$ has $q^a(y')/q^{a^*}(y')<q^a(y)/q^{a^*}(y)$.
\end{corollary}

Corollary \ref{prop:MH-slope-floor} gives the first-best sharing rule a second interpretation under moral hazard: it is a lower bound on performance sensitivity. Under the likelihood-ratio ordering, incentives can steepen the contract but cannot flatten it below the amount of output sharing that would already be optimal with observable action. Thus, the nonlinear term in Corollary \ref{thm:MH-contract} rewards outputs that are relatively diagnostic of the desired action on top of a baseline share that is required to allocate misspecification exposure efficiently.

A comparison with \citeauthor{smooth05}'s (\citeyear{smooth05}) smooth ambiguity is also informative. \citet{kellner17} shows that smooth-ambiguity preferences can generate nonmonotone optimal incentive schemes even under a natural likelihood-ratio condition. Under multiplier robustness, the first-best sharing component instead provides a quantitative monotonicity floor whenever the desired action monotone-likelihood-ratio dominates its deviations.

\subsubsection{A two-output case}\label{subsec:twooutcome}

The closed-form formula in Corollary \ref{thm:MH-contract} is especially transparent with two outputs and two actions. This case yields a closed-form incentive bonus and a sharp implementability limit. Let $\Y=\{0,1\}$ and $\Acal=\{H,L\}$. Write $q^H(1)=p_H$ and $q^L(1)=p_L$, with $0<p_L<p_H<1$, and let $\Delta c_{HL}:=c(H)-c(L)>0$. A contract has $w(0)=s$ and $w(1)=s+b$. For $p\in(0,1)$, define $\Psi(p,b):=-\frac{1}{\lambda_A}\log(1-p+pe^{-\lambda_A b})$. Incentive compatibility for $H$ is $\Psi(p_H,b)-\Psi(p_L,b)\ge\Delta c_{HL}$. The robust compensation advantage of $H$ is bounded even as the bonus becomes arbitrarily large. Define $\overline C_{HL}:=\frac{1}{\lambda_A}\log\frac{1-p_L}{1-p_H}.$

\begin{corollary}\label{cor:two-output-bonus}
A finite contract implements $H$ if and only if $\Delta c_{HL}<\overline C_{HL}$. If this holds,
$$
 b^*=\max\Big\{b^{IC},\frac{\lambda_P}{\lambda_A+\lambda_P}\Big\}
$$
is the optimal bonus for implementing $H$, where
$$
b^{IC}=\frac{1}{\lambda_A}\log
\frac{p_H-e^{-\lambda_A\Delta c_{HL}}p_L}
{e^{-\lambda_A\Delta c_{HL}}(1-p_L)-(1-p_H)}.
$$
The optimal fixed payment $s^*$ is then uniquely determined by binding participation.
\end{corollary}

The EU benchmark isolates what misspecification changes. If the agent is an EU maximizer, the incentive constraint is $(p_H-p_L)b\ge\Delta c_{HL}$, so $b^{EU}=\Delta c_{HL}/(p_H-p_L)$ and every finite cost difference can be overcome by a sufficiently large bonus; see Appendix \ref{app:EU} for these derivations. Here, it follows that $b^{IC}\to b^{EU}$ as $\lambda_A\downarrow0$, and locally $b^{IC}=b^{EU}+\frac{\lambda_A}{2}(1-p_H-p_L)(b^{EU})^2+O(\lambda_A^2).$ Thus, a small misspecification concern raises the incentive bonus when $p_H+p_L<1$ and lowers it when $p_H+p_L>1$: pessimism attenuates probability differences at low success rates but, since multiplier value is convex in the success probability, it can amplify them at high success rates. This local effect is distinct from the global implementability limit. Since $\overline C_{HL}$ falls with $\lambda_A$, sufficiently strong misspecification concern eventually makes $H$ impossible to implement.

The two terms inside the maximum operator of $b^*$ have different economic roles. The first, $b^{IC}$, is the minimum performance sensitivity required by hidden action and depends on the technology, the cost difference, and the agent's misspecification concern. The second, $\lambda_P/(\lambda_A+\lambda_P)$, is the first-best sharing slope and depends only on the parties' relative misspecification concerns. If $b^{IC}<\lambda_P/(\lambda_A+\lambda_P)$, the first-best contract already provides more than enough incentives and the incentive constraint is slack. At equality it binds without distorting the first-best sharing rule. If $b^{IC}>\lambda_P/(\lambda_A+\lambda_P)$, moral hazard raises performance sensitivity above its first-best level.

\section{Effort Polarization and Model Ambiguity}\label{sec:ext}

The baseline analysis abstracts from two relevant sources of friction: a continuum of actions and multiple competing models. This section studies the effects of these frictions.

\subsection{Effort polarization}\label{sec:continuum}

Section \ref{sec:moralhazard} studies moral hazard with finite actions. With a continuum of actions, it is tempting to replace the family of incentive constraints by the agent's first-order condition. Misspecification concerns introduce a curvature restriction. Let the action set be an interval $\Acal=[\underline a,\bar a]\subset\R$ and suppose $c$ is twice differentiable on its interior. For a fixed contract $w$, let $G(a;w):=\mathcal V_{\lambda_A}^a(w)$ denote the agent's gross compensation value. Whenever $G(\cdot;w)$ is twice differentiable and an interior action $a$ is implemented, the first- and second-order necessary conditions imply $G_a(a;w)=c'(a)$ and $G_{aa}(a;w)\le c''(a)$. Equivalently, when $G_a(a;w)\neq0$, defining the local ambiguity-curvature index $\kappa_A(a;w):=G_{aa}(a;w)/G_a(a;w)^2$ gives the necessary inequality $c''(a)\ge\kappa_A(a;w)c'(a)^2$.

\begin{proposition}\label{prop:continuum-general}
Suppose $a\mapsto\mathcal V_{\lambda_A}^a(w)$ is twice differentiable near an interior action $a$ implemented by a finite contract $w$. Then, $\partial_a\mathcal V_{\lambda_A}^a(w)=c'(a)$ and $\partial_{aa}\mathcal V_{\lambda_A}^a(w)\le c''(a)$. If $c'(a)\neq0$, this inequality is equivalent to $c''(a)\ge\kappa_A(a;w)c'(a)^2$.
\end{proposition}
 Different variational ambiguity indices can generate different curvatures of compensation value in effort, so the implementability restriction depends on the local index $\kappa_A$. Multiplier preferences are special because this index can become the constant $\lambda_A$, so we focus on this case for the rest of this section to obtain some closed-form expressions. A success-contingent bonus makes the agent's robust value of compensation convex in the success probability, so stronger incentives can generate increasing rather than diminishing returns to effort. The consequence is that intermediate actions may cease to be implementable and the agent's response to performance pay can become discontinuous.

To isolate this intensive-margin effect, specialize to the binary-output environment $\Y=\{0,1\}$ from Section \ref{subsec:twooutcome} and replace the two actions with $\Acal=[\underline a,\bar a]\subset(0,1)$. We parameterize action by the probability of success, so that the reference model satisfies $q^a(1)=a$. This is a standard reparametrization of the continuum-effort, two-outcome moral-hazard model; \citet[Section 5.2.2]{laffontmartimort02}, for example, explicitly normalize effort to equal the probability of high performance.\footnote{More generally, if primitive effort $e$ generates success with a strictly increasing probability $\pi(e)$, one may define $a:=\pi(e)$ and write $c(a):=\psi(\pi^{-1}(a))$. Thus, $q^a(1)=a$ need not be interpreted as imposing an affine technology in primitive effort: it measures effort in units of the success probability.} Let $c:[\underline a,\bar a]\to\R$ be twice continuously differentiable with $c'(\cdot)>0$. Write a contract as $w(0)=s$ and $w(1)=s+b$, where $b$ is the success bonus. The agent chooses $a$ to maximize $M(q^a,w;\lambda_A)-c(a)$. 

For a fixed bonus $b$, the agent's gross compensation value is $s-\lambda_A^{-1}\log(1-a+a e^{-\lambda_A b})$. Under EU this expression is affine in $a$. In contrast, under multiplier robustness, it is strictly convex in $a$ whenever $b\ne0$. The following result gives two contract-independent restrictions that this curvature imposes on interior actions.

\begin{obs}\label{prop:continuum-interior}
If a finite contract implements an action $a\in(\underline a,\bar a)$, then $b>0$ and
\begin{equation}\label{eq:continuum-conditions}
 \lambda_A(1-a)c'(a)<1,
 \qquad
 c''(a)\ge \lambda_Ac'(a)^2.
\end{equation}
\end{obs}

The first inequality in \eqref{eq:continuum-conditions} is a marginal-incentive bound. Even an arbitrarily large finite success bonus cannot make the multiplier marginal value of effort exceed $1/[\lambda_A(1-a)]$. Conditional on this bound, the first-order condition pins down the only bonus that can make $a$ stationary, namely $b=\lambda_A^{-1}\log\big\{[1+\lambda_A a c'(a)]/[1-\lambda_A(1-a)c'(a)]\big\}$. The second inequality is the multiplier specialization of Proposition \ref{prop:continuum-general}: the increase in marginal effort cost must be large enough to offset the endogenous convexity of robust compensation. These restrictions have no counterpart in the EU benchmark. In particular, greater misspecification concern can remove interior actions from the implementable set even though transfers are unrestricted. The second restriction also explains why the first-order approach can fail precisely where it is most tempting to use it.

\begin{corollary}\label{cor:continuum-foa}
If a finite contract makes an interior action $a^*$ satisfy the agent's first-order condition and $c''(a^*)<\lambda_Ac'(a^*)^2$, then $a^*$ is a strict local minimum of the agent's objective and is not implemented by the contract.
\end{corollary}

Thus, replacing the global incentive constraints by a local first-order condition can identify an action that the agent moves away from in either direction. This matters because the principal may appear to be able to fine-tune effort by adjusting $b$, while the contract induces a discrete move to a distant action. The finite-action setting in Section \ref{sec:moralhazard} does not hide this issue because it retains every deviation constraint; a continuum setting must do the same unless additional conditions guarantee global concavity.

The main implication is effort \textit{polarization}. Suppose the curvature bound in Observation \ref{prop:continuum-interior} fails at every interior action. Let $\Delta c:=c(\bar a)-c(\underline a)$ and let $\overline C:=\frac{1}{\lambda_A}\log\frac{1-\underline a}{1-\bar a}$. The quantity $\overline C$ is the largest robust compensation advantage that the high endpoint can approach relative to the low endpoint as the success bonus becomes arbitrarily large. When $\Delta c<\overline C$, define the polarization threshold $\widetilde{b}:=-\frac{1}{\lambda_A}\log\frac{e^{-\lambda_A\Delta c}(1-\underline a)-(1-\bar a)}{\bar a-e^{-\lambda_A\Delta c}\underline a}$.

\begin{theorem}\label{thm:continuum-polarization}
Suppose $c''(a)<\lambda_Ac'(a)^2$ for every $a\in(\underline a,\bar a)$. If $\Delta c\ge\overline C$, every finite contract uniquely induces $\underline a$. If $\Delta c<\overline C$, then
\begin{equation}\label{eq:continuum-jump}
\argmax_{a\in[\underline a,\bar a]}\Big\{M(q^a,w;\lambda_A)-c(a)\Big\}
=
\begin{cases}
\{\underline a\},& b<\widetilde{b},\\
\{\underline a,\bar a\},& b=\widetilde{b},\\
\{\bar a\},& b>\widetilde{b}.
\end{cases}
\end{equation}
\end{theorem}

Theorem \ref{thm:continuum-polarization} shows that a continuum of available effort levels need not generate a continuum of contractible effort levels. When cost curvature is too weak, the agent chooses only the two endpoints. Below the threshold $\widetilde{b}$, incentives are too weak to compensate for the additional cost of high effort; above $\widetilde{b}$, the multiplier criterion's increasing marginal valuation of the success bonus makes maximal effort optimal. Thus, an arbitrarily small increase in performance pay around $\widetilde{b}$ produces a discontinuous jump from $\underline a$ to $\bar a$.

\subsubsection{Illustration: polarization and discontinuity}

We provide a simple example to visualize the discontinuity in Theorem \ref{thm:continuum-polarization}. Let $\lambda_A=2$, $[\underline a,\bar a]=[0.20,0.80]$, and $c(a)=0.40a+0.05a^2$. The cost function is strictly convex, but $c''(a)=0.10<2c'(a)^2$ on the interior, so robust-compensation curvature dominates cost curvature. Moreover, $\Delta c=0.27<\overline C=0.69$, and Theorem \ref{thm:continuum-polarization} gives $\widetilde b=0.47$.

Figure \ref{fig:polarization-response} shows what the curvature implies for incentive design. Under multiplier preferences, the induced-action correspondence consists only of the two horizontal branches. At $b=\widetilde b$, both $0.20$ and $0.80$ are optimal; there is no vertical branch connecting them because no intermediate action is implementable. An arbitrarily small increase in the bonus through $0.47$ therefore moves the agent from the lowest to the highest action.

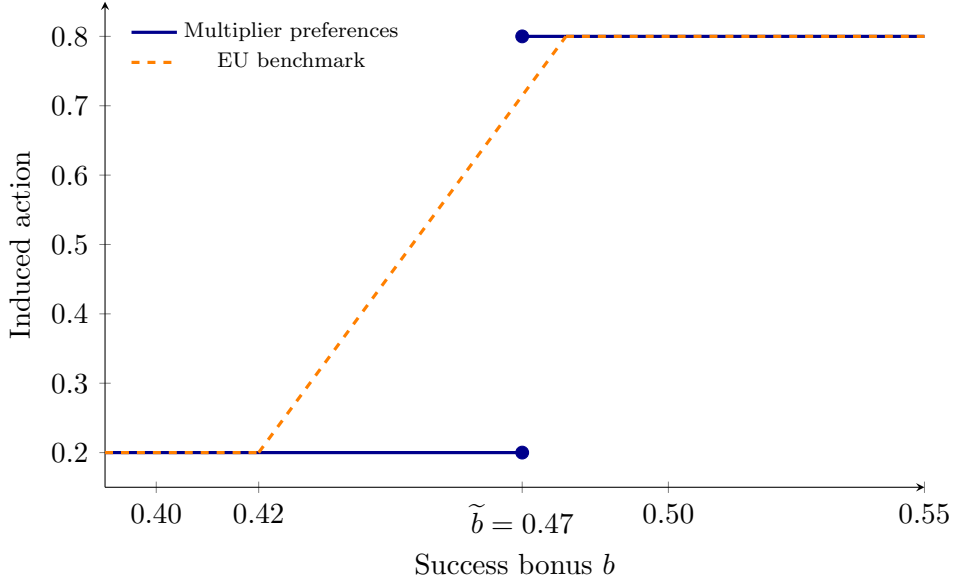
\begin{figure}[hbt!]
\centering
\begin{tikzpicture}
\begin{axis}[
width=0.78\textwidth,
height=7.5cm,
xmin=0.39,xmax=0.55,
ymin=0.15,ymax=0.85,
xtick={0.40,0.42,0.471455679,0.50,0.55},
xticklabels={$0.40$,$0.42$,$\widetilde b=0.47$,$0.50$,$0.55$},
ytick={0.20,0.30,0.40,0.50,0.60,0.70,0.80},
xlabel={Success bonus $b$},
ylabel={Induced action},
label style={font=\small},
tick label style={font=\small},
legend style={
draw=none,
fill=none,
font=\scriptsize,
at={(0.02,0.98)},
anchor=north west 
},
axis lines=left
]
\addplot[very thick,darkblue,domain=0.39:0.471455679,samples=2] {0.20};
\addlegendentry{Multiplier preferences}

\addplot[very thick,darkblue,domain=0.471455679:0.55,samples=2,forget plot] {0.80};

\addplot[only marks,mark=*,mark size=2.4pt,darkblue,forget plot]
coordinates {(0.471455679,0.20) (0.471455679,0.80)};

\addplot[very thick,orange,dashed,domain=0.39:0.42,samples=2] {0.20};
\addlegendentry{EU benchmark}

\addplot[very thick,orange,dashed,domain=0.42:0.48,samples=100,forget plot] {10*x-4};

\addplot[very thick,orange,dashed,domain=0.48:0.55,samples=2,forget plot] {0.80};
\end{axis}
\end{tikzpicture}
\caption{Discontinuity in implementation. Parameters: $\lambda_A=2$, $[\underline a,\bar a]=[0.20,0.80]$, and $c(a)=0.40a+0.05a^2$; the EU benchmark uses the same action set and cost function.}
\label{fig:polarization-response}
\end{figure}

In summary, stronger performance pay need not operate on an intensive margin. If the principal would prefer an intermediate action, adjusting the bonus cannot fine-tune effort toward that target: near $\widetilde b$, she must choose between an action below it and one above it. This is stronger than the conventional wisdom that moral hazard depresses effort; misspecification concern can instead disconnect the implementable action set.

\subsubsection{EU benchmark}

Figure \ref{fig:polarization-response} also isolates the source of the discontinuity. With the same strictly convex cost function, the EU agent moves continuously through every interior action as the bonus rises from $0.42$ to $0.48$. The next result formalizes this EU benchmark.

\begin{obs}\label{prop:continuum-eu}
Suppose $c$ is strictly convex. If the agent is an EU maximizer, every action $a\in(\underline a,\bar a)$ is uniquely implementable by the success bonus $b=c'(a)$. On the interior range of bonuses, the induced action is $(c')^{-1}(b)$ and varies continuously with $b$.
\end{obs}

Under EU, the agent's objective is $s+ab-c(a)$: gross compensation is affine in effort, so strict convexity of cost is enough to make the problem strictly concave. Under multiplier preferences, gross compensation itself is convex, and convex effort costs need not dominate this force. Moving from discrete to continuous actions therefore reveals a distinct effect of misspecification concern that is not present in the finite-action problem: robustness can change not only the cost of implementing an action, but also the shape and connectedness of the implementable action set.

\subsection{Model ambiguity and dispersion}\label{subsec:model-libraries}

The baseline analysis attaches one variational ambiguity index to each action. However, organizations often work with a set of plausible reference models for the same action, which may reflect different forecasts, scientific theories, or expert assessments. The multitude of models introduces a new friction in the agency relationship: \textit{model ambiguity}. This is the structure used by \citet{kellner15,kellner17} in agency problems with smooth ambiguity, where each action is associated with a second-order belief over candidate output distributions, and model-specific evaluations are aggregated with respect to that belief in a Bayesian way. We adopt the same Bayesian ``model library''  while retaining a variational structure that captures misspecification concerns about each candidate model.

Let $\Gamma:\DeltaY\times\DeltaY\to[0,\infty]$ be lower semicontinuous, jointly convex, and satisfy $\Gamma(p,q)=0\iff p=q$. For $\lambda>0$, define $\mathcal W_\lambda(q,x):=\min_{p\in\DeltaY}\big\{\E_p[x]+\frac{1}{\lambda}\Gamma(p,q)\big\}$. Notice that $\mathcal W_\lambda$ is a special case of the criterion in \citet[][eq. (1)]{hansenmiss25} when the set of reference (or ``structured'') models is a singleton, and the joint-convexity restriction is satisfied in many relevant cases, such as the $\varrho$-divergences and $\phi$-divergences.

For each action $a$, let $Q^a\subset\DeltaY$ be a nonempty finite set of full-support reference models and let $\mu^a$ be a (second-order) probability distribution on $Q^a$. The mean model is $\bar q^a(y):=\E_{\mu^a}[q(y)]$. We average the model-specific variational values linearly across the second-order belief $\mu^a$. This is the linear-aggregator special case of the smooth Bayesian criterion in \citet[][eq. (32)]{hansenmiss25}: their general criterion applies a quasi-arithmetic mean to reference-model variational evaluations, while the identity aggregator used here leaves only the Bayesian average. The entropic specification $\Gamma(p,q)=\DKL(p\|q)$ is the \textit{average robust-control} (ARC) criterion, axiomatized by \citet{lanzani2025}. \par The agent's value from action $a$ becomes $\widehat V_A^a(w;\mu^a):=\E_{\mu^a}[\mathcal W_{\lambda_A}(q,w)]-c(a)$. Similarly, for the principal, define 
$
\widehat V_P^a(w;\mu^a):=\E_{\mu^a}[\mathcal W_{\lambda_P}(q,y-w)].
$ Fix a desired action $a^*$. A contract implements $a^*$ under the library profile $\mu=(\mu^a)_{a\in\Acal}$ if it satisfies participation and all incentive constraints using the values $\widehat V_A^a$. Let $\widehat X^{a^*}(\mu)$ denote this implementation set, and define the principal's optimized payoff by $\widehat\Lambda_P^{a^*}(\mu):=\sup_{w\in\widehat X^{a^*}(\mu)}\widehat V_P^{a^*}(w;\mu^{a^*})$, with value $-\infty$ if $\widehat X^{a^*}(\mu)$ is empty.

A library $\tilde\mu^a$ is \textit{more dispersed} than $\mu^a$, written $\tilde\mu^a\succeq\mu^a$, if both have the same mean model and $\int f(q)d\tilde\mu^a(q)\leq\int f(q)d\mu^a(q)$ for every continuous concave $f$ on the relevant convex hull of models. Thus, two organizations can agree on the average model $\bar q^a$ while differing in how much model variation their libraries contain. Theorem \ref{thm:library} shows that the location of model dispersion matters. Dispersion in the desired-action library raises the agent's variational valuation of that action and also raises the principal's valuation of every fixed contract. Dispersion in a deviation library instead makes that deviation more attractive to the agent and therefore makes implementation harder.

\begin{theorem}\label{thm:library}
Fix $a^*$. If $\tilde\mu^{a^*}\succeq\mu^{a^*}$ and $\tilde\mu^a=\mu^a$ $\forall a\ne a^*$, then $\widehat X^{a^*}(\mu)\subseteq\widehat X^{a^*}(\tilde\mu)$, $\widehat V_P^{a^*}(w;\tilde\mu^{a^*})\geq\widehat V_P^{a^*}(w;\mu^{a^*})$ $\forall w$, and $\widehat\Lambda_P^{a^*}(\tilde\mu)\geq\widehat\Lambda_P^{a^*}(\mu)$. If instead $\tilde\mu^{\hat a}\succeq\mu^{\hat a}$ for some deviation $\hat a\ne a^*$ and all else unchanged, then $\widehat X^{a^*}(\tilde\mu)\subseteq\widehat X^{a^*}(\mu)$ and $\widehat\Lambda_P^{a^*}(\tilde\mu)\leq\widehat\Lambda_P^{a^*}(\mu)$.
\end{theorem}
 Section \ref{subsec:endogenous-dispersion} endogenizes model dispersion to study its effect on the principal's profit.

\subsubsection{Illustration: model dispersion}

For the remainder of this illustration, specialize to multiplier preferences, so $\mathcal W_\lambda(q,x)=M(q,x;\lambda)$. Let $\Y=\{0,1\}$ with desired action $H$ and deviation $L$. Set $\lambda_A=\lambda_P=1$, $c(H)=0.10$, $c(L)=0.01$, and $U_0=0$. Under the baseline models, success occurs with probability $0.60$ after $H$ and $0.40$ after $L$. Consider the contract $w(0)=0$ and $w(1)=0.50$, which gives the agent one half of successful output and coincides with the first-best sharing slope when $\lambda_A=\lambda_P$. If a model assigns success probability $p$, the agent's robust value of the bonus before effort cost is $\psi(p):=-\log(1-p+pe^{-1/2})$. At the baseline models, the contract implements $H$ with a small positive incentive margin.

To vary model dispersion without changing expected performance, let $d\in[0,0.35]$. Desired-action dispersion replaces the forecast $0.60$ by two equally weighted forecasts $0.60-d$ and $0.60+d$, leaving $L$ unchanged. Deviation dispersion instead replaces the forecast $0.40$ by $0.40-d$ and $0.40+d$, leaving $H$ unchanged. Figure \ref{fig:library-implementation} plots the agent's incentive margin---his value from $H$ minus his value from $L$---under these two experiments.
\begin{figure}[hbt!]
\centering
\begin{tikzpicture}
\begin{axis}[
width=0.75\textwidth,
height=7.3cm,
xmin=0,xmax=0.35,
ymin=-0.007,ymax=0.027,
xtick={0,0.10,0.20,0.270567,0.35},
xticklabels={$0$,$0.10$,$0.20$,$d^{IC}\simeq0.27$,$0.35$},
ytick={-0.005,0,0.005,0.010,0.015,0.020,0.025},
scaled y ticks=false,                        
yticklabel style={/pgf/number format/fixed}, 
xlabel={Model dispersion $d$},
ylabel={Implementation margin for $H$},
label style={font=\small},
tick label style={font=\small},
legend style={
draw=none,
fill=none,
font=\scriptsize,
at={(0.02,0.98)}, 
anchor=north west 
},
axis lines=left
]

\addplot[black!40,dashed,domain=0:0.35,samples=2,forget plot] {0};

\addplot[very thick,darkblue,domain=0:0.35,samples=180]
{
0.5*
(
-ln(1-(0.60-x)+(0.60-x)*exp(-0.5))
-ln(1-(0.60+x)+(0.60+x)*exp(-0.5))
)
-0.09
+ln(1-0.40+0.40*exp(-0.5))
};
\addlegendentry{Desired-action dispersion}

\addplot[very thick,bittersweet,domain=0:0.35,samples=180]
{
-ln(1-0.60+0.60*exp(-0.5))
-0.09
+0.5*ln(1-(0.40-x)+(0.40-x)*exp(-0.5))
+0.5*ln(1-(0.40+x)+(0.40+x)*exp(-0.5))
};
\addlegendentry{Deviation dispersion}

\addplot[black,dashed,forget plot]
coordinates {(0.270567,-0.007) (0.270567,0)};

\addplot[only marks,mark=*,mark size=2pt,black,forget plot]
coordinates {(0.270567,0)};

\end{axis}
\end{tikzpicture}
\caption{Model dispersion and asymmetric implementation effects.}
\label{fig:library-implementation}
\end{figure}
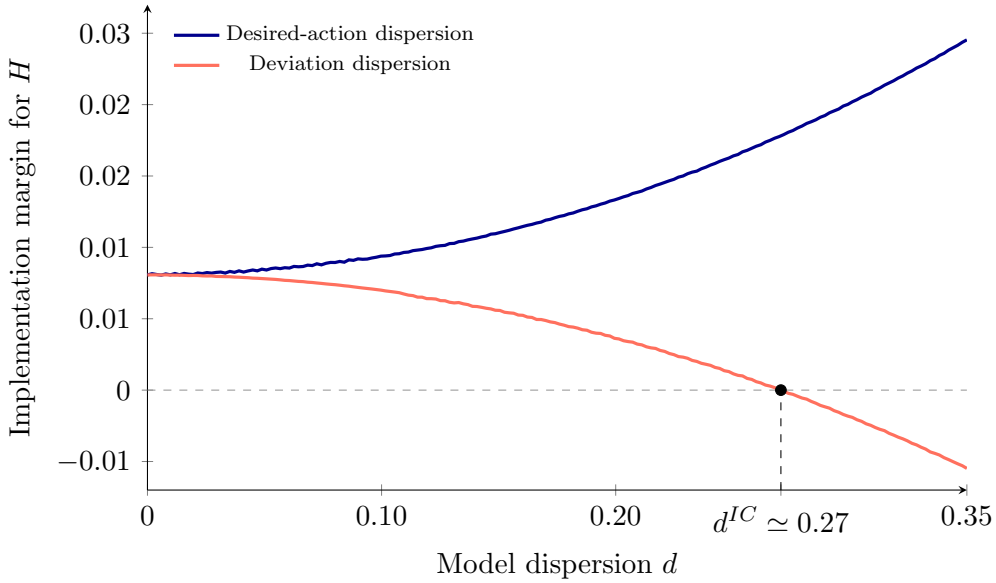

Figure \ref{fig:library-implementation} isolates the asymmetry in Theorem \ref{thm:library}. Since $\psi(p)$ is strictly convex, spreading two forecasts apart while holding their mean fixed raises their average robust value. When the desired-action models become more dispersed, this raises the value of $H$ and the implementation margin increases. When the same dispersion occurs around $L$, it instead raises the value of the deviation: the margin falls and reaches zero at $d^{IC}\simeq0.27$, after which the contract no longer implements $H$. Thus,  the effect of disagreement across models on incentives depends on which action the disagreement concerns.

\par The variance expansion below gives a local interpretation of Theorem \ref{thm:library}.

\begin{obs}\label{prop:variance}
Fix a finite model library $\mu^a$ and a contract $w$. As $\lambda_P\downarrow0$,
$$
\widehat V_P^a(w;\mu^a)
=
\E_{\bar q^a}[y-w]
-\frac{\lambda_P}{2}\E_{\mu^a}[\Var_q(y-w)]
+O(\lambda_P^2),
$$
and as $\lambda_A\downarrow0$,
$$
\widehat V_A^a(w;\mu^a)
=
\E_{\bar q^a}[w]-c(a)
-\frac{\lambda_A}{2}\E_{\mu^a}[\Var_q(w)]
+O(\lambda_A^2).
$$
Moreover, if $\tilde\mu^a\succeq\mu^a$, then $\E_{\tilde\mu^a}[\Var_q(x)]\leq\E_{\mu^a}[\Var_q(x)]$ for every payoff vector $x:\Y\to\R$.
\end{obs}

Observation \ref{prop:variance} makes the mechanism behind Theorem \ref{thm:library} more transparent. Holding the mean model fixed, greater dispersion across reference models lowers average within-model variance. Since multiplier preferences penalize within-model payoff variance locally, this raises the ARC valuation of a fixed payoff for either party. For the desired action, the same reallocation of uncertainty therefore benefits the agent's evaluation of wages and the principal's evaluation of residual output. For a deviation, only the former effect matters for implementation, which makes the deviation harder to deter.

\subsubsection{Comparison: smooth ambiguity}\label{subsec:smooth-library}

The second-order belief used to motivate the smooth Bayesian criterion also makes clear that Theorem \ref{thm:library} is not a general implication of ambiguity aversion. To formalize this, consider instead the smooth-ambiguity criterion used by \citet{kellner15,kellner17}. Let $\phi_A:\R\to\R$ be continuous, strictly increasing, and concave, and define the agent's smooth-ambiguity value from action $a$ by $U_{A,S}^a(w;\mu^a):=\E_{\mu^a}[\phi_A(\E_q[w]-c(a))]$. Let $X_S^{a^*}(\mu)$ be the set of contracts satisfying $U_{A,S}^{a^*}(w;\mu^{a^*})\ge \phi_A(U_0)$ and $U_{A,S}^{a^*}(w;\mu^{a^*})\ge U_{A,S}^{a}(w;\mu^{a})$  $\forall a\ne a^*$.

\begin{obs}\label{prop:smooth-library-dispersion}
Fix $a^*$. If $\tilde\mu^{a^*}\succeq\mu^{a^*}$ and all deviation libraries are unchanged, then $X_S^{a^*}(\tilde\mu)\subseteq X_S^{a^*}(\mu)$. If instead $\tilde\mu^{\hat a}\succeq\mu^{\hat a}$ for some deviation $\hat a\ne a^*$ and all other libraries are unchanged, then $X_S^{a^*}(\mu)\subseteq X_S^{a^*}(\tilde\mu)$. Moreover, if the principal also has a continuous, strictly increasing, concave smooth-ambiguity aggregator $\phi_P$, then dispersion in the desired-action library weakly lowers her evaluation of every fixed contract: $\E_{\tilde\mu^{a^*}}[\phi_P(\E_q[y-w])]\le \E_{\mu^{a^*}}[\phi_P(\E_q[y-w])]$ for every $w$.
\end{obs}

The implementation and contract valuation effects are reversed in Observation \ref{prop:smooth-library-dispersion}. Under smooth ambiguity, $q\mapsto\E_q[w]$ is affine and the concavity of $\phi_A$ makes $q\mapsto\phi_A(\E_q[w]-c(a))$ concave. A more dispersed desired-action library therefore lowers the agent's valuation, tightening participation and incentive compatibility, while dispersion in a deviation library lowers the value of that deviation and makes implementation easier. The principal's evaluation sharpens this contrast. Under smooth Bayesian criteria, $q\mapsto \mathcal W_{\lambda_P}(q,y-w)$ is convex, so dispersion in the desired-action library weakly raises the principal's robust payoff from every fixed contract. In contrast, with smooth ambiguity, $q\mapsto\phi_P(\E_q[y-w])$ is concave, so the same dispersion weakly lowers her evaluation. Thus, holding the mean model fixed, dispersion in the desired-action library is favorable to both parties under smooth Bayesian criteria but unfavorable to both under smooth ambiguity.

The difference comes from where curvature enters the two criteria. Smooth Bayesian criteria apply a robust transformation within each reference model and then average robust valuations linearly across models. Smooth ambiguity instead averages a concave transformation of model-specific expected payoffs. The opposite comparative statics are therefore not a general consequence of having multiple probability models; they reflect  different attitudes toward how uncertainty across reference models should be aggregated.

\subsubsection{Application: endogenous model dispersion}\label{subsec:endogenous-dispersion}

The dispersion of a model library may itself reflect an organizational choice. \citet{linear20} develop a framework where a firm estimates the parameters governing the output distribution and represents its statistical uncertainty by a confidence region around a nominal estimate. The size of this region is endogenous: a wider region gives the firm greater statistical confidence that the relevant parameter is covered, but exposes its contract to a larger set of possible models. We use their confidence-region technology to study how this choice relates to model dispersion and how it changes due to misspecification.

We specialize their framework to the binary-output environment in Section \ref{subsec:twooutcome}. Let $e_H>e_L>0$ denote the effort levels corresponding to $H$ and $L$, and write the existing nominal success probabilities as $p_j=\bar\theta e_j$, $j\in\{H,L\}$, where $\bar\theta>0$ is the nominal productivity parameter. In \citet{linear20}, the firm is uncertain about this common parameter: for a confidence radius $\delta\in[0,\bar\delta]$, it considers $\Theta(\delta):=[\bar\theta-\delta,\bar\theta+\delta]$, and success under action $j$ and parameter $\theta$ occurs with probability $\theta e_j$. The same uncertain productivity parameter therefore governs both actions: increasing $\delta$ cannot selectively change uncertainty about $H$ without also changing uncertainty about $L$.\footnote{\citet{linear20} allow several outputs and an ellipsoidal region for a vector of productivity parameters. The interval $\Theta(\delta)$ is its one-dimensional specialization. We take $\bar\delta$ small enough that every induced probability has full support. Their baseline also imposes limited liability, but Remark 4 shows that their main robust-contract conclusions continue to hold with a participation constraint instead. }

Under \citeauthor{linear20}'s (\citeyear{linear20}) worst-case criterion, this common uncertainty makes high effort more costly to implement. Specializing their Proposition 2 to $\Y=\{0,1\}$, optimized worst-case performance from inducing $H$ falls at rate $e_H$ as $\delta$ increases. If $\chi(\delta)$ denotes the increasing confidence level associated with the region and $K\geq0$ its value to the firm, their confidence choice therefore reduces, up to terms independent of $\delta$, to
$$
\delta^{WC}\in\argmax_{\delta\in[0,\bar\delta]}
\big\{K\chi(\delta)-e_H\delta\big\}.
$$
Greater coverage and contractual performance thus move in opposite directions. In particular, absent a value of greater confidence, the worst-case principal chooses $\delta^{WC}=0$.

To introduce misspecification concerns using ARC, we define a probability distribution over candidate models. Fix parameter perturbations $-1=z_1<\cdots<z_m=1$ with probabilities $\nu_k>0$, where $\sum_k\nu_k=1$ and $\sum_k\nu_kz_k=0$. For each $\delta$, let $\mu_\delta^j$ assign probability $\nu_k$ to the model with success probability
$q_\delta^{j,k}(1)=e_j(\bar\theta+\delta z_k).
$
The standardized scenarios and their relative weights remain fixed as the confidence radius changes; $\delta$  scales only their distance from the nominal parameter.\footnote{This is the probabilistic counterpart of varying the radius of the centered confidence region in \citet{linear20}. The centering restriction $\sum_k\nu_kz_k=0$ ensures that $\E_{\mu_\delta^j}[q]=q^j$ for every $\delta$, so the exercise changes model dispersion without changing the nominal action-output model. Including the endpoints $\{-1,1\}$ also makes the convex hull of the candidate parameters equal to $\Theta(\delta)$.}  Recall that $\widehat X^{a}(\mu)$ denotes the implementation set for action $a$ and $\widehat\Lambda_P^a(\mu)$ denotes the principal's optimized robust payoff from action $a$. Let $\mu(\delta):=(\mu_\delta^H,\mu_\delta^L)$ and suppose $H$ is implementable at $\delta=0$.

\begin{proposition}\label{prop:endogenous-dispersion}
If $0\leq\delta<\delta'\leq\bar\delta$, then
$$
\mu_{\delta'}^j\succeq\mu_\delta^j
\quad\text{for }j\in\{H,L\},\qquad
\widehat X^H(\mu(\delta))
\subseteq
\widehat X^H(\mu(\delta')),
\qquad
\widehat\Lambda_P^H(\mu(\delta'))
\geq
\widehat\Lambda_P^H(\mu(\delta)).
$$
Consequently, if $\chi$ is weakly increasing, $\bar\delta$ solves
$\max_{\delta\in[0,\bar\delta]}
\big\{\widehat\Lambda_P^H(\mu(\delta))+K\chi(\delta)\big\}$.
If $K>0$ and $\chi$ is strictly increasing, this solution is unique.
\end{proposition}

This result reveals that a wider confidence region  has opposite contractual implications under the two decision criteria. Under worst-case contracting, the contract must remain incentive compatible at the least favorable productivity parameter. Enlarging the region lowers that parameter, requiring stronger incentives and reducing the principal's guaranteed payoff. Greater statistical confidence consequently comes at a price.

Under ARC, widening the same confidence region spreads candidate models around a fixed distribution. Since productivity multiplies effort, a given parameter perturbation moves the success probability more under $H$ than under $L$. The convexity of robust compensation values then makes the agent's valuation of $H$ rise faster than his valuation of $L$, so the common increase in model dispersion relaxes the incentive constraint. Dispersion around the implemented action also raises the principal's ARC valuation. Greater coverage and contractual performance therefore reinforce rather than offset each other.

The difference between these two decision criteria changes how the principal manages model uncertainty. Under the worst-case criterion, more precise information can be valuable because it shrinks the set against which the contract must be protected. Under ARC, eliminating credible models can instead destroy useful dispersion: when the same underlying parameter affects productive effort more strongly than shirking, a broader model library both strengthens incentives and improves the principal's valuation of the relationship. Thus, the way an organization aggregates plausible models can determine whether model uncertainty is a feature that should be eliminated or preserved.

\section{Conclusion}\label{sec:conclusion}

We analyze agency problems where the probability model linking action to output is viewed as an approximation. The first-best analysis isolates a contractual role for performance pay that is absent from the classical principal-agent benchmark. When both parties evaluate misspecification using proportional variational costs, an optimal contract divides output linearly. Its slope allocates misspecification exposure according to the parties' relative concerns. At the first best, the parties must have some common effective models, which implies that the resulting allocation is Pareto efficient. This sharing rule also raises the value of the agency relationship relative to assigning the entire output claim to either party and can change which observable action is first-best optimal.

Moral hazard distorts the first-best effective-model alignment by the marginal incentive wedges associated with active deviations; under multiplier preferences, this distortion becomes a nonlinear likelihood-ratio correction to the first-best linear sharing rule. With a continuum of actions, misspecification concerns can also change the geometry of implementation: robust compensation may become sufficiently convex that intermediate actions cannot be implemented and performance pay produces discrete changes in effort. When a set of plausible models is available for each action, the location of uncertainty matters as well. Greater dispersion around the desired action expands the implementation set and can raise the principal's robust payoff, whereas dispersion around deviations makes implementation more difficult and can lower that payoff. Endogenizing dispersion  further shows that the criterion used to aggregate plausible models can determine whether the principal wants to eliminate model uncertainty or preserve it.

Taken together, our results suggest that introducing model misspecification in agency problems changes more than the level of contracting costs. It changes what performance pay is used for, how incentive contracts should be decomposed, and which actions can be implemented. The linear first-best contract is especially useful for separating these effects: it identifies the part of performance sensitivity that can be attributed to sharing exposure to misspecification before incentive frictions are introduced.

\appendix

\section{Appendix: Proofs}\label{app:proofs}

\subsection{Variational and multiplier properties}\label{app:prelim}

We begin by stating several properties that will be used repeatedly in the proofs. Fix an action $a$ and suppress the superscript $a$ when no confusion can arise. Because $\Y$ is finite, $\DeltaY$ is a nonempty compact subset of a finite-dimensional Euclidean space. Since $\Gamma$ is lower semicontinuous and takes values in $[0,\infty]$, for every finite payoff vector $x\in\R^{\Y}$ and every $\lambda>0$ the map $p\mapsto \E_p[x]+\Gamma(p)/\lambda$ is lower semicontinuous on $\DeltaY$. It therefore attains its minimum. Hence, $\Pi_\lambda(x)$ is nonempty. 

For $\lambda>0$, the functional $\mathcal V_\lambda$ is concave in $x$. Indeed, for any $x,x'\in\R^{\Y}$ and $t\in[0,1]$, we have $\E_p[tx+(1-t)x']+\Gamma(p)/\lambda=t\E_p[x]+(1-t)\E_p[x']+\Gamma(p)/\lambda$ $\forall p\in\DeltaY$. Since $\Gamma(p)/\lambda=t\Gamma(p)/\lambda+(1-t)\Gamma(p)/\lambda$, taking the minimum over $p$ gives the inequality $\mathcal V_\lambda(tx+(1-t)x')\ge t\mathcal V_\lambda(x)+(1-t)\mathcal V_\lambda(x')$. 

The variational value is translation invariant. For every $r\in\R$, every $\lambda>0$, and every $x$,  $\mathcal V_\lambda(x+r\m{1})=\mathcal V_\lambda(x)+r$, because $\E_p[x+r\m{1}]=\E_p[x]+r$ $\forall p\in \DeltaY$.  For every $\gamma>0$ and $\lambda>0$, it also satisfies the scaling identity $\mathcal V_\lambda(\gamma x)=\gamma\mathcal V_{\gamma\lambda}(x)$. To verify this identity, factor the positive scalar $\gamma$ out of the objective: $\E_p[\gamma x]+\Gamma(p)/\lambda=\gamma\{\E_p[x]+\Gamma(p)/(\gamma\lambda)\}$ for every $p$; minimizing both sides over the same set $\DeltaY$ yields the claim.

We will also use the relation between effective models and derivatives. Suppose $\lambda>0$ and $\Pi_\lambda(x)=\{p^*\}$ is a singleton. For a direction $h\in\R^{\Y}$ and scalar $t$, define $f(t,p):=\E_p[x+th]+\Gamma(p)/\lambda$. The feasible set $\DeltaY$ is compact, $f$ is affine and hence differentiable in $t$, and the minimizer at $t=0$ is unique. The directional derivative is $D\mathcal V_\lambda(x)[h]=\E_{p^*}[h]=p^*\cdot h$. Since this is linear in $h$, $\mathcal V_\lambda$ is differentiable at $x$ with gradient $p^*$. 
Consequently, at every smooth contract in the sense of Section \ref{subsec:MH-contract}, we have $\nabla_wV_A^a(w)=\pi_A^a(\cdot;w)$ and $\nabla_wV_P^a(w)=-\pi_P^a(\cdot;w)$.

For multiplier preferences, the preceding properties can be made explicit. Fix a full-support model $q$ and $\lambda>0$. For every finite payoff vector $x$, $Z(x):=\E_q[e^{-\lambda x(y)}]$ is finite and strictly positive. Thus, $M(q,x;\lambda)=-\lambda^{-1}\log Z(x)$ is finite. 

\subsection{Proofs for Section \ref{sec:firstbest}}\label{app:FB}

\begin{proof}[Proof of Lemma \ref{lem:FB-participation}]
Fix an action $a$ and let $w$ solve $(\mathsf{FB}_a)$. Suppose, toward a contradiction, that the participation constraint is slack, so $V_A^a(w)>U_0$. Set $\delta:=V_A^a(w)-U_0>0$ and choose any $\varepsilon\in(0,\delta)$. Consider the translated contract $\widetilde w:=w-\varepsilon\m{1}$.

By translation invariance of $\mathcal V_{\lambda_A}^a$, $V_A^a(\widetilde w)=\mathcal V_{\lambda_A}^a(w-\varepsilon\m{1})-c(a)=V_A^a(w)-\varepsilon>U_0$. Hence, $\widetilde w$ is feasible. The principal's residual payoff under $\widetilde w$ is $y-\widetilde w=(y-w)+\varepsilon\m{1}$. Translation invariance of $\mathcal V_{\lambda_P}^a$ therefore gives $V_P^a(\widetilde w)=V_P^a(w)+\varepsilon>V_P^a(w)$. This contradicts the optimality of $w$. Thus, no optimum can have slack participation, and every solution satisfies $V_A^a(w)=U_0$.
\end{proof}

\begin{proof}[Proof of Theorem \ref{thm:firstbest-linear}]
Fix $a$ and suppress the superscript $a$ throughout the proof. Write $B:=U_0+c(a)$, so that the participation constraint is $\mathcal V_{\lambda_A}(w)\ge B$.

 Set $\lambda_{AP}:=\lambda_A\lambda_P/(\lambda_A+\lambda_P)$. We begin with an upper bound that holds for every feasible contract. Fix any $w\in\R^{\Y}$ and any $p\in\DeltaY$. By the definition of a minimum, evaluating the agent's variational problem at this particular $p$ gives $\mathcal V_{\lambda_A}(w)\le \E_p[w]+\Gamma(p)/\lambda_A$. Likewise, evaluating the principal's variational problem at the same $p$ gives $\mathcal V_{\lambda_P}(y-w)\le \E_p[y-w]+\Gamma(p)/\lambda_P$. Adding these two inequalities cancels the transfer state by state and yields $\mathcal V_{\lambda_A}(w)+\mathcal V_{\lambda_P}(y-w)\le \E_p[y]+(1/\lambda_A+1/\lambda_P)\Gamma(p)$. Since $1/\lambda_{AP}=1/\lambda_A+1/\lambda_P$, the right-hand side equals $\E_p[y]+\Gamma(p)/\lambda_{AP}$. The inequality holds for every $p\in\DeltaY$, so minimizing its right-hand side over $p$ gives $\mathcal V_{\lambda_A}(w)+\mathcal V_{\lambda_P}(y-w)\le \mathcal V_{\lambda_{AP}}(y)$.

Now, let $w$ be feasible. Then, $\mathcal V_{\lambda_A}(w)\ge B$, and the preceding inequality implies $\mathcal V_{\lambda_P}(y-w)\le \mathcal V_{\lambda_{AP}}(y)-\mathcal V_{\lambda_A}(w)\le \mathcal V_{\lambda_{AP}}(y)-B$. Thus, $\mathcal V_{\lambda_{AP}}(y)-B$ is an upper bound on the principal's value over the entire feasible set.

We next construct a feasible contract attaining this upper bound. Let $s:=\lambda_P/(\lambda_A+\lambda_P)$, and set $w^*(y):=\alpha+sy$, where $\alpha:=B-s\mathcal V_{\lambda_{AP}}(y)$. Notice that $s\lambda_A=\lambda_{AP}$ and $(1-s)\lambda_P=\lambda_{AP}$. Translation invariance and the scaling identity from Appendix \ref{app:prelim} give $\mathcal V_{\lambda_A}(w^*)=\alpha+\mathcal V_{\lambda_A}(sy)=\alpha+s\mathcal V_{s\lambda_A}(y)=\alpha+s\mathcal V_{\lambda_{AP}}(y)=B$. Hence, $w^*$ satisfies participation with equality. The principal's residual payoff is $y-w^*=-\alpha+(1-s)y$, so the same two identities give $\mathcal V_{\lambda_P}(y-w^*)=-\alpha+(1-s)\mathcal V_{(1-s)\lambda_P}(y)=-\alpha+(1-s)\mathcal V_{\lambda_{AP}}(y)=\mathcal V_{\lambda_{AP}}(y)-B$. Therefore, $w^*$ attains the upper bound derived above and is optimal. Replacing $s$ by its definition yields the stated slope.

For this slope, the intercept is uniquely determined by participation. Indeed, translation invariance gives $\mathcal V_{\lambda_A}(\alpha\m{1}+sy)=\alpha+\mathcal V_{\lambda_A}(sy)$. Therefore, the equation $\mathcal V_{\lambda_A}(\alpha\m{1}+sy)=B$ has exactly one solution in $\alpha$, namely the value constructed above. Lemma \ref{lem:FB-participation} shows that every optimal contract with this slope must use that intercept.
\end{proof}

\begin{proof}[Proof of Corollary \ref{cor:effective-models}]
Fix an action $a$, and let $B:=U_0+c(a)$ and $\lambda_{AP}:=\lambda_A\lambda_P/(\lambda_A+\lambda_P)$. Let $w$ be any first-best optimum. By Lemma \ref{lem:FB-participation}, $\mathcal V_{\lambda_A}^a(w)=B$. By Theorem \ref{thm:firstbest-linear}, the optimal principal value equals $\mathcal V_{\lambda_{AP}}^a(y)-B$. Hence, at $w$, $\mathcal V_{\lambda_A}^a(w)+\mathcal V_{\lambda_P}^a(y-w)=\mathcal V_{\lambda_{AP}}^a(y)$. Choose any $p^*\in\Pi_{\lambda_{AP}}^a(y)$; such a model exists by Appendix \ref{app:prelim}. By definition of the variational minima, $\mathcal V_{\lambda_A}^a(w)\le \E_{p^*}[w]+\Gamma^a(p^*)/\lambda_A$ and $\mathcal V_{\lambda_P}^a(y-w)\le \E_{p^*}[y-w]+\Gamma^a(p^*)/\lambda_P$. The two right-hand sides add to $\E_{p^*}[y]+(1/\lambda_A+1/\lambda_P)\Gamma^a(p^*)=\mathcal V_{\lambda_{AP}}^a(y)$ because $p^*$ minimizes the pooled problem. The two left-hand sides also add to $\mathcal V_{\lambda_{AP}}^a(y)$ by the above. Therefore, the sum of the two weak inequalities is an equality. Neither inequality can then be strict: if, for example, the first were strict, adding the weak second inequality would make the sum of the left-hand sides strictly smaller than the sum of the right-hand sides. Thus, both inequalities hold with equality. Equality in the first says precisely that $p^*\in\Pi_A^a(w)$, and equality in the second says that $p^*\in\Pi_P^a(w)$. Hence, $\Pi_A^a(w)\cap\Pi_P^a(w)$ is nonempty. In fact, the argument shows the stronger inclusion $\Pi_{\lambda_{AP}}^a(y)\subseteq\Pi_A^a(w)\cap\Pi_P^a(w)$ for every first-best optimum $w$.

Now consider the particular linear solution in Theorem \ref{thm:firstbest-linear}. Let $s:=\lambda_P/(\lambda_A+\lambda_P)$  and write $w_a^{FB}=\alpha_a\m{1}+sy$. Translation does not affect an argmin, so $\Pi_A^a(w_a^{FB})$ is the set of minimizers of $p\mapsto s\E_p[y]+\Gamma^a(p)/\lambda_A$. Since $s>0$, multiplying this objective by $1/s$ does not change its minimizers. Because $s\lambda_A=\lambda_{AP}$, those minimizers are exactly $\Pi_{\lambda_{AP}}^a(y)$. Likewise, $y-w_a^{FB}=-\alpha_a\m{1}+(1-s)y$, and $\Pi_P^a(w_a^{FB})$ is the set of minimizers of $p\mapsto(1-s)\E_p[y]+\Gamma^a(p)/\lambda_P$. Dividing by $1-s>0$ and using $(1-s)\lambda_P=\lambda_{AP}$ shows that this set is also $\Pi_{\lambda_{AP}}^a(y)$. Therefore, $\Pi_A^a(w_a^{FB})=\Pi_P^a(w_a^{FB})$, as claimed.
\end{proof}

\begin{proof}[Proof of Corollary \ref{lem:FB-pareto}]
Fix an action $a$ and let $w$ be any first-best optimum. By Lemma
\ref{lem:FB-participation}, $V_A^a(w)=U_0$. Suppose, toward a
contradiction, that $w$ is not Pareto efficient. Then, there exists a
contract $\widetilde w$ such that
$V_A^a(\widetilde w)\ge V_A^a(w)$ and
$V_P^a(\widetilde w)\ge V_P^a(w)$, with at least one inequality strict. Since $V_A^a(\widetilde w)\ge V_A^a(w)=U_0$, the contract
$\widetilde w$ is feasible for $(\mathsf{FB}_a)$. Optimality of $w$
therefore implies $V_P^a(\widetilde w)\le V_P^a(w)$. Hence,
$V_P^a(\widetilde w)=V_P^a(w)$, and Pareto dominance must be strict
for the agent: $V_A^a(\widetilde w)>U_0$.

Let $\delta:=V_A^a(\widetilde w)-U_0>0$, choose
$\varepsilon\in(0,\delta)$, and define
$\widehat w:=\widetilde w-\varepsilon\m{1}$. Translation invariance
gives
$V_A^a(\widehat w)=V_A^a(\widetilde w)-\varepsilon>U_0$, so
$\widehat w$ remains feasible. It also gives
$V_P^a(\widehat w)=V_P^a(\widetilde w)+\varepsilon
=V_P^a(w)+\varepsilon>V_P^a(w)$.
This contradicts the optimality of $w$. Therefore, every first-best
optimal contract is Pareto efficient.
\end{proof}

\begin{proof}[Proof of Corollary \ref{cor:firstbest-unique}]
Fix $a$ and let $p^*\in\operatorname{ri}\DeltaY$ be a minimizer of $p\mapsto\E_p[y]+(1/\lambda_A+1/\lambda_P)\Gamma^a(p)$ at which $\Gamma^a$ is differentiable. Let $w$ be an arbitrary first-best optimum. We will show that $w$ must equal the linear solution in Theorem \ref{thm:firstbest-linear}.

By the stronger conclusion established in the proof of Corollary \ref{cor:effective-models}, every pooled minimizer is an effective model for both parties at every first-best optimum. Thus, $p^*$ minimizes $p\mapsto\E_p[w]+\Gamma^a(p)/\lambda_A$ and also minimizes $p\mapsto\E_p[y-w]+\Gamma^a(p)/\lambda_P$.

Because $p^*$ lies in the relative interior of the simplex, none of the nonnegativity constraints on probabilities binds at $p^*$. The feasible first-order directions are therefore exactly the tangent vectors $d\in\R^{\Y}$ satisfying $\m{1}\cdot d=0$. Since $\Gamma^a$ is differentiable at $p^*$, first-order optimality of $p^*$ in the agent's minimization gives $d\cdot[w+(1/\lambda_A)\nabla\Gamma^a(p^*)]=0$ for every such tangent vector $d$. A vector is orthogonal to every tangent vector of the simplex if and only if it is proportional to $\m{1}$. Hence, there is a scalar $k_A$ such that $w+(1/\lambda_A)\nabla\Gamma^a(p^*)=k_A\m{1}$. Applying the same argument to the principal's minimization gives a scalar $k_P$ such that $y-w+(1/\lambda_P)\nabla\Gamma^a(p^*)=k_P\m{1}$.

Adding these two vector equalities gives $y+(1/\lambda_A+1/\lambda_P)\nabla\Gamma^a(p^*)=(k_A+k_P)\m{1}$. Since $(1/\lambda_A+1/\lambda_P)^{-1}=\lambda_A\lambda_P/(\lambda_A+\lambda_P)$, we can solve this equality for $\nabla\Gamma^a(p^*)$ and substitute the result back into the agent's first-order condition. The coefficient on $y$ that results is $(1/\lambda_A)/(1/\lambda_A+1/\lambda_P)=\lambda_P/(\lambda_A+\lambda_P)$. Therefore, for some scalar $\alpha$, $w=\alpha\m{1}+[\lambda_P/(\lambda_A+\lambda_P)]y$.

Finally, Lemma \ref{lem:FB-participation} implies that the agent's participation constraint binds at $w$. For a contract with the displayed slope, translation invariance makes the agent's value strictly increase one-for-one with $\alpha$. Hence, exactly one value of $\alpha$ can satisfy participation with equality. Theorem \ref{thm:firstbest-linear} constructs that value. Thus, every first-best optimum coincides with the linear solution from Theorem \ref{thm:firstbest-linear}, proving uniqueness.
\end{proof}

\begin{proof}[Proof of Proposition \ref{prop:robustness-pooling}]
Fix an observable action $a$ and write $B:=U_0+c(a)$ and $\lambda_{AP}:=\lambda_A\lambda_P/(\lambda_A+\lambda_P)$. The proof of Theorem \ref{thm:firstbest-linear} established both an upper bound and a contract attaining it. It therefore established directly that the optimized first-best value for this action is $\mathcal V_{\lambda_{AP}}^a(y)-B=\mathcal V_{\lambda_{AP}}^a(y)-c(a)-U_0$. Since $U_0$ is independent of $a$, maximizing the principal's optimized value over observable actions is equivalent to maximizing $\mathcal V_{\lambda_{AP}}^a(y)-c(a)$. This proves the first two claims.

We next compare this optimized value with the two polar sharing arrangements. First, consider a constant participating wage. Since a constant payoff has variational value equal to that constant, the unique constant wage satisfying participation with equality is $w(y)=B$ for every $y$. The principal then receives $y-B$, so her value is $\mathcal V_{\lambda_P}^a(y)-B$. Second, consider a contract under which the agent receives the entire stochastic output claim up to a constant transfer, $w(y)=y-s$. Translation invariance gives the agent value $\mathcal V_{\lambda_A}^a(y)-s-c(a)$. Binding participation therefore requires $s=\mathcal V_{\lambda_A}^a(y)-B$. The principal's residual payoff is the constant $s$, so her value is $s=\mathcal V_{\lambda_A}^a(y)-B$.

It remains to compare the three variational values. We first show that, for a fixed payoff vector $x$, $\lambda\mapsto\mathcal V_\lambda^a(x)$ is weakly decreasing on $(0,\infty)$. Let $0<\lambda_1<\lambda_2$ and choose $p_1\in\Pi_{\lambda_1}^a(x)$. Since $1/\lambda_2<1/\lambda_1$ and $\Gamma^a(p_1)\ge0$, evaluating the $\lambda_2$ problem at $p_1$ gives $\mathcal V_{\lambda_2}^a(x)\le\E_{p_1}[x]+\Gamma^a(p_1)/\lambda_2\le\E_{p_1}[x]+\Gamma^a(p_1)/\lambda_1=\mathcal V_{\lambda_1}^a(x)$. Because $\lambda_{AP}<\lambda_A$ and $\lambda_{AP}<\lambda_P$, it follows that $\mathcal V_{\lambda_{AP}}^a(y)\ge\mathcal V_{\lambda_A}^a(y)$ and $\mathcal V_{\lambda_{AP}}^a(y)\ge\mathcal V_{\lambda_P}^a(y)$. Subtracting $B$ proves both weak comparisons in the statement.

Suppose now that $y$ is nonconstant and the condition in Corollary \ref{cor:firstbest-unique} holds. Let $p^*$ be the interior differentiability point appearing in that condition. We first prove that $\Gamma^a(p^*)>0$. If instead $\Gamma^a(p^*)=0$, then $p^*$ is a global minimizer of the nonnegative convex function $\Gamma^a$. Since $p^*$ is in the relative interior and $\Gamma^a$ is differentiable there, its derivative in every tangent direction must vanish. Equivalently, $\nabla\Gamma^a(p^*)$ is proportional to $\m{1}$. But $p^*$ also minimizes $p\mapsto\E_p[y]+\Gamma^a(p)/\lambda_{AP}$. The first-order condition for this minimization at an interior point then requires $y+(1/\lambda_{AP})\nabla\Gamma^a(p^*)$ to be proportional to $\m{1}$. Since the second term is already proportional to $\m{1}$, this would imply that $y$ itself is constant on $\Y$, contrary to the hypothesis. Hence, $\Gamma^a(p^*)>0$.

Now evaluate the $\lambda_A$ criterion at $p^*$. Because $p^*$ minimizes the pooled criterion, $\mathcal V_{\lambda_{AP}}^a(y)=\E_{p^*}[y]+\Gamma^a(p^*)/\lambda_{AP}$. Since $\lambda_{AP}<\lambda_A$ and $\Gamma^a(p^*)>0$, $\E_{p^*}[y]+\Gamma^a(p^*)/\lambda_A<\mathcal V_{\lambda_{AP}}^a(y)$. The $\lambda_A$ value is no larger than the left-hand side because it minimizes over all $p$, so $\mathcal V_{\lambda_A}^a(y)<\mathcal V_{\lambda_{AP}}^a(y)$. The identical argument with $\lambda_P$ gives $\mathcal V_{\lambda_P}^a(y)<\mathcal V_{\lambda_{AP}}^a(y)$. Subtracting $B$ proves that both gains from optimal sharing are strict.
\end{proof}

\begin{proof}[Proof of Proposition \ref{prop:different-reference-models}]
Fix action $a$ and suppress the action indices throughout to ease the notation. Write $B:=U_0+c(a)$. The agent's value is $M(q_A,w;\lambda_A)-c(a)$ and the principal's value is $M(q_P,y-w;\lambda_P)$. Both reference models have full support and both $\lambda$'s are strictly positive.

We first establish that any optimum must satisfy participation with equality. The proof is identical to Lemma \ref{lem:FB-participation}: if participation were slack, lowering every wage by a sufficiently small common constant would preserve feasibility, lower the agent's value by exactly that constant, and raise the principal's value by exactly the same constant. Thus, at every optimum, $M(q_A,w;\lambda_A)=B$.

We next derive a necessary condition for every optimum. Consider the equivalent convex-optimization formulation that minimizes $-M(q_P,y-w;\lambda_P)$ subject to $B-M(q_A,w;\lambda_A)\le0$. Both functions in this formulation are convex because multiplier values are concave. Slater's condition holds: for example, adding a sufficiently large positive constant to every wage makes $M(q_A,w;\lambda_A)>B$. Hence, the KKT conditions are necessary at every optimum. Since both multiplier values are differentiable, there is a multiplier $\eta\ge0$ such that stationarity is $-\pi_P(y;w)+\eta\pi_A(y;w)=0$ for every $y\in\Y$, where $\pi_A$ and $\pi_P$ are the two multiplier effective models. Summing this equality over $y$ gives $-1+\eta=0$, because both effective models are probability distributions. Thus, $\eta=1$ and every optimum satisfies $\pi_A(\cdot;w)=\pi_P(\cdot;w)$.

Write $Z_A:=\E_{q_A}[e^{-\lambda_Aw}]$ and $Z_P:=\E_{q_P}[e^{-\lambda_P(y-w)}]$. Equality of the effective models gives $q_A(y)e^{-\lambda_Aw(y)}/Z_A=q_P(y)e^{-\lambda_P[y-w(y)]}/Z_P$ for every output. Every quantity in this equality is strictly positive, so we may take logarithms. Rearranging the resulting equality gives $(\lambda_A+\lambda_P)w(y)=\lambda_Py+\log[q_A(y)/q_P(y)]+\log(Z_P/Z_A)$. Therefore, every optimum must have the form presented in the statement, with the output-independent term $\log(Z_P/Z_A)/(\lambda_A+\lambda_P)$ absorbed into an intercept $\alpha$. Binding participation determines that intercept uniquely because the agent's multiplier value changes one-for-one with a common wage translation.

It remains to show that a contract of this form actually exists and is globally optimal. Start from the output-dependent part of the displayed formula and choose the unique intercept $\alpha$ that makes $M(q_A,w;\lambda_A)=B$. The equality just derived can then be read in reverse: for this contract, the ratio between the two effective-model formulas is independent of $y$, and normalization of both probability distributions forces $\pi_A=\pi_P$. Hence, the KKT stationarity condition holds with multiplier $\eta=1$, and participation binds. These KKT conditions are sufficient for global optimality. To see this, let $w^*$ denote the constructed contract and let $w$ be any feasible contract. Concavity and differentiability of the principal's value imply $M(q_P,y-w;\lambda_P)\le M(q_P,y-w^*;\lambda_P)-\pi_P(\cdot;w^*)\cdot(w-w^*)$. Concavity of the agent's value implies $M(q_A,w;\lambda_A)\le M(q_A,w^*;\lambda_A)+\pi_A(\cdot;w^*)\cdot(w-w^*)$. Since $w$ is feasible and $w^*$ satisfies participation with equality, the second inequality implies $\pi_A(\cdot;w^*)\cdot(w-w^*)\ge M(q_A,w;\lambda_A)-M(q_A,w^*;\lambda_A)\ge0$. Using $\pi_A=\pi_P$ in the first inequality gives $M(q_P,y-w;\lambda_P)\le M(q_P,y-w^*;\lambda_P)$. Thus, $w^*$ is globally optimal.

Finally, every optimum has already been shown to possess the displayed output-dependent form, and binding participation uniquely fixes its intercept. Hence, the optimum is unique. The formula is affine in output if and only if $y\mapsto\log[q_A(y)/q_P(y)]$ is affine on $\Y$, which proves the final statement.
\end{proof}

\begin{proof}[Proof of Proposition \ref{prop:different-divergences}]
Let $\mathcal W_i(x):=\min_{p\in\DeltaY}\{\E_p[x]+\Gamma_i(p)\}$ for $i\in\{A,P\}$. Define $F(p):=\Gamma_A(p)/s$ and $G(p):=\Gamma_P(p)/(1-s)$. Since $s\in(0,1)$ and both ambiguity indices are continuous and strictly convex, $F$ and $G$ are also continuous and strictly convex on $\DeltaY$, and they are continuously differentiable on $\operatorname{ri}\DeltaY$.

We first prove sufficiency. Suppose $F(p)=G(p)$ $\forall p\in\DeltaY$. Fix an arbitrary payoff vector $x\in\R^\Y$ and an arbitrary intercept $\alpha$, and let the agent receive $w=\alpha\m{1}+sx$. Translation does not affect minimizing models. The agent's criterion for the state-dependent part of his claim is $p\mapsto s\E_p[x]+\Gamma_A(p)=s\big(\E_p[x]+F(p)\big)$. The principal's residual claim is $-\alpha\m{1}+(1-s)x$, and her criterion for its state-dependent part is $p\mapsto(1-s)\E_p[x]+\Gamma_P(p)=(1-s)\big(\E_p[x]+G(p)\big)$. Since $F=G$, these two problems have the same minimizers. Strict convexity gives a unique common minimizer; call it $p^*$.

We show that the resulting allocation is Pareto efficient. Suppose, to the contrary, that there is a transfer perturbation $d\in\R^{\Y}$ such that the agent weakly prefers $w+d$ to $w$, the principal weakly prefers $x-w-d$ to $x-w$, and at least one of these preferences is strict. Since $p^*$ is an effective model for the agent at $w$, evaluating the agent's variational criterion for $w+d$ at $p^*$ gives $\mathcal W_A(w+d)\le\mathcal W_A(w)+p^*\cdot d$. Therefore, $\mathcal W_A(w+d)\ge\mathcal W_A(w)$ implies $p^*\cdot d\ge0$, and strict improvement of the agent would imply $p^*\cdot d>0$. Similarly, because $p^*$ is an effective model for the principal at $x-w$, $\mathcal W_P(x-w-d)\le\mathcal W_P(x-w)-p^*\cdot d$. If $p^*\cdot d>0$, the principal is strictly worse off; if $p^*\cdot d=0$, the principal cannot be strictly better off. Thus, no perturbation can weakly improve both parties and strictly improve one. The allocation is Pareto efficient. Since $x$ and $\alpha$ were arbitrary, the sharing rule is Pareto efficient for every payoff vector and every intercept.

We now prove necessity. Suppose that, for the fixed $s$, every sharing rule $w=\alpha\m{1}+sx$ is Pareto efficient for every $x\in\R^{\Y}$ and every $\alpha\in\R$. We first establish a general implication of Pareto efficiency in the present regular class: the two parties' effective models must coincide. Because $\Gamma_A$ and $\Gamma_P$ are strictly convex, each variational minimization problem has a unique effective model. Let $p_A$ and $p_P$ denote these models at an arbitrary Pareto-efficient sharing allocation. Suppose $p_A\ne p_P$. Set $h:=p_A-p_P$. Then, $p_A\cdot h-p_P\cdot h=\|p_A-p_P\|^2>0$, so $p_A\cdot h>p_P\cdot h$. Choose a scalar $k$ satisfying $-p_A\cdot h<k<-p_P\cdot h$, which is possible because the left endpoint is strictly smaller than the right endpoint. Let $d:=h+k\m{1}$. Since both $p_A$ and $p_P$ sum to one, $p_A\cdot d=p_A\cdot h+k>0$ while $p_P\cdot d=p_P\cdot h+k<0$.

Unique effective models imply differentiability of both variational values by Appendix \ref{app:prelim}. Hence, the derivative of the agent's value in direction $d$ is $p_A\cdot d>0$, and the derivative of the principal's value when her residual claim changes by $-d$ is $-p_P\cdot d>0$. By the definition of a strictly positive directional derivative, there is an $\varepsilon>0$ small enough that giving the agent the additional claim $\varepsilon d$ strictly raises the agent's value and, simultaneously, removing that same claim from the principal strictly raises the principal's value. This is a Pareto improvement, contradicting efficiency. Therefore, Pareto efficiency requires $p_A=p_P$.

Fix now an arbitrary $p\in\operatorname{ri}\DeltaY$. We will use universal efficiency to compare the marginal statistical costs represented by $F$ and $G$ at this arbitrary model. Since $F$ is differentiable at $p$, choose a vector representative of its derivative on the affine hull of the simplex and set $x:=-\nabla F(p)$. We claim that $p$ is the unique minimizer of $q\mapsto q\cdot x+F(q)$. Convexity and differentiability of $F$ at the relative-interior point $p$ imply the supporting-hyperplane inequality $F(q)\ge F(p)+\nabla F(p)\cdot(q-p)$ for every $q\in\DeltaY$. Therefore, $q\cdot x+F(q)\ge p\cdot x+F(p)$. Strict convexity makes this inequality strict whenever $q\ne p$: if equality held at some $q\ne p$, $F$ would be affine on the nondegenerate segment joining $p$ and $q$, contradicting strict convexity. Thus, $p$ is the unique minimizer.

Consider the universally efficient sharing rule $w=\alpha\m{1}+sx$ for this payoff vector $x$. The agent's minimization problem is, up to the positive factor $s$, exactly $q\mapsto q\cdot x+F(q)$, so the agent's unique effective model is $p$. By the common-effective-model implication proved in the preceding paragraph, the principal's unique effective model must also be $p$. Her minimization problem is, up to the positive factor $1-s$, $q\mapsto q\cdot x+G(q)$. Hence, $p$ minimizes this problem as well. Because $p$ is in the relative interior and both $F$ and $G$ are differentiable there, their first-order conditions on the simplex imply that $\nabla F(p)+x$ and $\nabla G(p)+x$ are each proportional to $\m{1}$. Subtracting shows that $\nabla[G-F](p)$ is proportional to $\m{1}$. Equivalently, for every tangent vector $d$ satisfying $\m{1}\cdot d=0$, $d\cdot\nabla[G-F](p)=0$. Since the choice of $p\in\operatorname{ri}\DeltaY$ was arbitrary, the derivative of $G-F$ vanishes in every feasible direction at every point of the relative interior.

The relative interior of a simplex is convex and hence path connected. Take any two points $p_0,p_1$ in the relative interior and define $p(t):=(1-t)p_0+tp_1$. The vector $p_1-p_0$ is tangent to the simplex. The chain rule gives $\frac{d}{dt}[G-F](p(t))=(p_1-p_0)\cdot\nabla[G-F](p(t))=0$ for every $t\in(0,1)$. Hence, $G(p_1)-F(p_1)=G(p_0)-F(p_0)$. Thus, $G-F$ is constant on the relative interior. Continuity of $F$ and $G$ extends this equality to the boundary of $\DeltaY$, so $G=F+C$ on all of $\DeltaY$ for some constant $C$.

Finally, groundedness of $\Gamma_A$ and $\Gamma_P$ implies $\min F=0$ and $\min G=0$. If $G=F+C$, then $\min G=\min F+C=C$, so $C=0$. Therefore, $F=G$, which is exactly $\Gamma_A(p)/s=\Gamma_P(p)/(1-s)$ for every $p\in\DeltaY$. This proves necessity and completes the proof.
\end{proof}

\subsection{Proofs for Section \ref{sec:moralhazard}}\label{app:MH}

\begin{proof}[Proof of Lemma \ref{lem:MH-binding}]
Let $w$ solve $(\mathsf{MH}_{a^*})$. Suppose first that participation is slack, so $V_A^{a^*}(w)>U_0$. Choose $\varepsilon>0$ strictly smaller than $V_A^{a^*}(w)-U_0$ and define $\widetilde w:=w-\varepsilon\m{1}$. Translation invariance gives $V_A^{a^*}(\widetilde w)=V_A^{a^*}(w)-\varepsilon>U_0$, so participation remains satisfied.

For every deviation $a\ne a^*$, translation invariance applies to the agent's value under both actions by exactly the same amount. Hence, $V_A^{a^*}(\widetilde w)-V_A^a(\widetilde w)=[V_A^{a^*}(w)-\varepsilon]-[V_A^a(w)-\varepsilon]=V_A^{a^*}(w)-V_A^a(w)\ge0$. Thus, every incentive constraint remains satisfied. The principal's residual claim rises by the constant $\varepsilon$, so $V_P^{a^*}(\widetilde w)=V_P^{a^*}(w)+\varepsilon$. This contradicts optimality. Therefore, participation binds at every solution.

Now suppose a first-best solution $w^{FB}_{a^*}$ is incentive compatible, meaning $V_A^{a^*}(w^{FB}_{a^*})\ge V_A^a(w^{FB}_{a^*})$ for every $a\ne a^*$. By definition of a first-best solution, it also satisfies participation. Hence, it is feasible for $(\mathsf{MH}_{a^*})$. The feasible set of $(\mathsf{MH}_{a^*})$ is a subset of the first-best feasible set for the fixed action $a^*$ because it imposes the same participation constraint and additional incentive constraints. Since $w^{FB}_{a^*}$ maximizes the principal's objective over the larger first-best feasible set, no moral-hazard feasible contract can yield a strictly larger principal payoff. Therefore, $w^{FB}_{a^*}$ solves $(\mathsf{MH}_{a^*})$.
\end{proof}

\begin{proof}[Proof of Theorem \ref{thm:MH-general}]
Let $w^*$ be a smooth, regular solution of $(\mathsf{MH}_{a^*})$. Define the participation function $g_0(w):=V_A^{a^*}(w)-U_0$ and, for each deviation $a\ne a^*$, the incentive function $g_a(w):=V_A^{a^*}(w)-V_A^a(w)$. The feasible set is described locally by $g_0(w)\ge0$ and $g_a(w)\ge0$ for all deviations. Lemma \ref{lem:MH-binding} implies $g_0(w^*)=0$.

Let $B(w^*):=\{a\ne a^*:g_a(w^*)=0\}$ denote the active deviations. Smoothness means that the principal's effective model under $a^*$, the agent's effective model under $a^*$, and the agent's effective model under every $a\in B(w^*)$ are unique. By Appendix \ref{app:prelim}, the relevant functions are differentiable at $w^*$ and their gradients are $\nabla V_P^{a^*}(w^*)=-\pi_P^{a^*}(\cdot;w^*)$, $\nabla g_0(w^*)=\pi_A^{a^*}(\cdot;w^*)$, and $\nabla g_a(w^*)=\pi_A^{a^*}(\cdot;w^*)-\pi_A^a(\cdot;w^*)$ for each active deviation $a$.

Regularity is the Mangasarian-Fromovitz constraint qualification for these binding constraints. The standard KKT theorem for a finite-dimensional nonlinear program therefore applies. Thus, there are multipliers $\eta\ge0$ for participation and $\beta^a\ge0$ for the incentive constraints, with $\beta^a=0$ for every inactive constraint by complementary slackness, such that stationarity is
$-\pi_P^{a^*}(\cdot;w^*)+\eta\pi_A^{a^*}(\cdot;w^*)+\sum_{a\ne a^*}\beta^a\big[\pi_A^{a^*}(\cdot;w^*)-\pi_A^a(\cdot;w^*)\big]=0.$
We now determine the participation multiplier. Take the inner product of this vector equality with $\m{1}$. Every effective model is a probability distribution, so its coordinates sum to one. Therefore, $\m{1}\cdot\pi_P^{a^*}=1$, $\m{1}\cdot\pi_A^{a^*}=1$, and $\m{1}\cdot(\pi_A^{a^*}-\pi_A^a)=0$ for every deviation. The inner product of stationarity with $\m{1}$ therefore reduces to $-1+\eta=0$, so $\eta=1$. Substituting $\eta=1$ into stationarity and rearranging yields $\pi_P^{a^*}(\cdot;w^*)=\pi_A^{a^*}(\cdot;w^*)+\sum_{a\ne a^*}\beta^a[\pi_A^{a^*}(\cdot;w^*)-\pi_A^a(\cdot;w^*)]$. The multipliers are nonnegative, and complementary slackness already gives $\beta^a=0$ whenever the corresponding incentive constraint is inactive. 
\end{proof}

\begin{proof}[Proof of Corollary \ref{thm:MH-contract}]
Assume multiplier preferences and suppose $(\mathsf{MH}_{a^*})$ is feasible. We proceed in four steps. First, we transform the contracting problem into a finite-dimensional convex program. Second, we establish existence and uniqueness. Third, we derive the multiplier representation of the unique solution. Finally, we translate back into wages.

\emph{Step 1: transformation.} We may restrict without loss to contracts with binding participation. Indeed, for any feasible $w$, let $\delta:=V_A^{a^*}(w)-U_0\ge0$ and set $\widetilde w:=w-\delta\m{1}$. Translation invariance leaves every incentive difference unchanged, makes participation bind, and raises the principal's value by $\delta$. Hence, the original problem and its restriction to binding participation have the same supremum, and any optimizer of the latter solves the former. Define, for every action $a$, $C_a:=\exp\{-\lambda_A[U_0+c(a)]\}>0$, and make the one-to-one change of variables $t(y):=e^{-\lambda_A w(y)}$. Since $w(y)$ is finite, $t(y)>0$ for every output, and conversely every strictly positive vector $t$ corresponds to exactly one wage vector through $w(y)=-(1/\lambda_A)\log t(y)$.

Under the desired action, binding participation is $-\lambda_A^{-1}\log\E_{q^{a^*}}[t]-c(a^*)=U_0$. Exponentiating this equality gives $\E_{q^{a^*}}[t]=C_{a^*}$. For a deviation $a\ne a^*$, incentive compatibility is $-\lambda_A^{-1}\log\E_{q^{a^*}}[t]-c(a^*)\ge-\lambda_A^{-1}\log\E_{q^a}[t]-c(a)$. Multiplying by $-\lambda_A<0$ reverses the inequality, and then exponentiating gives $\E_{q^a}[t]\ge C_a$. To see the constant exactly, use the binding participation equality: the incentive inequality is equivalent to $\log\E_{q^a}[t]\ge-\lambda_A[U_0+c(a)]$, hence to $\E_{q^a}[t]\ge C_a$. Thus, after imposing binding participation, the feasible contracts correspond exactly to strictly positive vectors $t$ satisfying one affine equality, $\E_{q^{a^*}}[t]=C_{a^*}$, and the affine inequalities $\E_{q^a}[t]\ge C_a$ for all $a\ne a^*$.

Write $r:=\lambda_P/\lambda_A>0$. The principal's multiplier value is $-\lambda_P^{-1}\log\sum_yq^{a^*}(y)e^{-\lambda_Py}e^{\lambda_Pw(y)}$. Since $e^{\lambda_Pw(y)}=t(y)^{-r}$ and $-\lambda_P^{-1}\log(\cdot)$ is strictly decreasing, maximizing the principal's value is equivalent to minimizing $F(t):=\sum_yq^{a^*}(y)e^{-\lambda_Py}t(y)^{-r}$ over the transformed feasible set.

\emph{Step 2: existence and uniqueness.} Let $\overline T$ be the set of nonnegative vectors satisfying the transformed affine equality and inequalities. It is closed. It is also bounded: because $q^{a^*}$ has full support and $t(y)\ge0$, the equality $\sum_zq^{a^*}(z)t(z)=C_{a^*}$ implies $0\le t(y)\le C_{a^*}/q^{a^*}(y)$ for every $y$. Thus, $\overline T$ is compact. Feasibility of the original problem provides at least one strictly positive feasible vector, so the transformed feasible set is nonempty.

Each function $u\mapsto u^{-r}$ is strictly convex on $(0,\infty)$ and tends to $+\infty$ as $u\downarrow0$. Because every coefficient $q^{a^*}(y)e^{-\lambda_Py}$ is strictly positive, $F$ is strictly convex on the positive orthant and tends to $+\infty$ along any sequence in $\overline T$ for which at least one coordinate tends to zero. Take a minimizing sequence of strictly positive feasible vectors. Compactness of $\overline T$ gives a convergent subsequence. Its limit cannot have a zero coordinate because then $F$ along that subsequence would diverge to $+\infty$, while a strictly positive feasible vector has finite objective. Hence, the limit is strictly positive and feasible, and continuity of $F$ there shows that it attains the minimum. Strict convexity of $F$ on a convex feasible set makes this minimizer unique.
\par Since $t\mapsto w$ is one-to-one, the original moral-hazard problem has a unique optimal contract. Denote the unique minimizer by $t^*$.

\emph{Step 3: multipliers for the transformed linear constraints.} Let $B$ denote the set of deviations whose transformed constraints bind at $t^*$, so $\E_{q^a}[t^*]=C_a$ for $a\in B$. Since $t^*$ is strictly positive, the positivity constraints are locally inactive. The feasible first-order directions at $t^*$ are therefore the vectors $d$ satisfying $q^{a^*}\cdot d=0$ and $q^a\cdot d\ge0$ for every $a\in B$. Optimality implies that the directional derivative of the relevant objective is nonnegative in every such direction.

We now justify the existence of Lagrange multipliers without imposing a constraint qualification. Let $g:=\nabla F(t^*)$. If there were no scalar $\eta$ and nonnegative numbers $\gamma^a$, $a\in B$, satisfying $g+\eta q^{a^*}-\sum_{a\in B}\gamma^aq^a=0$, the finite-dimensional Farkas alternative \citep[Section 5.8.3, p. 263]{boyd04} applied to the cone generated by the active inequality normals and the equality normal would yield a direction $d$ with $q^{a^*}\cdot d=0$, $q^a\cdot d\ge0$ for all $a\in B$, and $g\cdot d<0$. For sufficiently small $\varepsilon>0$, $t^*+\varepsilon d$ would remain strictly positive; all inactive inequalities would remain satisfied because they have positive slack; the active inequalities would be weakly relaxed; and the equality would continue to hold. The differentiability of $F$ and $g\cdot d<0$ would then imply $F(t^*+\varepsilon d)<F(t^*)$ for all sufficiently small $\varepsilon$, contradicting optimality. Hence, such multipliers exist. Set $\gamma^a=0$ for inactive deviations. 

Differentiating $F$ gives $\partial F(t^*)/\partial t(y)=-r q^{a^*}(y)e^{-\lambda_Py}t^*(y)^{-r-1}$. The multiplier identity therefore reads $-r q^{a^*}(y)e^{-\lambda_Py}t^*(y)^{-r-1}+\eta q^{a^*}(y)-\sum_{a\ne a^*}\gamma^aq^a(y)=0$ for every $y$. Multiply this equality by $t^*(y)$ and sum over outputs. The first term sums to $-rF(t^*)$; the equality constraint makes the second term $\eta C_{a^*}$; and complementary slackness in the construction of the normal cone gives $\gamma^a>0$ only for active constraints, for which $\E_{q^a}[t^*]=C_a$. Hence, $\eta C_{a^*}=rF(t^*)+\sum_{a\ne a^*}\gamma^aC_a$. Define $\beta^a:=\gamma^aC_a/[rF(t^*)]\ge0$. Since $F(t^*)>0$, this is well defined, and $\beta^a=0$ for every inactive deviation. Also, $C_{a^*}/C_a=e^{\lambda_A[c(a)-c(a^*)]}$. Substituting the expression for $\eta$ into the coordinatewise multiplier identity and dividing by the positive number $r q^{a^*}(y)$ yields $e^{-\lambda_Py}t^*(y)^{-r-1}=[F(t^*)/C_{a^*}]R(y)$, where $R(y):=1+\sum_{a\ne a^*}\beta^a[1-e^{\lambda_A[c(a)-c(a^*)]}q^a(y)/q^{a^*}(y)]$. The left-hand side is strictly positive, as is $F(t^*)/C_{a^*}$, so $R(y)>0$ for every output.

\emph{Step 4: return to wages.} Taking logarithms of $e^{-\lambda_Py}t^*(y)^{-r-1}=[F(t^*)/C_{a^*}]R(y)$ is legitimate because every factor is strictly positive. Since $\log t^*(y)=-\lambda_Aw^*(y)$ and $(r+1)\lambda_A=\lambda_A+\lambda_P$, the output-dependent terms rearrange to $(\lambda_A+\lambda_P)w^*(y)=\lambda_Py+\log R(y)+K$ for a constant $K$ independent of $y$. Dividing by $\lambda_A+\lambda_P$ and writing $\alpha:=K/(\lambda_A+\lambda_P)$ gives exactly the formula in the corollary. By Lemma \ref{lem:MH-binding}, participation binds. For any fixed output-dependent part of the wage formula, adding a constant $\alpha$ changes the agent's value one-for-one, so there is only one intercept satisfying participation. This proves existence of the stated representation, positivity of the log term, the inactivity property of multipliers, and uniqueness of the optimal contract.
\end{proof}

\begin{proof}[Proof of Corollary \ref{prop:MH-slope-floor}]
Use the multiplier case in Corollary \ref{thm:MH-contract} and let $R(y):=1+\sum_{a\ne a^*}\beta^a[1-e^{\lambda_A[c(a)-c(a^*)]}q^a(y)/q^{a^*}(y)]$. Corollary \ref{thm:MH-contract} proves that $R(y)>0$ for every output.

Fix a deviation $a$. By hypothesis, $y\mapsto q^a(y)/q^{a^*}(y)$ is weakly decreasing. Multiplication by the positive constant $e^{\lambda_A[c(a)-c(a^*)]}$ preserves this ordering, and subtracting the result from one reverses it, so $y\mapsto1-e^{\lambda_A[c(a)-c(a^*)]}q^a(y)/q^{a^*}(y)$ is weakly increasing. Multiplication by the nonnegative multiplier $\beta^a$ preserves weak increase. Summing across deviations and adding one therefore shows that $R$ is weakly increasing. Since $R$ is strictly positive and the logarithm is strictly increasing, $\log R$ is weakly increasing as well.

Take any two outputs $y'>y$. Subtracting the wage formula in Corollary \ref{thm:MH-contract} at these outputs gives
$w^*(y')-w^*(y)=\frac{\lambda_P}{\lambda_A+\lambda_P}(y'-y)+\frac{\log R(y')-\log R(y)}{\lambda_A+\lambda_P}.$
The second term is nonnegative, proving the claimed lower bound. The first term is strictly positive because $y'-y>0$, so the optimal contract is strictly increasing. Finally, suppose there is a deviation $a$ with $\beta^a>0$ and $q^a(y')/q^{a^*}(y')<q^a(y)/q^{a^*}(y)$. The corresponding summand in $R$ increases strictly between $y$ and $y'$, while every other summand weakly increases. Hence, $R(y')>R(y)$ and $\log R(y')>\log R(y)$, so the lower bound is strict.
\end{proof}

\begin{proof}[Proof of Corollary \ref{cor:two-output-bonus}]
Let $w(0)=s$ and $w(1)=s+b$. If an action has success probability $p$, multiplier translation invariance gives $M(q,w;\lambda_A)=s+\Psi(p,b)$, where $\Psi(p,b):=-\lambda_A^{-1}\log(1-p+pe^{-\lambda_Ab})$. Therefore, the fixed payment $s$ cancels from the incentive constraint. Implementing $H$ requires $D(b):=\Psi(p_H,b)-\Psi(p_L,b)\ge\Delta c_{HL}$.

We first characterize this constraint. Differentiating $\Psi$ with respect to $b$ gives $\partial_b\Psi(p,b)=pe^{-\lambda_Ab}/(1-p+pe^{-\lambda_Ab})$. Hence,
$D'(b)=\frac{e^{-\lambda_Ab}(p_H-p_L)}{(1-p_H+p_He^{-\lambda_Ab})(1-p_L+p_Le^{-\lambda_Ab})}>0$
for every finite $b$. Thus, $D$ is continuous and strictly increasing. Moreover, $D(0)=0$, while $\lim_{b\to\infty}D(b)=\lambda_A^{-1}\log[(1-p_L)/(1-p_H)]=\overline C_{HL}$. Since $\Delta c_{HL}>0$, any implementing bonus must satisfy $b>0$. Because $D$ is strictly increasing and approaches $\overline C_{HL}$ from below as $b\to\infty$, a finite implementing bonus exists if and only if $\Delta c_{HL}<\overline C_{HL}$. When this inequality holds, there is a unique finite $b^{IC}>0$ satisfying $D(b^{IC})=\Delta c_{HL}$, and incentive compatibility is equivalent to $b\ge b^{IC}$.

We solve explicitly for $b^{IC}$. Let $r:=e^{-\lambda_A\Delta c_{HL}}\in(0,1)$ and $t:=e^{-\lambda_Ab^{IC}}\in(0,1)$. The equality $D(b^{IC})=\Delta c_{HL}$ is equivalent to $(1-p_H+p_Ht)/(1-p_L+p_Lt)=r$. Solving this linear equation in $t$ gives $t=[r(1-p_L)-(1-p_H)]/[p_H-rp_L]$. The numerator is strictly positive exactly because $\Delta c_{HL}<\overline C_{HL}$, and the denominator is strictly positive because $r<1$ and $p_H>p_L$. Moreover, the numerator minus the denominator equals $r-1<0$, so $t<1$. Thus, $t\in(0,1)$ and taking $-(1/\lambda_A)\log t$ gives the expression for $b^{IC}$.

We next optimize the principal's payoff over the feasible bonuses. By Lemma \ref{lem:MH-binding}, participation binds. Hence, $s+\Psi(p_H,b)-c(H)=U_0$, so $s=U_0+c(H)-\Psi(p_H,b)$. Substituting this fixed payment into the principal's payoff leaves, up to the constant $-U_0-c(H)$, the one-dimensional objective $G(b):=\Psi(p_H,b)+M(q^H,(1-b)y;\lambda_P)$.

Differentiation gives $G'(b)=p_He^{-\lambda_Ab}/(1-p_H+p_He^{-\lambda_Ab})-p_He^{-\lambda_P(1-b)}/(1-p_H+p_He^{-\lambda_P(1-b)})$. For fixed $p_H\in(0,1)$, the map $x\mapsto p_Hx/(1-p_H+p_Hx)$ is strictly increasing on $(0,\infty)$. Therefore, $G'(b)=0$ if and only if $e^{-\lambda_Ab}=e^{-\lambda_P(1-b)}$, which is equivalent to $b=\lambda_P/(\lambda_A+\lambda_P)$. The second derivative of $\Psi(p_H,b)$ is strictly negative. The second derivative of $b\mapsto M(q^H,(1-b)y;\lambda_P)$ is also strictly negative because its first derivative is minus the principal's effective success probability, and that probability is strictly increasing in $b$. Hence, $G$ is strictly concave. The displayed stationary point is therefore its unique unconstrained maximizer.

The feasible set for the bonus is the interval $[b^{IC},\infty)$ and $G$ is strictly concave with unique unconstrained maximizer $\lambda_P/(\lambda_A+\lambda_P)$. Therefore, its unique constrained maximizer is the larger of the lower bound and the unconstrained maximizer: $b^*=\max\{b^{IC},\lambda_P/(\lambda_A+\lambda_P)\}$. Once $b^*$ is fixed, binding participation determines a unique fixed payment, namely $s^*=U_0+c(H)-\Psi(p_H,b^*)$. This proves the result.
\end{proof}

\subsection{Proofs for Section \ref{sec:continuum}}

For the multiplier specialization used after Proposition \ref{prop:continuum-general}, fix a finite contract $w(0)=s$ and $w(1)=s+b$, and define $d:=1-e^{-\lambda_Ab}$. Since $b$ is finite, $d<1$. If action $a$ is identified with the success probability, multiplier utility is $M(q^a,w;\lambda_A)=s-\lambda_A^{-1}\log(1-ad)$. The quantity inside the logarithm is strictly positive because $1-ad=1-a+ae^{-\lambda_Ab}$ and $a\in(0,1)$. Differentiating with respect to $a$ gives $\partial_aM(q^a,w;\lambda_A)=d/[\lambda_A(1-ad)]$ and $\partial_{aa}M(q^a,w;\lambda_A)=d^2/[\lambda_A(1-ad)^2]=\lambda_A[\partial_aM(q^a,w;\lambda_A)]^2$. 

\begin{proof}[Proof of Proposition \ref{prop:continuum-general}]
Fix a finite contract $w$ and suppose that an interior action $a$ is implemented. Define the agent's objective as a function of a possible action $\tilde a$ by $J(\tilde a):=\mathcal V_{\lambda_A}^{\tilde a}(w)-c(\tilde a)$. By assumption, the first term is twice differentiable in a neighborhood of $a$, and $c$ is twice differentiable, so $J$ is twice differentiable there. Because $a$ is an interior global maximizer, it is in particular an interior local maximizer. The elementary first- and second-order necessary conditions therefore give $J'(a)=0$ and $J''(a)\le0$.

Expanding these derivatives yields $\partial_a\mathcal V_{\lambda_A}^a(w)-c'(a)=0$ and $\partial_{aa}\mathcal V_{\lambda_A}^a(w)-c''(a)\le0$. Rearranging gives the first two conditions  in the statement. If $c'(a)\ne0$, the first-order condition implies $\partial_a\mathcal V_{\lambda_A}^a(w)=c'(a)\ne0$, so the local ambiguity-curvature index $\kappa_A(a;w)$ is well defined. By its definition, $\partial_{aa}\mathcal V_{\lambda_A}^a(w)=\kappa_A(a;w)[\partial_a\mathcal V_{\lambda_A}^a(w)]^2$. Substituting the first-order condition into the second-order inequality therefore gives $c''(a)\ge\kappa_A(a;w)[c'(a)]^2$. 
\end{proof}

\begin{proof}[Proof of Observation \ref{prop:continuum-interior}]
Let $a\in(\underline a,\bar a)$ be implemented by a finite contract $w(0)=s$, $w(1)=s+b$. Define $J(\tilde a):=M(q^{\tilde a},w;\lambda_A)-c(\tilde a)$. Since $a$ is an interior maximizer, $J'(a)=0$ and $J''(a)\le0$. We first show that $b>0$. Recall $d:=1-e^{-\lambda_Ab}$. If $b=0$, then $d=0$, so the derivative of gross compensation with respect to action is zero. Hence, $J'(a)=-c'(a)<0$, contradicting $J'(a)=0$. If $b<0$, then $e^{-\lambda_Ab}>1$ and therefore $d<0$. The denominator $1-ad$ is strictly positive, so $\partial_aM=d/[\lambda_A(1-ad)]<0$. Since $c'(a)>0$, this again gives $J'(a)<0$. Therefore, $b>0$. In particular, $d\in(0,1)$.

The first-order condition now reads $c'(a)=d/[\lambda_A(1-ad)]$. Since $0<d<1$ and $a<1$, we have $d/(1-ad)<1/(1-a)$. To verify this inequality without cross-multiplication ambiguity, note that both denominators are strictly positive and $d(1-a)<1-ad$ is equivalent simply to $d<1$, which holds. Multiplying by $1/\lambda_A$ therefore gives $c'(a)<1/[\lambda_A(1-a)]$, or equivalently $\lambda_A(1-a)c'(a)<1$. This is the first inequality in \eqref{eq:continuum-conditions}.

The same first-order condition can be solved for the bonus. Multiplying by $\lambda_A(1-ad)$ gives $\lambda_Ac'(a)-\lambda_Aa c'(a)d=d$, so $d=\lambda_Ac'(a)/[1+\lambda_Aa c'(a)]$. Hence, $e^{-\lambda_Ab}=1-d=[1-\lambda_A(1-a)c'(a)]/[1+\lambda_Aa c'(a)]$. The first inequality just proved makes the numerator strictly positive. Taking logarithms therefore gives the unique stationary bonus.

Finally, the second-order necessary condition and the multiplier derivative identity recorded before the proof give $0\ge J''(a)=\lambda_A[\partial_aM(q^a,w;\lambda_A)]^2-c''(a)$. The first-order condition identifies $\partial_aM(q^a,w;\lambda_A)$ with $c'(a)$. Substitution yields $0\ge\lambda_A[c'(a)]^2-c''(a)$, or $c''(a)\ge\lambda_A[c'(a)]^2$. This is the second inequality in \eqref{eq:continuum-conditions}.
\end{proof}

\begin{proof}[Proof of Corollary \ref{cor:continuum-foa}]
Let $a^*\in(\underline a,\bar a)$ satisfy the agent's first-order condition under a finite contract. Then, $\partial_aM(q^{a^*},w;\lambda_A)=c'(a^*)$. The second derivative identity recorded above gives $\partial_{aa}M(q^{a^*},w;\lambda_A)=\lambda_A[c'(a^*)]^2$. Therefore, the second derivative of the agent's net objective at $a^*$ is $\lambda_A[c'(a^*)]^2-c''(a^*)$, which is strictly positive by hypothesis. The second-derivative test then implies that $a^*$ is a strict local minimum of the agent's objective. In particular, every sufficiently small neighborhood of $a^*$ contains feasible actions that give the agent strictly greater utility. Hence, $a^*$ cannot be a global maximizer and is not implemented by the contract.
\end{proof}

\begin{proof}[Proof of Theorem \ref{thm:continuum-polarization}]
Fix a finite contract. The agent's objective is continuous in $a$ on the compact interval $[\underline a,\bar a]$, so the maximizer set is nonempty. Suppose an interior action were a maximizer. Observation \ref{prop:continuum-interior} would then require $c''(a)\ge\lambda_A[c'(a)]^2$, contradicting the strict inequality assumed at every interior action. Thus, no interior action can maximize the agent's objective. Every maximizer must therefore be one of the two endpoints.

It remains to compare the endpoint utilities as a function of the bonus. Define $D(b)$ as the agent's payoff under $\bar a$ minus his payoff under $\underline a$. The fixed payment cancels, so direct substitution into the binary multiplier formula gives $D(b)=-\lambda_A^{-1}\log\{[1-\bar a+\bar ae^{-\lambda_Ab}]/[1-\underline a+\underline ae^{-\lambda_Ab}]\}-\Delta c$. Differentiating gives
$D'(b)=\frac{e^{-\lambda_Ab}(\bar a-\underline a)}{(1-\bar a+\bar ae^{-\lambda_Ab})(1-\underline a+\underline ae^{-\lambda_Ab})}>0.$
Thus, $D$ is continuous and strictly increasing on $\R$. At $b=0$, the two actions generate the same gross compensation and the high endpoint costs $\Delta c>0$ more, so $D(0)=-\Delta c<0$. As $b\to\infty$, $e^{-\lambda_Ab}\to0$ and therefore $D(b)\to\lambda_A^{-1}\log[(1-\underline a)/(1-\bar a)]-\Delta c=\overline C-\Delta c$.

First suppose $\Delta c\ge\overline C$. If $\Delta c>\overline C$, the limiting value is strictly negative. Because $D$ is increasing, $D(b)$ is then strictly negative for every finite $b$. If $\Delta c=\overline C$, the limit is zero. Strict increase implies that every finite value of $D$ lies strictly below its limit, so again $D(b)<0$ for every finite $b$. In either case, the low endpoint gives strictly higher utility than the high endpoint under every finite contract. Since no interior action can maximize, $\underline a$ is the unique maximizer.

Now suppose $\Delta c<\overline C$. Then, $D(0)<0$ and $\lim_{b\to\infty}D(b)>0$. Continuity gives at least one root, and strict monotonicity gives exactly one. Denote it by $\widetilde b$. Because $D(0)<0$ and the root is unique, $\widetilde b>0$. To obtain its formula, let $t:=e^{-\lambda_A\widetilde b}$. The equality $D(\widetilde b)=0$ is equivalent to $(1-\bar a+\bar at)/(1-\underline a+\underline at)=e^{-\lambda_A\Delta c}$. Solving for $t$ gives the ratio displayed in the definition of $\widetilde b$. The numerator is strictly positive because $\Delta c<\overline C$ is equivalent to $e^{-\lambda_A\Delta c}(1-\underline a)>1-\bar a$. The denominator is strictly positive because $e^{-\lambda_A\Delta c}<1$ and $\bar a>\underline a$. The numerator is smaller than the denominator because their difference is $e^{-\lambda_A\Delta c}-1<0$. Hence, $t\in(0,1)$ and the logarithm defining $\widetilde b$ is finite and positive. Strict monotonicity now implies $D(b)<0$ for $b<\widetilde b$, $D(b)=0$ for $b=\widetilde b$, and $D(b)>0$ for $b>\widetilde b$. Since only endpoints can maximize, these three cases imply, respectively, that the unique maximizer is $\underline a$, both endpoints are maximizers, and the unique maximizer is $\bar a$. This is exactly \eqref{eq:continuum-jump}.
\end{proof}

\begin{proof}[Proof of Observation \ref{prop:continuum-eu}]
Under EU, the expected wage from $w(0)=s$ and $w(1)=s+b$ at action $a$ is $(1-a)s+a(s+b)=s+ab$. Hence, the agent maximizes $s+ab-c(a)$ over $a\in[\underline a,\bar a]$. Fix any target $a^*\in(\underline a,\bar a)$ and set $b=c'(a^*)$.

Strict convexity and differentiability of $c$ imply the strict supporting-line inequality $c(a)>c(a^*)+c'(a^*)(a-a^*)$  $\forall a\ne a^*$. Rearranging gives $s+a c'(a^*)-c(a)<s+a^*c'(a^*)-c(a^*)$  $\forall a\ne a^*$. Thus, under  $b=c'(a^*)$, the target $a^*$ gives strictly greater utility than every other feasible action and is uniquely implemented. Because $c$ is continuously differentiable and strictly convex, $c'$ is continuous and strictly increasing on the interior of the action interval. Its image is therefore an interval, and it has a continuous inverse on that image. For any bonus $b$ in the interior image of $c'$, the unique implemented action is the unique $a$ satisfying $c'(a)=b$, namely $(c')^{-1}(b)$. Continuity of the inverse proves that the induced action varies continuously with the bonus on this range.
\end{proof}

\subsection{Proofs for Section \ref{subsec:model-libraries}}

\begin{proof}[Proof of Theorem \ref{thm:library}]
We begin with a property that will be used in both parts of the theorem. Fix a payoff vector $x$ and $\lambda>0$. Take any two reference models $q_0,q_1$ and any $t\in[0,1]$. Let $p_0$ and $p_1$ be minimizers in the definitions of $\mathcal W_\lambda(q_0,x)$ and $\mathcal W_\lambda(q_1,x)$, respectively; existence follows from compactness and lower semicontinuity. Define $q_t:=tq_0+(1-t)q_1$ and $p_t:=tp_0+(1-t)p_1$. Joint convexity of $\Gamma$ gives $\Gamma(p_t,q_t)\le t\Gamma(p_0,q_0)+(1-t)\Gamma(p_1,q_1)$. Evaluating the minimization that defines $\mathcal W_\lambda(q_t,x)$ at the feasible model $p_t$ therefore gives $\mathcal W_\lambda(q_t,x)\le t\mathcal W_\lambda(q_0,x)+(1-t)\mathcal W_\lambda(q_1,x)$. Thus, $q\mapsto\mathcal W_\lambda(q,x)$ is convex. Moreover, $\mathcal W_\lambda(\cdot,x)$ is finite because $p=q$ is feasible and $\Gamma(q,q)=0$; hence, as a finite convex function, it is continuous on $\operatorname{ri}\DeltaY$, and therefore on the convex hull of the full-support reference models. 

Now let $\tilde\mu^a\succeq\mu^a$. By definition, the two distributions over models have the same mean and every continuous concave function has a weakly smaller expectation under $\tilde\mu^a$. Applying this definition to the concave function $-\mathcal W_\lambda(\cdot,x)$ gives $\E_{\tilde\mu^a}[\mathcal W_\lambda(q,x)]\ge\E_{\mu^a}[\mathcal W_\lambda(q,x)]$.

Suppose first that only the desired-action library becomes more dispersed: $\tilde\mu^{a^*}\succeq\mu^{a^*}$ and $\tilde\mu^a=\mu^a$ for every deviation. Apply the preceding inequality with $x=w$ and $\lambda=\lambda_A$. For every contract, the agent's value of the desired action weakly rises, while the value of every deviation is unchanged. Hence, every contract satisfying participation under $\mu$ continues to satisfy participation under $\tilde\mu$, and every incentive inequality $\widehat V_A^{a^*}(w;\mu^{a^*})\ge\widehat V_A^a(w;\mu^a)$ remains valid because only its left-hand side weakly increases. Therefore, $\widehat X^{a^*}(\mu)\subseteq\widehat X^{a^*}(\tilde\mu)$. For the principal, apply the same convexity result to $x=y-w$. Then, $\widehat V_P^{a^*}(w;\tilde\mu^{a^*})\ge\widehat V_P^{a^*}(w;\mu^{a^*})$ for every fixed contract. 
Combining this fixed-contract comparison with the implementation-set inclusion gives $\widehat\Lambda_P^{a^*}(\tilde\mu)\ge\widehat\Lambda_P^{a^*}(\mu)$. Explicitly, the supremum under $\tilde\mu$ is taken over a weakly larger set, and on every contract in the old set the new objective is weakly larger. The conclusion also holds under the convention that an empty implementation set has value $-\infty$.

Suppose instead that only a deviation library $\hat a\ne a^*$ becomes more dispersed. Applying the convexity result to the agent's wage payoff shows $\widehat V_A^{\hat a}(w;\tilde\mu^{\hat a})\ge\widehat V_A^{\hat a}(w;\mu^{\hat a})$ for every contract. The desired-action value and all other deviation values are unchanged. Therefore, any contract that satisfies the incentive constraint against $\hat a$ under the more dispersed library also satisfies it under the original library, whereas the converse need not hold. All other feasibility constraints are identical. Hence, $\widehat X^{a^*}(\tilde\mu)\subseteq\widehat X^{a^*}(\mu)$. Since the desired-action library is unchanged, the principal's objective is identical under the two profiles for every fixed contract. Maximizing the same objective over the smaller implementation set gives $\widehat\Lambda_P^{a^*}(\tilde\mu)\le\widehat\Lambda_P^{a^*}(\mu)$. This proves the theorem.
\end{proof}

\begin{proof}[Proof of Observation \ref{prop:variance}]
Fix a model $q$ and a finite payoff vector $x$. Because $\Y$ is finite, $x$ is bounded. Let $m_1:=\E_q[x]$ and $m_2:=\E_q[x^2]$. Taylor's theorem applied coordinatewise to the exponential, with a remainder uniform over the finite support, gives $\E_q[e^{-\lambda x}]=1-\lambda m_1+(\lambda^2/2)m_2+O(\lambda^3)$ as $\lambda\downarrow0$. Let $u_\lambda:=-\lambda m_1+(\lambda^2/2)m_2+O(\lambda^3)$, so the preceding expression is $1+u_\lambda$ and $u_\lambda=O(\lambda)$. Taylor's theorem for the logarithm gives $\log(1+u_\lambda)=u_\lambda-u_\lambda^2/2+O(u_\lambda^3)$. Substituting $u_\lambda$ and keeping terms through order $\lambda^2$ yields $\log\E_q[e^{-\lambda x}]=-\lambda m_1+(\lambda^2/2)(m_2-m_1^2)+O(\lambda^3)$. Since $m_2-m_1^2=\Var_q(x)$, multiplying by $-1/\lambda$ gives $M(q,x;\lambda)=\E_q[x]-(\lambda/2)\Var_q(x)+O(\lambda^2)$.

Each model library in the statement of the result is finite. Therefore, for a fixed contract, the $O(\lambda^2)$ remainder can be bounded uniformly over the finitely many $q$ in the library. Averaging the expansion over $\mu^a$ is consequently valid term by term. The average of the mean term is $\E_{\mu^a}[\E_q[x]]=\E_{\bar q^a}[x]$ by the definition of the mean model. Taking $x=y-w$ and $\lambda=\lambda_P$ gives the principal's expansion, and taking $x=w$ and $\lambda=\lambda_A$ and then subtracting $c(a)$ gives the agent's expansion.

For the final claim, fix a payoff vector $x$. As a function of $q$, $\Var_q(x)=\sum_yq(y)x(y)^2-[\sum_yq(y)x(y)]^2$. The first term is affine in $q$. The second term is the square of an affine function, hence convex. Its negative is concave. Therefore, $q\mapsto\Var_q(x)$ is concave. If $\tilde\mu^a\succeq\mu^a$, the definition of the dispersion order applied to this continuous concave function gives $\E_{\tilde\mu^a}[\Var_q(x)]\le\E_{\mu^a}[\Var_q(x)]$, as required.
\end{proof}

\begin{proof}[Proof of Observation \ref{prop:smooth-library-dispersion}]
Fix a contract $w$ and an action $a$. The map $q\mapsto\E_q[w]-c(a)$ is affine in the probability vector $q$. Since $\phi_A$ is concave and increasing, the composition $q\mapsto\phi_A(\E_q[w]-c(a))$ is concave. Therefore, if $\tilde\mu^a\succeq\mu^a$, the definition of the dispersion order gives $U_{A,S}^a(w;\tilde\mu^a)\le U_{A,S}^a(w;\mu^a)$.

Suppose first that only the desired-action library becomes more dispersed. Let $w\in X_S^{a^*}(\tilde\mu)$. Under $\tilde\mu$, the contract satisfies participation and every incentive constraint. Replacing $\tilde\mu^{a^*}$ by the less dispersed $\mu^{a^*}$ weakly raises the agent's value of the desired action, while leaving every deviation value unchanged. Hence, participation remains satisfied and every incentive inequality remains satisfied. Thus, $w\in X_S^{a^*}(\mu)$, proving $X_S^{a^*}(\tilde\mu)\subseteq X_S^{a^*}(\mu)$. Now, suppose instead that only a deviation library $\hat a$ becomes more dispersed. Let $w\in X_S^{a^*}(\mu)$. The desired-action value is unchanged, while the value of deviation $\hat a$ weakly falls under the more dispersed library. The incentive constraint against $\hat a$ therefore remains satisfied. All other incentive constraints and participation are unchanged. Hence, $w\in X_S^{a^*}(\tilde\mu)$, proving $X_S^{a^*}(\mu)\subseteq X_S^{a^*}(\tilde\mu)$.

For the principal, the map $q\mapsto\E_q[y-w]$ is affine, so concavity of $\phi_P$ makes $q\mapsto\phi_P(\E_q[y-w])$ concave. The dispersion order therefore gives $\E_{\tilde\mu^{a^*}}[\phi_P(\E_q[y-w])]\le\E_{\mu^{a^*}}[\phi_P(\E_q[y-w])]$ for every fixed contract. This is the final claim.
\end{proof}

\begin{proof}[Proof of Proposition \ref{prop:endogenous-dispersion}]
Fix $0\le\delta<\delta'\le\bar\delta$. We first prove that each library becomes more dispersed in the convex-order sense while preserving its mean. For action $j\in\{H,L\}$, the mean success probability under $\mu_\delta^j$ is $\sum_k\nu_ke_j(\bar\theta+\delta z_k)=e_j\bar\theta$ because $\sum_k\nu_kz_k=0$. Thus, the mean model is independent of $\delta$.

Let $r:=\delta/\delta'\in[0,1)$ and denote the common mean model for action $j$ by $q^j$. For each scenario $k$, linearity of success probability in the parameter implies $q_\delta^{j,k}=r q_{\delta'}^{j,k}+(1-r)q^j$. Let $\phi$ be any continuous concave function of a model. Concavity gives $\phi(q_\delta^{j,k})\ge r\phi(q_{\delta'}^{j,k})+(1-r)\phi(q^j)$. Averaging with weights $\nu_k$ yields $\E_{\mu_\delta^j}[\phi(q)]\ge r\E_{\mu_{\delta'}^j}[\phi(q)]+(1-r)\phi(q^j)$. Jensen's inequality applied to the more dispersed library gives $\phi(q^j)=\phi(\E_{\mu_{\delta'}^j}[q])\ge\E_{\mu_{\delta'}^j}[\phi(q)]$. Substituting this lower bound into the preceding inequality gives $\E_{\mu_\delta^j}[\phi(q)]\ge\E_{\mu_{\delta'}^j}[\phi(q)]$. Together with equality of the means, this proves $\mu_{\delta'}^j\succeq\mu_\delta^j$ for both actions.

We next prove the implementation-set inclusion. Let $w$ implement $H$ under $\mu(\delta)$, and write $w(0)=s$, $w(1)=s+b$. Necessarily $b>0$. If $b=0$, expected and robust gross wage compensation is independent of action, while $c(H)>c(L)$, so $H$ cannot be optimal. If $b<0$, the function $p\mapsto\Psi(p,b)$ is strictly decreasing because its derivative with respect to $p$ is $(1-e^{-\lambda_Ab})/[\lambda_A(1-p+pe^{-\lambda_Ab})]<0$. Since every common parameter realization generates a larger success probability under $H$ than under $L$, gross wage value is then strictly smaller under $H$ in every scenario while the cost is higher, again making implementation impossible. Therefore, every contract in $\widehat X^H(\mu(\delta))$ has $b>0$.

For such a bonus, define $G_j(\delta,b):=\sum_k\nu_k\Psi(e_j(\bar\theta+\delta z_k),b)$. The incentive constraint is $G_H(\delta,b)-G_L(\delta,b)\ge\Delta c_{HL}$. Let $d_b:=1-e^{-\lambda_Ab}\in(0,1)$. Direct differentiation of $\Psi$ with respect to the success probability gives $\Psi_p(p,b)=d_b/[\lambda_A(1-d_bp)]$ and $\Psi_{pp}(p,b)=d_b^2/[\lambda_A(1-d_bp)^2]>0$. Differentiating the gross incentive advantage with respect to $\delta$ gives $\sum_k\nu_kz_kh_b(\bar\theta+\delta z_k)$, where $h_b(\theta):=e_H\Psi_p(e_H\theta,b)-e_L\Psi_p(e_L\theta,b)$.

We claim that $h_b$ is strictly increasing. Its derivative is $e_H^2\Psi_{pp}(e_H\theta,b)-e_L^2\Psi_{pp}(e_L\theta,b)$. Substituting the formula for $\Psi_{pp}$ shows that the sign is the sign of $[e_H/(1-d_be_H\theta)]^2-[e_L/(1-d_be_L\theta)]^2$. On the admissible parameter range, all denominators are positive because all success probabilities lie in $(0,1)$. For fixed $d_b\theta\ge0$, the function $x\mapsto x/(1-d_b\theta x)$ is strictly increasing wherever its denominator is positive. Since $e_H>e_L$, the derivative of $h_b$ is strictly positive.

Using $\sum_k\nu_kz_k=0$, the derivative of the incentive advantage can be symmetrized as one half of the double sum $\sum_{k,\ell}\nu_k\nu_\ell(z_k-z_\ell)[h_b(\bar\theta+\delta z_k)-h_b(\bar\theta+\delta z_\ell)]$. Each term in this double sum is nonnegative because $h_b$ is increasing: the two bracketed differences always have the same sign. Hence, $G_H(\delta,b)-G_L(\delta,b)$ is weakly increasing in $\delta$. Every incentive-compatible contract at $\delta$ therefore remains incentive compatible at $\delta'$.

Participation also remains satisfied. We have already established $\mu_{\delta'}^H\succeq\mu_\delta^H$. In the ARC specialization used here, Theorem \ref{thm:library} implies that the agent's value of a fixed wage under $H$ weakly rises with this dispersion. Hence, any contract satisfying participation at $\delta$ also satisfies it at $\delta'$. Combining participation and incentive compatibility gives $\widehat X^H(\mu(\delta))\subseteq\widehat X^H(\mu(\delta'))$. Finally, for every fixed contract, Theorem \ref{thm:library} and $\mu_{\delta'}^H\succeq\mu_\delta^H$ imply $\widehat V_P^H(w;\mu_{\delta'}^H)\ge\widehat V_P^H(w;\mu_\delta^H)$. The objective is thus weakly higher at every old feasible contract, and the feasible set weakly expands. Taking suprema proves $\widehat\Lambda_P^H(\mu(\delta'))\ge\widehat\Lambda_P^H(\mu(\delta))$. The final optimization statement now follows immediately. The function $\delta\mapsto\widehat\Lambda_P^H(\mu(\delta))$ is weakly increasing by what we just proved, and $\delta\mapsto K\chi(\delta)$ is weakly increasing because $K\ge0$ and $\chi$ is weakly increasing. Their sum is therefore weakly increasing on $[0,\bar\delta]$, so $\bar\delta$ is a maximizer. If $K>0$ and $\chi$ is strictly increasing, then $K\chi(\delta)$ is strictly increasing, while the first term is weakly increasing. Their sum is strictly increasing, making $\bar\delta$ the unique maximizer.
\end{proof}

\bibliography{ref}
\bibliographystyle{apalike}

\clearpage

\renewcommand{\thepage}{OA-\arabic{page}} \setcounter{page}{1}
{\noindent\LARGE\bf For Online Publication}

\section{Appendix: Generality of First Best}\label{app:general-criteria}

To formalize the scope of Theorem \ref{thm:firstbest-linear}, we consider the Monotonic, Bernoullian, and Archimedean (MBA) preferences of \citet{mba11}. This class imposes general ``rationality'' axioms: weak order, monotonicity, risk independence, and Archimedean continuity, so most ambiguity preferences are special cases of MBA. Within this class, imposing \citeauthor{schmeidler89}'s (\citeyear{schmeidler89}) Uncertainty Aversion axiom and invariance to the location of utility yields the variational subclass axiomatized by \citet{mmr06}; under risk neutrality, the latter property is translation invariance. Multiplier preferences were axiomatized by \citet{strz11}. 
Appendix \ref{app:smooth-ambiguity} uses smooth ambiguity to illustrate what happens when translation invariance is relaxed, while Appendix \ref{app:veu} uses vector expected utility to illustrate what happens when Uncertainty Aversion is relaxed. These comparisons identify the preference restrictions that make linear contracts first-best optimal.

\subsection{Smooth ambiguity}\label{app:smooth-ambiguity}

We first relax translation invariance while retaining Uncertainty Aversion. The smooth-ambiguity preference of \citet{smooth05} is a popular example because its ambiguity attitude can depend on the level of model-specific expected payoffs.\footnote{These preferences are neutral to model misspecification \citep[][Section 6.1]{hansenmiss25}.} Observation \ref{prop:smooth-nonlinear} shows that this payoff-level dependence can produce a unique nonlinear first-best contract.
\par To keep the comparison parallel to the first-best problem with variational preferences, let both parties be risk neutral  and share a finite library $Q^a\subset\DeltaY$ of full-support models and a second-order probability $\mu^a\in\Delta(Q^a)$ with full support. For party $i\in\{A,P\}$, let $\phi_i:\R\to\R$ be continuous and strictly increasing and write its smooth-ambiguity certainty equivalent as $\mathcal S_i^a(x):=\phi_i^{-1}\big(\sum_{q\in Q^a}\mu^a(q)\phi_i(\E_q[x])\big)$. The corresponding first-best problem maximizes $\mathcal S_P^a(y-w)$ subject to $\mathcal S_A^a(w)-c(a)\ge U_0$.

\begin{obs}\label{prop:smooth-nonlinear}
There is a three-output first-best problem in which the agent and principal share the same full-support model library and second-order belief, both are risk neutral in money and strictly ambiguity averse, and the unique optimal contract is nonlinear in output. In particular, take $\Y=\{0,1,2\}$, $c(a)=0$, equal second-order probability on $q^1=(0.8,0.1,0.1)$, $q^2=(0.1,0.8,0.1)$, and $q^3=(0.1,0.1,0.8)$, together with $\phi_P(t)=-e^{-t}$ and $\phi_A(t)=-e^{-t}-(1/2)e^{-2t}$. There exists an outside option $U_0$ for which the first-best solution is unique and nonlinear in output.
\end{obs}

The key mechanism is visible in the marginal weights. For party $i\in\{A,P\}$, the marginal weight on model $q$ is proportional to $\mu^a(q)\phi_i'(\E_q[x])$. A common shift in the payoff changes these relative weights whenever ambiguity attitudes vary with payoff levels. Since the agent evaluates wages and the principal evaluates residual output, a contract can change their model weights differently. Nonlinear compensation can then improve the allocation by changing model-specific expected payoffs asymmetrically across the two parties. Notice that this channel is absent under proportional variational preferences. There, $\Gamma^a(p)$ assigns a penalty to the candidate model $p$ that is independent of the payoff being evaluated. Changing the contract changes $\E_p[x]$, but not the cost attached to $p$.

\subsection{Vector expected utility}\label{app:veu}

We now relax Uncertainty Aversion while retaining translation invariance. A popular example is the Vector expected utility (VEU) of \citet{vec09}. It evaluates a payoff from its expectation under a baseline probability together with an adjustment for its exposure to distinct sources of ambiguity. To align the comparison with the preceding problems, let both parties be risk neutral in money, share the same full-support baseline probability $p$, and share the same adjustment factors, denoted $\zeta_1,\zeta_2$. For $i\in\{A,P\}$,
$$
\mathcal V_i^{\mathrm{VEU}}(x)
:=
\E_p[x]
+
A_i\big(\E_p[\zeta_1x],\E_p[\zeta_2x]\big).
$$
For the primitives, take $\Y=\{0,1,2\}$, $p=(1/3,1/3,1/3)$, $c(a)=0$, and, ordering coordinates by output, let $e_1:=\sqrt{3/2}(-1,0,1)$, $e_2:=(1,-2,1)/\sqrt2$, $\zeta_1:=(4e_1+3e_2)/5$, and $\zeta_2:=(3e_1-4e_2)/5$. Finally, let $g(t):=1-e^{-t^2}$, $\varepsilon:=1/10$, $A_A(\varphi_1,\varphi_2):=-\varepsilon[g(\varphi_1)+g(\varphi_2)]$ and $A_P(\varphi_1,\varphi_2):=-2\varepsilon[g(\varphi_1)+g(\varphi_2)]$. The resulting VEU first-best problem becomes $\max_w\mathcal V_P^{\mathrm{VEU}}(y-w)$ subject to $\mathcal V_A^{\mathrm{VEU}}(w)\ge U_0$.

\begin{obs}\label{prop:veu-nonlinear}
Under the VEU specification above, both parties' preferences are strictly monotone and ambiguity averse relative to EU in the comparative sense of \citet{vec09}. They have the same baseline prior and adjustment factors, and their adjustment functions are proportional. However, for every outside option $U_0\in\R$, the VEU first-best problem has a unique solution, and that contract is nonlinear in output.
\end{obs}

The VEU specification above is translation invariant, as shown in the proof of Observation \ref{prop:veu-nonlinear}, but it violates the Uncertainty Aversion axiom \citep[][Corollary 2]{vec09}. Thus, translation invariance alone does not place preferences in the variational class. Here, the ambiguity adjustment is nonlinear across distinct ambiguity factors, so the optimal exposure to one factor need not equal the optimal exposure to another. A single output slope cannot generally satisfy both margins. The unique optimum can therefore be nonlinear even though both parties share the same baseline probability and adjustment factors, and their adjustment functions differ only by scale.

\section{Appendix: Regular Solutions}\label{app:regularity}

Section \ref{sec:moralhazard} uses the KKT conditions to characterize solutions of $(\mathsf{MH}_{a^*})$. We take regularity to mean the Mangasarian-Fromovitz constraint qualification (MFCQ) \citep{mangasarian67}: at a feasible contract, there must exist a wage perturbation that strictly relaxes every binding constraint to first order. When the relevant effective models are unique, this condition is transparent because the gradient of $V_A^a(w)$ is the effective model $\pi_A^a(\cdot;w)$. For a feasible contract $w$, let $B(w):=\{a\neq a^*:V_A^{a^*}(w)=V_A^a(w)\}$.

\begin{lemma}\label{lem:regularity-characterization}
Suppose participation binds at a smooth feasible contract $w$. Then, $w$ is regular if and only if $\pi_A^{a^*}(\cdot;w)\notin\co\{\pi_A^a(\cdot;w):a\in B(w)\}$, which holds when $B(w)=\varnothing$.
\end{lemma}

The characterization depends only on effective models and therefore applies to any smooth variational specification. Regularity fails precisely when the desired-action marginal probability model can be reproduced as a convex combination of the marginal models associated with binding deviations. In that case, no first-order wage perturbation can move strictly away from every active deviation at once.

We next specialize the characterization to a subclass of variational preferences. Following \citet[eq. (6)]{hansenmiss25}, suppose that, for every action $a$, the ambiguity index is the $\phi$-divergence from the full-support reference model $q^a$, so that $\Gamma^a(p):=\sum_{y\in\Y}q^a(y)\phi(p(y)/q^a(y))$. Here, $\phi:[0,\infty)\to[0,\infty)$ is continuous and strictly convex, satisfies $\phi(1)=0$ with $\phi(t)/t\to\infty$ as $t\to\infty$. This class includes relative entropy, obtained from $\phi(t)=t\log t-t+1$, and the Gini index used in \citet[eq. (21)--(22)]{hansenmiss25}, obtained from $\phi(t)=(t-1)^2/2$. For the sufficient condition below, suppose also that $\phi$ is twice continuously differentiable on $(0,\infty)$, $\phi''>0$, and the inverse marginal penalty $\ell_\phi:=(\phi')^{-1}$ is log concave on $\phi'((0,\infty))$. Both relative entropy and the Gini index satisfy this condition: respectively, $\ell_\phi(z)=e^z$ and $\ell_\phi(z)=1+z$.

\begin{proposition}\label{prop:regularity-sufficient}
Let the variational indices be generated by the same $\phi$-divergence above. Then, every feasible contract is smooth. Let $w$ be feasible with binding participation. 
\begin{itemize}
    \item If $B(w)=\varnothing$, then $w$ is regular.
    \item If $B(w)\neq\varnothing$, suppose $w$ is also nondecreasing in output, that $y\mapsto q^a(y)/q^{a^*}(y)$ is weakly decreasing and nonconstant for every $a\in B(w)$, and that the agent's effective models under $a^*$ and every $a\in B(w)$ have full support. Then, $w$ is regular.
\end{itemize}

\end{proposition}

Proposition \ref{prop:regularity-sufficient} gives familiar primitive conditions for the regularity imposed in Theorem \ref{thm:MH-general}. The likelihood-ratio condition is the same ordering used in Section \ref{sec:moralhazard}: the desired action shifts probability toward higher outputs relative to every binding deviation. Under a $\phi$-divergence whose inverse marginal penalty is log concave, a common nondecreasing wage tilt preserves this ordering in the parties' effective models. The desired-action effective model therefore has strictly higher expected output than every binding-deviation effective model and cannot lie in their convex hull.

For relative entropy, the full-support condition in Proposition \ref{prop:regularity-sufficient} is automatic for every finite wage schedule, so monotonicity and the standard strict monotone-likelihood-ratio property (MLRP) ordering are sufficient for regularity. The result also applies to the Gini divergence. In that case, the effective model under action $a$ is interior exactly when $1-\lambda_A\big(w(y)-\E_{q^a}w\big)>0$ for every output $y$; when $w$ is nondecreasing, this is equivalent to $\lambda_A\big(\max_yw(y)-\E_{q^a}w\big)<1$. Thus, the same regularity argument applies whenever the Gini distortion remains interior.

These conditions are distinct from assumptions used to justify the first-order approach with continuous effort. \citet[Section 5.2.3]{laffontmartimort02} discuss MLRP together with convexity of the distribution function as sufficient conditions for replacing the agent's global optimization problem by its local first-order condition. Those assumptions must guarantee global incentive compatibility from a local condition. Here, the action set is finite and every incentive constraint is retained; regularity is needed only to ensure valid KKT multipliers at the optimum.

\section{Appendix: General Divergences}\label{app:divergences}

\subsection{First-best uniqueness}
This appendix verifies Corollary \ref{cor:firstbest-unique} for standard divergence specifications and formalizes their behavioral implications. Let $q^a$ denote the full-support reference model associated with action $a$. We will work with two classes of divergences, including the $\phi$-divergence already defined in Appendix \ref{app:regularity}; here we retain only its defining assumptions there, not the additional conditions imposed for Proposition \ref{prop:regularity-sufficient}.

For $\varrho\in(0,1)$, the $\varrho$-divergence between $p$ and $q$ \citep[][eq. (4)]{zhang06} is
\[
D_\varrho(p\Vert q):=
\frac{1}{\varrho(1-\varrho)}
\left[
1-\sum_{y\in\Y}p(y)^{1-\varrho}q(y)^\varrho
\right].
\]
The case $\varrho=1/2$ is the squared Hellinger distance up to normalization, while relative entropy is obtained as $\varrho\downarrow0$. Recall from Section \ref{subsec:model-libraries} that, for a divergence $D$ and $\lambda>0$,
\[
\mathcal W_\lambda(q,x)
:=
\min_{p\in\DeltaY}
\left\{
\E_p[x]+\frac{1}{\lambda}D(p\Vert q)
\right\}.
\]

\begin{definition}\label{def:bet-consistency}
Following \citet[Definition 7]{hansenmiss25},\footnote{This definition is the special case of \citet[Definition 7]{hansenmiss25} where the set of ``structured models'' is a singleton. It also uses risk neutrality to simplify the expressions on the right-hand side of the implication.} $\mathcal W_\lambda(q,\cdot)$ is \textit{bet-consistent} if, for every $\bar x>\underline x$ and events $A,B\subseteq\Y$,
\[
q(A)\ge q(B)
\quad\Longrightarrow\quad
\mathcal W_\lambda
\big(q,\bar x\mathds{1}_A+\underline x\mathds{1}_{A^c}\big)
\ge
\mathcal W_\lambda
\big(q,\bar x\mathds{1}_B+\underline x\mathds{1}_{B^c}\big).
\]
\end{definition}
 \citet{hansenmiss25} refer to preferences that satisfy bet consistency as capturing a mild form of misspecification aversion. The idea is that, although such preferences reflect concerns that the reference model is misspecified, they still give it a privileged role in the sense that it guides the ranking of bets. \citet[][Section 4.3.4]{hansenmiss25} interpret this property as capturing the idea that the decision marker finds the reference model to be a useful approximation of the true data-generating process.
\begin{proposition}\label{prop:divergence-classes}
The following properties hold.
\begin{enumerate}
    \item[(i)] Every $\varrho$-divergence is bet-consistent. Moreover, for every $\lambda>0$ and payoff $x\in\R^\Y$, $\Pi_\lambda^a(x)$ is a singleton contained in $\operatorname{ri}\DeltaY$, and $D_\varrho(\cdot\Vert q^a)$ is differentiable at its element. Thus, Corollary \ref{cor:firstbest-unique} applies to every $\varrho$-divergence.
    \item[(ii)] Every $\phi$-divergence defined in Appendix \ref{app:regularity} is bet-consistent. If, in addition, $\phi$ is differentiable on $(0,\infty)$ and $\phi'(0^+)=-\infty$, then, for every $\lambda>0$ and payoff $x\in\R^\Y$, $\Pi_\lambda^a(x)$ is a singleton contained in $\operatorname{ri}\DeltaY$ and $\Gamma^a$ is differentiable at its element. Thus, Corollary \ref{cor:firstbest-unique} applies under this additional condition.
\end{enumerate}
\end{proposition}

The takeaway from this result is that the linear first-best contract in Theorem \ref{thm:firstbest-linear} is typically unique when the parties exhibit a mild form of misspecification aversion.  

\subsection{Moral hazard: incentive corrections}\label{app:rho-incentives}

Corollary \ref{thm:MH-contract} decomposes the multiplier contract into the first-best linear sharing rule and an incentive correction. The same decomposition holds for every $\varrho$-divergence, although the correction is now expressed through the wedge between the parties' effective models.

\begin{corollary}\label{prop:rho-MH-correction}
Suppose $\Gamma^a(p)=D_\varrho(p\Vert q^a)$ for every action $a$, for some common $\varrho\in(0,1)$. Let $w^*$ be a regular solution of $(\mathsf{MH}_{a^*})$. Then, there exist nonnegative multipliers $\{\beta^a\}_{a\ne a^*}$, with $\beta^a=0$ for every inactive incentive constraint, a constant $C\in\R$, and an intercept $\alpha\in\R$ such that
\[
w^*(y)
=
\alpha
+\frac{\lambda_P}{\lambda_A+\lambda_P}y
+
\left(
\frac{R(y)^\varrho}
{\lambda_A/\lambda_P+R(y)^\varrho}
-
\frac{\lambda_P}{\lambda_A+\lambda_P}
\right)(y+C)
\qquad \forall y\in\Y,
\]
where
\[
R(y)
:=
1+\sum_{a\ne a^*}\beta^a
\Big(
1-\frac{\pi_A^a(y;w^*)}
{\pi_A^{a^*}(y;w^*)}
\Big)
=
\frac{\pi_P^{a^*}(y;w^*)}
{\pi_A^{a^*}(y;w^*)}
>0.
\]
For any such representation, $\alpha$ is uniquely determined by binding participation. If no incentive constraint is active, then $R\equiv1$ and the correction term vanishes.
\end{corollary}

Corollary \ref{prop:rho-MH-correction} generalizes Corollary \ref{thm:MH-contract} by showing that under every $\varrho$-divergence, moral hazard adds only a nonlinear incentive correction to the linear sharing rule. This correction term also depends on the effective model likelihood ratios of deviations, making less attractive outputs that are relatively more likely under deviations and doing the opposite for those that are more likely under the desired action.

\section{Appendix: Omitted Proofs}\label{app:oproofs}
\subsection{EU derivation in Section \ref{subsec:twooutcome}}\label{app:EU}
The EU comparison used in Section \ref{subsec:twooutcome} follows from a direct expansion. For clarity, write $\Psi_\lambda(p,b):=-\lambda^{-1}\log(1-p+pe^{-\lambda b})$ for $\lambda>0$ and define its continuous extension at $\lambda=0$ by $\Psi_0(p,b):=pb$. Taylor expansion of the exponential gives $1-p+pe^{-\lambda b}=1-p\lambda b+(p/2)\lambda^2b^2+O(\lambda^3)$. Expanding the logarithm around one and simplifying yields $\Psi_\lambda(p,b)=pb-(\lambda/2)p(1-p)b^2+O(\lambda^2)$, locally uniformly for $b$ in a bounded set.

Therefore, $D_\lambda(b):=\Psi_\lambda(p_H,b)-\Psi_\lambda(p_L,b)$ satisfies $D_\lambda(b)=(p_H-p_L)b-(\lambda/2)(p_H-p_L)(1-p_H-p_L)b^2+O(\lambda^2)$. At $\lambda=0$, the incentive equation $D_0(b)=\Delta c_{HL}$ has the unique solution $b^{EU}=\Delta c_{HL}/(p_H-p_L)$, and $\partial_bD_0(b^{EU})=p_H-p_L>0$. The implicit-function theorem therefore gives a unique smooth solution $b^{IC}(\lambda)$ near $\lambda=0$ with $b^{IC}(0)=b^{EU}$. Write $b^{IC}(\lambda)=b^{EU}+\lambda\beta+O(\lambda^2)$ and substitute this expression into the preceding expansion of $D_\lambda$. The zero-order term is $\Delta c_{HL}$ by definition of $b^{EU}$, and the coefficient on $\lambda$ must vanish. This gives $(p_H-p_L)\beta-(1/2)(p_H-p_L)(1-p_H-p_L)(b^{EU})^2=0$, so $\beta=(1/2)(1-p_H-p_L)(b^{EU})^2$. This is the expansion stated in the text.

Finally, $\overline C_{HL}=\lambda_A^{-1}\log[(1-p_L)/(1-p_H)]$. The logarithm is a strictly positive constant because $p_H>p_L$, so $\overline C_{HL}$ is strictly decreasing in $\lambda_A$. Hence, stronger misspecification concern tightens the global implementability bound, as stated.

\subsection{Proofs for Appendices \ref{app:general-criteria}--\ref{app:divergences}}\label{app:general-criteria-proofs}

\begin{proof}[Proof of Observation \ref{prop:smooth-nonlinear}]
We verify in turn that the specified preferences satisfy the stated properties, that the constructed contract is the unique first-best optimum, and that this optimum is nonlinear. The two outer aggregators are $\phi_P(t)=-e^{-t}$ and $\phi_A(t)=-e^{-t}-(1/2)e^{-2t}$. Their derivatives satisfy $\phi_P'(t)=e^{-t}>0$, $\phi_P''(t)=-e^{-t}<0$, $\phi_A'(t)=e^{-t}+e^{-2t}>0$, and $\phi_A''(t)=-e^{-t}-2e^{-2t}<0$ for every $t$. Thus, both aggregators are strictly increasing and strictly concave. With linear utility over money, the two parties are therefore risk neutral within each model and strictly ambiguity averse across models in the smooth-ambiguity sense.

Let $Q$ be the $3\times3$ matrix whose rows are the three reference models in the statement. Explicitly, $Q=0.7I+0.1\m{1}\m{1}^{\top}$. The vector $\m{1}$ is an eigenvector with eigenvalue one, and every vector orthogonal to $\m{1}$ is an eigenvector with eigenvalue $0.7$. Hence, all eigenvalues are strictly positive and $Q$ is invertible. Write the aggregate output vector as $y=(0,1,2)^{\top}$ and set $m:=Qy=(0.3,1,1.7)^{\top}$. For any wage vector $w$, define $z:=Qw$. Then, $z_j=\E_{q^j}[w]$, while $m_j-z_j=\E_{q^j}[y-w]$. Since $Q$ is invertible, choosing $w\in\R^3$ is equivalent to choosing $z\in\R^3$.

For a scalar $m\in\R$, define $z(m)$ as the solution of $e^{-(m-z)}=e^{-z}+e^{-2z}$. We first show that this solution exists and is unique. Set $t:=e^z>0$. Multiplying the defining equality by $t^2$ gives $e^{-m}t^3=t+1$. Let $H_m(t):=e^{-m}t^3-t-1$. We have $H_m(0)=-1$, and $H_m(t)\to+\infty$ as $t\to\infty$. Its derivative is $H_m'(t)=3e^{-m}t^2-1$, which vanishes at exactly one positive point $\widehat t=(e^m/3)^{1/2}$. The derivative is negative below $\widehat t$ and positive above it, so $H_m$ decreases and then increases. At the unique critical point, the identity $3e^{-m}\widehat t^2=1$ gives $e^{-m}\widehat t^3=\widehat t/3$, and therefore $H_m(\widehat t)=\widehat t/3-\widehat t-1=-2\widehat t/3-1<0$. Since the function is negative at its unique minimum and tends to $+\infty$, it crosses zero exactly once on $(\widehat t,\infty)$. Hence, there is a unique $t>0$, so a unique $z(m)\in\R$ satisfying the equation. Set $z_j:=z(m_j)$ for $j=1,2,3$. Because $\phi_A$ is continuous, strictly increasing, tends to $-\infty$ as $t\to-\infty$, and tends to zero from below as $t\to+\infty$, it maps $\R$ bijectively onto $(-\infty,0)$. The number $\frac{1}{3}\sum_j\phi_A(z_j)$ belongs to $(-\infty,0)$, so there is a unique outside option $U_0$ such that $\phi_A(U_0)=\frac{1}{3}\sum_j\phi_A(z_j)$. Fix this $U_0$ and let $z^*:=(z_1,z_2,z_3)$ and $w^*:=Q^{-1}z^*$. Since both outer aggregators are strictly increasing, maximizing the principal's smooth-ambiguity valuation is equivalent to maximizing its argument $F(z):=\frac{1}{3}\sum_j\phi_P(m_j-z_j)$, and the agent's participation constraint is equivalent to $G(z):=\frac{1}{3}\sum_j\phi_A(z_j)\ge\phi_A(U_0)$. By construction, $G(z^*)=\phi_A(U_0)$, so participation binds at $z^*$.

We next verify first-order optimality. For each coordinate, $\partial F/\partial z_j=-(1/3)\phi_P'(m_j-z_j)=-(1/3)e^{-(m_j-z_j)}$, while $\partial G/\partial z_j=(1/3)\phi_A'(z_j)=(1/3)(e^{-z_j}+e^{-2z_j})$. The defining equation for $z_j$ says these two quantities sum to zero. Hence, $\nabla F(z^*)+\nabla G(z^*)=0$. These conditions are sufficient for global optimality, and we prove this directly. Let $z'$ be any feasible vector. Concavity of $G$ gives $G(z')\le G(z^*)+\nabla G(z^*)\cdot(z'-z^*)$. Since both $z'$ and $z^*$ are feasible and $G(z^*)$ equals the constraint threshold, $G(z')\ge G(z^*)$. Combining the two inequalities yields $\nabla G(z^*)\cdot(z'-z^*)\ge0$. Concavity of $F$ gives $F(z')\le F(z^*)+\nabla F(z^*)\cdot(z'-z^*)$. Using $\nabla F(z^*)=-\nabla G(z^*)$ then implies $F(z')\le F(z^*)$. Thus, $z^*$ is globally optimal. Moreover, $F$ is strictly concave because each map $z_j\mapsto-e^{-(m_j-z_j)}=-e^{-m_j}e^{z_j}$ is strictly concave. Therefore, if $z'\ne z^*$, the concavity inequality for $F$ is strict along the segment connecting the two points, and $z'$ cannot also be optimal. Hence, $z^*$ is the unique optimum. Since $Q$ is invertible, $w^*=Q^{-1}z^*$ is the unique optimal wage vector.

It remains to prove nonlinearity. The equation defining $z(m)$ can be rewritten by taking logarithms as $m=3z-\log(e^z+1)$. Define $L(z):=3z-\log(e^z+1)$. Then, $L'(z)=3-e^z/(e^z+1)>2>0$ and $L''(z)=-e^z/(e^z+1)^2<0$. Thus, $L$ is strictly increasing and strictly concave. Its inverse $z(m)$ is therefore strictly convex. For completeness, differentiating the identity $L(z(m))=m$ twice gives $z''(m)=-L''(z(m))/[L'(z(m))]^3>0$. Since $m_2=1$ is the midpoint of $m_1=0.3$ and $m_3=1.7$, strict convexity implies $z_1+z_3>2z_2$. If the unique wage vector were affine in output, $w^*=\alpha\m{1}+\beta y$ for some scalars $\alpha,\beta$. Every row of $Q$ sums to one, so $Q\m{1}=\m{1}$. Hence, $z^*=Qw^*=\alpha\m{1}+\beta Qy=\alpha\m{1}+\beta m$. Because the entries of $m$ are equally spaced, this affine relation would imply $z_1+z_3=2z_2$, contradicting the strict inequality established above. Therefore, the unique first-best contract is nonlinear.
\end{proof}

\begin{proof}[Proof of Observation \ref{prop:veu-nonlinear}]
We first verify the properties of the specified VEU representation. Equip $\R^3$ with the inner product $\langle x,z\rangle_p:=\E_p[xz]=(1/3)\sum_yx(y)z(y)$. The vectors $e_1=\sqrt{3/2}(-1,0,1)$ and $e_2=(1,-2,1)/\sqrt2$ both have zero mean under $p=(1/3,1/3,1/3)$. Direct calculation gives $\langle e_1,e_1\rangle_p=1$, $\langle e_2,e_2\rangle_p=1$, and $\langle e_1,e_2\rangle_p=0$. Thus, $e_1,e_2$ form an orthonormal basis for the two-dimensional mean-zero subspace.

The transformation from $(e_1,e_2)$ to $(\zeta_1,\zeta_2)$ is given by the matrix $(1/5)\left(\begin{smallmatrix}4&3\\3&-4\end{smallmatrix}\right)$. Its rows are orthonormal, so it is an orthogonal matrix. Hence, $\zeta_1$ and $\zeta_2$ also have zero mean and are orthonormal under $p$. Orthogonality additionally preserves the sum of squared coordinates state by state: for every output state $y$, $\zeta_1(y)^2+\zeta_2(y)^2=e_1(y)^2+e_2(y)^2=2$.

The adjustment function $g(t)=1-e^{-t^2}$ is continuous, even, nonnegative, and satisfies $g(0)=0$. Therefore, each $A_i$ is continuous, symmetric, and satisfies $A_i(0,0)=0$, as required for the VEU form. We now verify strict monotonicity of the resulting preferences. Let $\varrho_A:=1$ and $\varrho_P:=2$, so $A_i(\varphi_1,\varphi_2)=-\varrho_i\varepsilon[g(\varphi_1)+g(\varphi_2)]$ with $\varepsilon=1/10$. For a payoff vector $x$, write $\varphi_k(x):=\E_p[\zeta_kx]$. Differentiating the VEU value with respect to the payoff in state $y$ gives $(1/3)\{1-\varrho_i\varepsilon[g'(\varphi_1(x))\zeta_1(y)+g'(\varphi_2(x))\zeta_2(y)]\}$.

Now $g'(t)=2te^{-t^2}$, whose absolute value is maximized at $|t|=1/\sqrt2$ and has supremum $\sqrt{2/e}$. By Cauchy--Schwarz and the identity $\zeta_1(y)^2+\zeta_2(y)^2=2$, $|\zeta_1(y)|+|\zeta_2(y)|\le\sqrt2\sqrt{\zeta_1(y)^2+\zeta_2(y)^2}=2$. Therefore, the absolute value of the adjustment inside braces is at most $2\varrho_i\varepsilon\sqrt{2/e}$. Since $\varrho_i\le2$ and $\varepsilon=1/10$, this upper bound is at most $(2/5)\sqrt{2/e}<1$. Every state derivative of each party's VEU value is therefore strictly positive. Hence, both preferences are strictly monotone.

Because $g\ge0$, both adjustment functions are everywhere nonpositive, and $A_P\le A_A\le0$. By \citet[][Proposition 4]{vec09}, this means that both preferences are ambiguity averse relative to the EU benchmark in the comparative sense used there, with the principal more ambiguity averse than the agent. The preferences are nevertheless not variational. Indeed, $g''(t)=2e^{-t^2}(1-2t^2)$ is negative whenever $|t|>1/\sqrt2$, so $-g$ is locally strictly convex on such regions. Thus, $A_i$ is not concave. By \citet[][Corollary 2]{vec09}, a VEU preference satisfies the Uncertainty Aversion axiom exactly under the corresponding concavity restriction on its adjustment, so the specified preferences violate that axiom and fall outside the variational class.

We now solve the first-best problem. Since each $\zeta_k$ has zero mean, adding a constant $r$ to every payoff changes $\E_p[x]$ by $r$ but leaves $\E_p[\zeta_kx]$ unchanged. Hence, both VEU functionals are translation invariant: $\mathcal V_i^{\mathrm{VEU}}(x+r\m{1})=\mathcal V_i^{\mathrm{VEU}}(x)+r$. Exactly as in Lemma \ref{lem:FB-participation}, this implies that the agent's participation constraint binds at every first-best optimum.

For an arbitrary wage vector $w$, define $t:=\E_p[w]$ and $v_k:=\E_p[\zeta_kw]$ for $k=1,2$. The three vectors $\m{1},\zeta_1,\zeta_2$ form an orthonormal basis of $\R^3$ under $\langle\cdot,\cdot\rangle_p$, so the triple $(t,v_1,v_2)$ uniquely determines $w$. Let $c_k:=\E_p[\zeta_ky]$ denote the ambiguity exposure of aggregate output. Direct calculation gives $\E_p[e_1y]=\sqrt{2/3}$ and $\E_p[e_2y]=0$. Hence, $c_1=(4/5)\sqrt{2/3}$ and $c_2=(3/5)\sqrt{2/3}$. In particular, $0<c_2<c_1<1/\sqrt2$.

Binding participation gives $t+A_A(v_1,v_2)=U_0$, so $t=U_0-A_A(v_1,v_2)$. Since $\E_p[y]=1$, the principal's residual claim has mean $1-t$ and ambiguity exposures $c_k-v_k$. Substituting the binding value of $t$ into the principal's VEU value shows that maximizing the principal's payoff is equivalent to minimizing, separately across $k=1,2$, the functions $h_{c_k}(v):=g(v)+2g(c_k-v)$. We now solve this one-dimensional problem for an arbitrary $c\in(0,1/\sqrt2)$.

First, every global minimizer of $h_c$ lies in $[0,c]$. If $v<0$, then $g(v)>0$ and $c-v>c>0$. Since $g$ is strictly increasing on $(0,\infty)$, $g(c-v)>g(c)$, so $h_c(v)>2g(c)=h_c(0)$. If $v>c$, then $g(v)>g(c)$ and $g(c-v)>0$, so $h_c(v)>g(c)=h_c(c)$. Thus, no point outside $[0,c]$ can minimize.

On $[0,c]$, both $v$ and $c-v$ belong to $[0,c]\subset[0,1/\sqrt2)$, where $g''$ is strictly positive. Hence, $h_c''(v)=g''(v)+2g''(c-v)>0$ throughout the interval, so $h_c$ is strictly convex there. Moreover, $h_c'(0)=g'(0)-2g'(c)=-2g'(c)<0$, while $h_c'(c)=g'(c)-2g'(0)=g'(c)>0$. By continuity and strict increase of $h_c'$, there is a unique $v(c)\in(0,c)$ satisfying $h_c'(v(c))=0$, and this point is the unique global minimizer.

The first-order condition is $g'(v(c))=2g'(c-v(c))$. Write $v(c)=\theta(c)c$. Since $h_c'(c/2)=g'(c/2)-2g'(c/2)<0$ and the unique root of $h_c'$ lies to the right of any point at which the derivative is negative, $\theta(c)>1/2$. Substituting $g'(t)=2te^{-t^2}$ into the first-order condition and simplifying gives $\log\{\theta(c)/[2(1-\theta(c))]\}=[2\theta(c)-1]c^2$.

Define $F(\theta,c):=\log\{\theta/[2(1-\theta)]\}-(2\theta-1)c^2$. At the solution, $F(\theta(c),c)=0$. Its partial derivative with respect to $\theta$ is $1/[\theta(1-\theta)]-2c^2$. Because $\theta\in(0,1)$ gives $1/[\theta(1-\theta)]\ge4$ and $c<1/\sqrt2$ gives $2c^2<1$, this derivative is strictly larger than three. Its partial derivative with respect to $c$ is $-2c(2\theta-1)<0$ because $c>0$ and $\theta>1/2$. The implicit-function theorem therefore applies and gives $\theta'(c)=-F_c/F_\theta>0$. Since $c_1>c_2$, we conclude $v(c_1)/c_1=\theta(c_1)>\theta(c_2)=v(c_2)/c_2$.

Let $(v_1^*,v_2^*)$ be the unique pair of minimizing ambiguity exposures. Suppose the resulting optimal wage were affine in output, $w^*=\alpha\m{1}+\beta y$. Since each $\zeta_k$ has zero mean, $v_k^*=\E_p[\zeta_kw^*]=\beta\E_p[\zeta_ky]=\beta c_k$. Hence, an affine contract requires $v_1^*/c_1=v_2^*/c_2=\beta$, contradicting the strict inequality established above. Thus, the optimal contract is nonlinear.

Finally, each one-dimensional problem $h_{c_k}$ has a unique minimizer, so $(v_1^*,v_2^*)$ is unique. Binding participation uniquely determines $t^*$. Since $\m{1},\zeta_1,\zeta_2$ form a basis, the triple $(t^*,v_1^*,v_2^*)$ determines a unique wage vector. Therefore, for every outside option $U_0$, the first-best solution exists, is unique, and is nonlinear.
\end{proof}

\begin{proof}[Proof of Lemma \ref{lem:regularity-characterization}]
At a smooth feasible contract $w$ with binding participation, Appendix \ref{app:prelim} implies that the gradient of participation is $\pi_A^{a^*}(\cdot;w)$ and the gradient of the incentive constraint against an active deviation $a\in B(w)$ is $\pi_A^{a^*}(\cdot;w)-\pi_A^a(\cdot;w)$. Since all constraints are written as weak inequalities that must be nonnegative, MFCQ requires a direction $h\in\R^\Y$ such that $\pi_A^{a^*}\cdot h>0$ and $(\pi_A^{a^*}-\pi_A^a)\cdot h>0$ for every $a\in B(w)$.

Suppose first that $\pi_A^{a^*}$ belongs to the convex hull of the active deviation models. Then, there are weights $\gamma_a\ge0$, summing to one over $a\in B(w)$, such that $\pi_A^{a^*}=\sum_{a\in B(w)}\gamma_a\pi_A^a$. If a direction $h$ strictly relaxed every active incentive constraint, then $(\pi_A^{a^*}-\pi_A^a)\cdot h>0$ for every $a\in B(w)$. Multiplying each inequality by $\gamma_a$ and summing yields $(\pi_A^{a^*}-\sum_{a\in B(w)}\gamma_a\pi_A^a)\cdot h>0$. The vector in parentheses is zero by the convex-combination identity, which gives the contradiction $0>0$. Hence, no direction can satisfy MFCQ, so $w$ is not regular.

Conversely, suppose $\pi_A^{a^*}$ does not belong to the convex hull of the active deviation models. If $B(w)=\varnothing$, take $h=\m{1}$. Since $\pi_A^{a^*}$ is a probability distribution, $\pi_A^{a^*}\cdot\m{1}=1>0$, and there are no binding incentive constraints. Thus, MFCQ holds.

Now suppose $B(w)\neq\varnothing$. The convex hull of the finitely many vectors $\{\pi_A^a:a\in B(w)\}$ is compact and convex. Since $\pi_A^{a^*}$ lies outside it, the strict hyperplane separation theorem gives a vector $h_0\in\R^\Y$ such that $(\pi_A^{a^*}-\pi_A^a)\cdot h_0>0$ for every $a\in B(w)$. This direction need not yet strictly relax participation. Choose a scalar $K>-\pi_A^{a^*}\cdot h_0$ and set $h:=h_0+K\m{1}$. Because all effective models are probability distributions, $(\pi_A^{a^*}-\pi_A^a)\cdot\m{1}=0$ for every deviation. Hence, $(\pi_A^{a^*}-\pi_A^a)\cdot h=(\pi_A^{a^*}-\pi_A^a)\cdot h_0>0$ for every active deviation. Moreover, $\pi_A^{a^*}\cdot h=\pi_A^{a^*}\cdot h_0+K>0$. Thus, the direction $h$ strictly relaxes participation and every binding incentive constraint. MFCQ holds, which proves the characterization.
\end{proof}

\begin{proof}[Proof of Proposition \ref{prop:regularity-sufficient}]
We first prove smoothness. Fix an action $a$. Because $q^a$ has full support, every $p\in\DeltaY$ is absolutely continuous with respect to $q^a$, and therefore $\Gamma^a(p)=\sum_yq^a(y)\phi(p(y)/q^a(y))$ is finite and continuous on the compact simplex. Since $\phi$ is strictly convex and every $q^a(y)$ is strictly positive, $\Gamma^a$ is strictly convex in $p$. Indeed, if $p\neq r$ and $t\in(0,1)$, then $p(y)/q^a(y)\neq r(y)/q^a(y)$ for at least one output $y$; strict convexity of $\phi$ at that coordinate and convexity at every other coordinate imply $\Gamma^a\big(tp+(1-t)r\big)<t\Gamma^a(p)+(1-t)\Gamma^a(r)$.

The ambiguity index has the unique minimizer $q^a$. To see this, $\Gamma^a(q^a)=0$ because $\phi(1)=0$. Since $\phi$ is nonnegative, $\Gamma^a(p)\ge0$ for every $p$. Moreover, strict convexity and nonnegativity of $\phi$ imply that $1$ is its unique zero: if $\phi(t)=0$ for some $t\neq1$, strict convexity would make $\phi$ strictly negative at every nontrivial convex combination of $1$ and $t$, contradicting nonnegativity. Hence, $\Gamma^a(p)=0$ requires $p(y)/q^a(y)=1$ for every $y$, so $p=q^a$. 

For every $\lambda>0$ and finite payoff vector $x$, the function $p\mapsto\E_p[x]+\Gamma^a(p)/\lambda$ is the sum of an affine function and a strictly convex function. It is therefore strictly convex. Since it is continuous on the compact simplex, it has a unique minimizer. Thus, every effective-model set $\Pi_\lambda^a(x)$ is a singleton.  Consequently, for every feasible contract, all effective models entering the principal's objective, participation, and the incentive constraints are unique. Appendix \ref{app:prelim} then implies differentiability of the relevant value functions. Every feasible contract is therefore smooth.

If $B(w)=\varnothing$, regularity follows from Lemma \ref{lem:regularity-characterization}. Suppose that $B(w)\neq\varnothing$ and that the additional hypotheses of the result hold. Write $p^*:=\pi_A^{a^*}(\cdot;w)$ and, for each $a\in B(w)$, write $p^a:=\pi_A^a(\cdot;w)$. These models have full support by assumption.

We next characterize each effective model. Fix an action $a\in B(w)\cup\{a^*\}$. Since $p^a$ is an interior minimizer of $p\mapsto\E_p[w]+\Gamma^a(p)/\lambda_A$, the first-order conditions for this minimization subject to $\sum_yp(y)=1$ are necessary. Hence, there is a scalar $\nu_a$ such that $w(y)+(1/\lambda_A)\phi'(p^a(y)/q^a(y))+\nu_a=0$ for every output $y$. Let $\kappa_a:=-\lambda_A\nu_a$. Since $\phi''>0$, the function $\phi'$ is strictly increasing and has inverse $\ell_\phi$. The first-order conditions therefore give $p^a(y)/q^a(y)=\ell_\phi(\kappa_a-\lambda_Aw(y))$ for every $y$. Normalization gives $\sum_yq^a(y)\ell_\phi(\kappa_a-\lambda_Aw(y))=1$. Because $\ell_\phi$ is strictly increasing, the left-hand side is strictly increasing in $\kappa_a$ wherever it is defined, so the normalizing scalar is unique.

Fix an active deviation $a$. Define $r(y):=q^{a^*}(y)/q^a(y)$ and $f(y):=\ell_\phi(\kappa_a-\lambda_Aw(y))$. The MLRP hypothesis says that $r$ is weakly increasing and nonconstant in output. Because $w$ is nondecreasing and $\ell_\phi$ is increasing, $f$ is weakly decreasing. Moreover, normalization of $p^a$ gives $\E_{q^a}[f]=1$, while $\E_{q^a}[r]=1$ because $r$ is the likelihood ratio of $q^{a^*}$ with respect to $q^a$.

The opposite monotonicity of $r$ and $f$ implies $\operatorname{Cov}_{q^a}(r,f)\le0$. For completeness, the standard pairwise covariance identity expresses this covariance as one half of the sum over all pairs of outputs $y,z$ of $q^a(y)q^a(z)[r(y)-r(z)][f(y)-f(z)]$. Each summand is nonpositive because an increase in output can only increase $r$ and decrease $f$. Hence, $\E_{q^a}[rf]\le\E_{q^a}[r]\E_{q^a}[f]=1$. Since $\E_{q^a}[rf]=\E_{q^{a^*}}[f]$, we obtain $\sum_yq^{a^*}(y)\ell_\phi(\kappa_a-\lambda_Aw(y))\le1$.

Let $\kappa_*:=\kappa_{a^*}$. Normalization of $p^*$ gives $\sum_yq^{a^*}(y)\ell_\phi(\kappa_*-\lambda_Aw(y))=1$. The same sum is strictly increasing in its scalar argument, while its value at $\kappa_a$ is weakly below one. Therefore, $\kappa_*\ge\kappa_a$. Set $d:=\kappa_*-\kappa_a\ge0$.

We now compare the two effective models state by state. Their likelihood ratio satisfies $p^*(y)/p^a(y)=r(y)\ell_\phi(\kappa_a+d-\lambda_Aw(y))/\ell_\phi(\kappa_a-\lambda_Aw(y))$. Let $g:=\log\ell_\phi$. By assumption, $g$ is concave. For every fixed $d\ge0$, the increment $g(x+d)-g(x)$ is weakly decreasing in $x$: differentiability gives derivative $g'(x+d)-g'(x)\le0$ because the derivative of a concave function is weakly decreasing. Hence, $\ell_\phi(x+d)/\ell_\phi(x)$ is weakly decreasing in $x$. Since $x(y):=\kappa_a-\lambda_Aw(y)$ is weakly decreasing in output, the ratio $\ell_\phi(x(y)+d)/\ell_\phi(x(y))$ is weakly increasing in output.

It follows that $p^*(y)/p^a(y)$ is the product of two strictly positive weakly increasing functions of output. The first factor $r(y)$ is nonconstant, so this product is itself weakly increasing and nonconstant. Thus, the desired-action effective model strictly dominates the deviation effective model in the monotone-likelihood-ratio order.

This ordering implies a strict ordering of expected output. Let $L(y):=p^*(y)/p^a(y)$. Since $p^a$ has full support, $\E_{p^a}[L]=1$, and $\E_{p^*}[y]-\E_{p^a}[y]=\operatorname{Cov}_{p^a}(y,L)$. The pairwise covariance identity writes the latter as one half of the sum over $y,z$ of $p^a(y)p^a(z)(y-z)[L(y)-L(z)]$. Every term is nonnegative because $L$ is weakly increasing in output. At least one term is strictly positive because $L$ is nonconstant, the outputs in $\Y$ are distinct, and $p^a$ has full support. Therefore, $\E_{p^*}[y]>\E_{p^a}[y]$ for every $a\in B(w)$.

Suppose, toward a contradiction, that $p^*$ belonged to the convex hull of the active deviation models. Then, there would be nonnegative weights $\gamma_a$, summing to one, such that $p^*=\sum_{a\in B(w)}\gamma_ap^a$. Taking expected output would give $\E_{p^*}[y]=\sum_{a\in B(w)}\gamma_a\E_{p^a}[y]$. But every term on the right is strictly smaller than $\E_{p^*}[y]$, so their convex combination is strictly smaller as well, a contradiction. Thus, $p^*\notin\co\{p^a:a\in B(w)\}$. Lemma \ref{lem:regularity-characterization} implies that $w$ is regular.\end{proof}

\begin{proof}[Proof of Proposition \ref{prop:divergence-classes}]
We first prove bet consistency. Both classes can be written as $D_g(p\Vert q)=\sum_yq(y)g(p(y)/q(y))$ for a convex generator $g$ on $[0,\infty)$. For the $\phi$-divergence, take $g=\phi$. For the $\varrho$-divergence, take $g=\psi_\varrho$, where $\psi_\varrho(t):=[\varrho+(1-\varrho)t-t^{1-\varrho}]/[\varrho(1-\varrho)]$; direct substitution gives $D_{\psi_\varrho}=D_\varrho$, and $\psi_\varrho''(t)=t^{-1-\varrho}>0$ for every $t>0$.

Fix $\lambda>0$, $\bar x>\underline x$, and a nonempty proper event $A$. Write $s:=q(A)\in(0,1)$. For any model $p$, set $z(y):=p(y)/q(y)$. The variational problem is the convex program that minimizes $\sum_yq(y)\{z(y)x(y)+g(z(y))/\lambda\}$ over $z(y)\ge0$ subject to $\sum_yq(y)z(y)=1$. Since $z\equiv1$ lies in the relative interior of its domain, standard finite-dimensional convex duality gives its value as $\sup_{\eta\in\R}\{\eta-\lambda^{-1}\sum_yq(y)g^*(\lambda[\eta-x(y)])\}$, where $g^*(u):=\sup_{t\ge0}\{ut-g(t)\}$ is the convex conjugate of $g$. The conjugate $g^*$ is nondecreasing: if $u'\ge u$, then $u't-g(t)\ge ut-g(t)$ for every $t\ge0$, so taking suprema gives $g^*(u')\ge g^*(u)$.

For the bet paying $\bar x$ on $A$ and $\underline x$ otherwise, the dual objective at a fixed $\eta$ is $\eta-\lambda^{-1}\{s g^*(\lambda[\eta-\bar x])+(1-s)g^*(\lambda[\eta-\underline x])\}$. Since $\bar x>\underline x$, monotonicity of $g^*$ gives $g^*(\lambda[\eta-\bar x])\le g^*(\lambda[\eta-\underline x])$. Hence, for every $\eta$ at which the dual objective is finite, that objective is nondecreasing in $s$. In the $\varrho$-divergence case, where the conjugate may take the value $+\infty$, the set of $\eta$ for which both conjugate terms are finite is independent of $s\in(0,1)$; outside that set the dual objective equals $-\infty$. Taking the supremum over $\eta$ therefore preserves monotonicity in $s$. Thus, the value of the bet is nondecreasing in $q(A)$. Empty and full events follow immediately from monotonicity of the variational criterion with respect to payoffs. Consequently, $q(A)\ge q(B)$ implies that the bet on $A$ is weakly preferred to the otherwise identical bet on $B$. Both divergence classes are therefore bet-consistent.

We next establish the effective-model properties. Consider first a generic divergence $D_g(p\Vert q^a)=\sum_yq^a(y)g(p(y)/q^a(y))$ for which $g$ is continuous and strictly convex on $[0,\infty)$, differentiable on $(0,\infty)$, and satisfies $g'(0^+)=-\infty$. Fix $\lambda>0$ and $x\in\R^\Y$. Continuity on the compact simplex implies that $p\mapsto\E_p[x]+D_g(p\Vert q^a)/\lambda$ attains a minimum. Strict convexity of $g$, together with $q^a(y)>0$ for every output, makes $D_g(\cdot\Vert q^a)$ strictly convex, so this minimizer is unique.

The minimizer must have full support. Suppose, to the contrary, that a minimizer $p^*$ satisfies $p^*(y_0)=0$ for some $y_0$. Since $p^*$ is a probability distribution, choose $y_1$ with $p^*(y_1)>0$ and, for small $\varepsilon>0$, transfer mass $\varepsilon$ from $y_1$ to $y_0$. The right directional derivative of the objective in this direction is $x(y_0)-x(y_1)+\lambda^{-1}\{g'(0^+)-g'(p^*(y_1)/q^a(y_1))\}=-\infty$, because $p^*(y_1)/q^a(y_1)>0$ and hence $g'(p^*(y_1)/q^a(y_1))$ is finite, whereas $g'(0^+)=-\infty$. Thus, every sufficiently small such transfer strictly lowers the objective, contradicting optimality. The unique minimizer therefore belongs to $\operatorname{ri}\DeltaY$. Since all its likelihood ratios are positive and $g$ is differentiable on $(0,\infty)$, $D_g(\cdot\Vert q^a)$ is differentiable there.

For a $\varrho$-divergence, the generator $\psi_\varrho$ above satisfies $\psi_\varrho'(t)=(1-t^{-\varrho})/\varrho$, so $\psi_\varrho'(0^+)=-\infty$, while $\psi_\varrho''(t)=t^{-1-\varrho}>0$. It therefore satisfies all the preceding conditions for every $\varrho\in(0,1)$. This proves the effective-model claim in part (i).

For a $\phi$-divergence, the defining assumptions in Appendix \ref{app:regularity} already give continuity and strict convexity. Under the additional assumptions in part (ii), $\phi$ is differentiable on $(0,\infty)$ and $\phi'(0^+)=-\infty$, so the same argument applies verbatim. This proves the effective-model claim in part (ii).

Finally,  the pooled concern $\lambda_{AP}$ is strictly positive. Applying the preceding conclusions with $\lambda=\lambda_{AP}$ and payoff $x=y$ gives an interior pooled effective model at which the ambiguity index is differentiable, as required by Corollary \ref{cor:firstbest-unique}. 
\end{proof}

\begin{proof}[Proof of Corollary \ref{prop:rho-MH-correction}]
Proposition \ref{prop:divergence-classes} implies that all effective models under the $\varrho$-divergence specification are unique and have full support. Since $w^*$ is regular, Theorem \ref{thm:MH-general} therefore yields nonnegative multipliers $\{\beta^a\}_{a\ne a^*}$, equal to zero for inactive incentive constraints, such that $\pi_P^{a^*}(y;w^*)=\pi_A^{a^*}(y;w^*)+\sum_{a\ne a^*}\beta^a[\pi_A^{a^*}(y;w^*)-\pi_A^a(y;w^*)]$ for every $y$. Dividing by the strictly positive probability $\pi_A^{a^*}(y;w^*)$ gives exactly the expression for $R(y)$ in the statement. Since the left-hand side of that ratio is $\pi_P^{a^*}(y;w^*)/\pi_A^{a^*}(y;w^*)$, full support also gives $R(y)>0$.

Write $p_A:=\pi_A^{a^*}(\cdot;w^*)$, $p_P:=\pi_P^{a^*}(\cdot;w^*)$, and $q:=q^{a^*}$. On $\operatorname{ri}\DeltaY$, differentiation of $D_\varrho(p\Vert q)$ with respect to $p(y)$ gives $-\varrho^{-1}[q(y)/p(y)]^\varrho$. The first-order conditions for the agent's effective-model problem therefore imply that there is an output-independent constant $\kappa_A$ such that $w^*(y)+\kappa_A=[\lambda_A\varrho]^{-1}[q(y)/p_A(y)]^\varrho$ for every $y$. The corresponding first-order conditions for the principal imply that there is an output-independent constant $\kappa_P$ such that $y-w^*(y)+\kappa_P=[\lambda_P\varrho]^{-1}[q(y)/p_P(y)]^\varrho$.

Because $p_P(y)=R(y)p_A(y)$, the principal's equality becomes $y-w^*(y)+\kappa_P=(\lambda_A/\lambda_P)R(y)^{-\varrho}[w^*(y)+\kappa_A]$. Rearranging shows that, for some output-independent constant $C$, $w^*(y)+\kappa_A=\{R(y)^\varrho/[\lambda_A/\lambda_P+R(y)^\varrho]\}(y+C)$ for every $y$. Adding and subtracting $\lambda_Py/(\lambda_A+\lambda_P)$ and absorbing the remaining output-independent term into an intercept gives precisely the contract displayed in the proposition.

It remains to pin down the intercept. Adding a constant to a contract leaves every effective model unchanged because variational preferences are translation invariant; it also leaves every incentive difference unchanged. Hence, once the output-dependent part of the representation is fixed, $\alpha$ changes only the level of the agent's payoff and does so one-for-one. Lemma \ref{lem:MH-binding} says that participation binds at every optimum, so exactly one value of $\alpha$ satisfies participation.

Finally, if no incentive constraint is active, every $\beta^a$ is zero and therefore $R(y)=1$ for every output. In that case, $R(y)^\varrho/[\lambda_A/\lambda_P+R(y)^\varrho]=1/[\lambda_A/\lambda_P+1]=\lambda_P/(\lambda_A+\lambda_P)$, so the correction term is identically zero. The contract consequently reduces to the first-best linear sharing rule.
\end{proof}
\end{document}